\documentclass[conference]{IEEEtran}
\IEEEoverridecommandlockouts
\usepackage{boldline}
\usepackage{booktabs}
\usepackage{cite}
\usepackage{amsmath,amssymb,amsfonts}
\usepackage{amsthm}
\usepackage{tcolorbox}
\usepackage{tabularx}
\usepackage{tipa}
\usepackage{makecell}
\usepackage{eurosym}
\usepackage{textcomp}
\usepackage{amsmath}
\usepackage{tikz}
\usepackage{caption}
\usepackage{eurosym}
\usepackage{subfig}
\tikzset{edge/.style={-stealth, thick}}
\usetikzlibrary{positioning, calc}

\usepackage[vlined,ruled,linesnumbered]{algorithm2e}
\usepackage{hyperref}
\hypersetup{
    colorlinks=true,
    linkcolor=blue,
    filecolor=magenta,      
    urlcolor=red,
    }
\usepackage{graphicx}
\newtheorem{theorem}{Theorem}
\newtheorem{lemma}{Lemma}
\newtheorem{definition}{Definition}
\newtheorem{example}{Example}

\usepackage{textcomp}
\usepackage{xcolor}
\def\BibTeX{{\rm B\kern-.05em{\sc i\kern-.025em b}\kern-.08em
    T\kern-.1667em\lower.7ex\hbox{E}\kern-.125emX}}
\begin{document}

\title{Tag-specific Regret Minimization Problem in Outdoor Advertising}

% \author{\IEEEauthorblockN{Dildar Ali}
% \IEEEauthorblockA{\textit{Department of Computer Science and Engineering} \\
% \textit{Indian Institute of Technology Jammu}\\
% NH-44, Jagti, Jammu and Kashmir, 181221, India \\
% 2021rcs2009@iitjammu.ac.in}
% \and
% \IEEEauthorblockN{Abishek Salaria}
% \IEEEauthorblockA{\textit{Department of Information Technology} \\
% \textit{National Institute of Technology Srinagar}\\
% Hazratbal, Srinagar, Jammu and Kashmir, 190006, India\\
% 2023nitsgr185@nitsri.ac.in}
% \and
% \IEEEauthorblockN{Ansh Jasrotia}
% \IEEEauthorblockA{\textit{Department of Information Technology} \\
% \textit{National Institute of Technology Srinagar}\\
% Hazratbal, Srinagar, Jammu and Kashmir, 190006, India\\
% 2023nitsgr161@nitsri.ac.in}
% \and
% \IEEEauthorblockN{ Suman Banerjee}
% \IEEEauthorblockA{\textit{Department of Computer Science and Engineering} \\
% \textit{Indian Institute of Technology Jammu}\\
% NH-44, Jagti, Jammu and Kashmir, 181221, India \\
% suman.banerjee@iitjammu.ac.in }
% }
\author{
\IEEEauthorblockN{
Dildar Ali\IEEEauthorrefmark{1},
Abishek Salaria\IEEEauthorrefmark{2},
Ansh Jasrotia\IEEEauthorrefmark{2},
Suman Banerjee\IEEEauthorrefmark{1}
}
\IEEEauthorblockA{\IEEEauthorrefmark{1}
Indian Institute of Technology Jammu, India, \IEEEauthorrefmark{2}National Institute of Technology Srinagar, India\\
\{2021rcs2009,suman.banerjee\}@iitjammu.ac.in, \{2023nitsgr185,2023nitsgr161\}@nitsri.ac.in}
}

\maketitle
\begin{abstract}
Recently, out-of-home advertising has become popular as one of the most effective marketing techniques, due to its higher return on investment. E-commerce houses approach the influence provider to achieve effective advertising through their tags (advertising content), influence demand, and budgets. The influence provider's goal will be to make proper tag allocations and meet the required influence demand within the budget constraint. If the influence provider provides the required influence, advertisers will make full payment; otherwise, a partial payment based on the fraction of influence (s)he provides. However, if the influence provider provides more influence than is demanded from advertisers, they will not receive an additional incentive. So, in both cases, the loss will occur on the influence provider, and this loss is formalized as \texttt{Regret}. The influence provider's goal will be to minimize the total regret. In this paper, we formalize this problem as a combinatorial optimization problem and refer to it as the \textsc{Tag-specific Regret Minimization in Outdoor Advertising (TRMOA)}. We show that TRMOA is NP-hard and inapproximable within a constant factor. The regret model we consider is non-monotone and non-submodular, and for this simple greedy approach is ineffective. We introduce a fairness-aware greedy round-robin approach that reduces regret and provides balanced allocation across advertisers. To improve scalability and solution quality, we further introduce randomized greedy and randomized local search algorithms. We have experimented with all proposed and baseline methodologies using real-world trajectory and billboard datasets to demonstrate the effectiveness and efficiency of the solution methodologies.
\end{abstract}
% \begin{abstract}
% Recently, out-of-home advertising has become popular marketing techniques, due to its higher return on investment. E-commerce houses approach the influence provider to achieve effective advertising through their tags (advertising content), influence demand, and budgets. The influence provider's goal will be to make proper tag allocations, meet the required influence demand within the budget constraint, and minimize total regret. We formalize this as a combinatorial optimization problem and refer to it as \textsc{Tag-specific Regret Minimization in Outdoor Advertising (TRMOA)}. We show that TRMOA is NP-hard and inapproximable within a constant factor. The regret model we consider is non-monotone and non-submodular, and the simple greedy approach is ineffective. We introduce a fairness-aware greedy round-robin approach that reduces regret with balanced allocation across advertisers. To improve, we also introduce randomized greedy and local search algorithms. We have experimented with all the methodologies using real-world trajectory and billboard datasets to show the effectiveness and efficiency of the solution methodologies.
% \end{abstract}
\begin{IEEEkeywords}
Regret Minimization, Billboard Slot, Tags, Advertiser, Influence Provider, Out-of-Home Advertisement.
\end{IEEEkeywords}

\section{Introduction}
\IEEEPARstart{I}{nfluence} maximization problems are studied extensively in the area of advertisement. There are several advertising methods, such as social media, billboards, television, and newspapers, to maximize the influence. Among them, billboard advertising is one of the most effective out-of-home (OOH) advertising techniques, offering higher returns on investment. In real life, if advertisers need influence, they buy it from an influence provider. The influence provider provides influence to the customer and earns profit. Most internet giants like Google and Facebook use their platforms to advertise products as influence providers and gain huge profits. Most of the literature on advertising focuses on maximizing influence through social networks and billboard advertising from the customer's perspective. There are a limited number of studies from the influence providers' end. Deviating from the existing literature, we study this problem from the perspective of the influence provider, aiming to maximize its profit. In the rest of the paper, we interchangeably use the following terms: influence provider and host, ads and tags. Since host-based influence maximization is harder to optimize than customer-based optimization, as we will demonstrate shortly. For this work, we chose a tag-specific influence model and describe our problem in OOH advertising, which has been extensively studied in the existing literature \cite{zhang2018trajectory, ali2022influential, zhang2019optimizing}. The general applicability of this paper is described in a later section.
% \begin{figure}[t]
%     \centering
%     \includegraphics[width=0.85\linewidth]{Tags.drawio.pdf}
%     \caption{TRMOA framework}
%     \label{fig:framework}
% \end{figure}
% \vspace{-0.1in}
\paragraph{Our Observations} 
In the existing literature on the OOH scenario, all share two types of objectives: (a) helping the single or multiple advertisers maximize their total influence \cite{ali2022influential,ali2023influential,aslay2017revenue,lai2017improved, zhang2018trajectory,zhang2019optimizing,zhang2020towards}. (b) minimize the total regret and maximize profit of an influence provider in a multi-advertising setting \cite{zhang2021minimizing,sharma2024minimizing,ali2023efficient,ali2024regret,ali2024toward,ali2024minimizing}. All existing literature uses a user-based influence model. However, a more challenging aspect for real-world influence providers is managing multiple advertisers on a daily basis. Advertisers come with a budget, demand, and a set of tags. The influence provider's task is to select a subset of tags and allocate tag-specific influence to the advertiser within a given budget; otherwise, the advertiser will not pay the full payment. 

\paragraph{Our Problem}
 Motivated by this observation, we study the slot allocation problem from the perspective of the influence provider. The influence provider owns many billboard slots and assigns them to advertisers. Each advertiser aims to obtain a subset of slots whose total influence meets her demand. The influence provider earns the highest payment only when all advertisers’ demands are fully satisfied. Based on this setting, we adopt a \textit{Regret Model} \cite{zhang2019optimizing} to guide the influence provider's allocation decisions.
 We identify two sources of regret that directly impact profit: unsatisfied regret and excessive regret. Unsatisfied regret occurs when the influence provider fails to meet an advertiser's demand, while excessive regret occurs when an advertiser is assigned more influence than required. These two types of regret are independent and may occur simultaneously, as illustrated in Example \ref{Exalmpe:1}. 

\begin{example}\label{Exalmpe:1}
Suppose an influence provider owns five billboard slots $\mathcal{BS} = \{s_1, s_2, s_3, s_4, s_5\}$, where the influence of each slots is given in Table~\ref{tab:billboard-influence-x}. Three advertisers $\mathcal{A} = \{a_1, a_2, a_3\}$ approach the influence provider, each specifying a demanded influence $\mathcal{I}_i$ and a payment $\mathcal{L}_i$, as shown in Table~\ref{tab:advertiser-contract-x}. To serve these advertisers, the influence provider assigns each advertiser a disjoint set of billboards $\mathcal{S}_i$.

\begin{table}[!h]
\centering
\begin{minipage}{0.48\columnwidth}
\centering
\caption{Billboard Slot Info.}
\label{tab:billboard-influence-x}
\begin{tabular}{|c|c|c|c|c|c|}
\hline
$\mathcal{BS}$ & \(s_1\) & \(s_2\) & \(s_3\) & \(s_4\) & \(s_5\) \\
\hline
\(\mathcal{I}(s_i)\) & 4 & 5 & 3 & 6 & 2 \\
\hline
\end{tabular}
\end{minipage}\hfill
\begin{minipage}{0.48\columnwidth}
\centering
\caption{Advertiser Info.}
\label{tab:advertiser-contract-x}
\begin{tabular}{|c|c|c|c|}
\hline
$\mathcal{A}$ & \(a_1\) & \(a_2\) & \(a_3\) \\
\hline
\(\mathcal{I}_i\) & 6 & 7 & 8 \\
\hline
\(\mathcal{L}_i\) & \$9 & \$12 & \$18 \\
\hline
\end{tabular}
\end{minipage}
\end{table}

The first deployment Strategy-I is shown in Table~\ref{tab:strategy1-x}. The difference between the supply and demand of influence is calculated as $\mathcal{X} = \mathcal{I}(\mathcal{S}_i)-\mathcal{I}_i$. Advertiser $a_1$ is over-satisfied, leading to excessive regret, while advertisers $a_2$ and $a_3$ are under-satisfied, causing unsatisfied regret. An alternative deployment Strategy-II is shown in Table~\ref{tab:strategy2-x}. Although minor deviations exist, this strategy results in a lower total regret and is therefore preferable.

%%%%%%%%%%%%%%%%%%%%%%%%%%%%%%%%%
\begin{table}[!h]
\centering
\begin{minipage}{0.46\columnwidth}
\centering
\caption{Strategy I}
\label{tab:strategy1-x}
\scriptsize
\begin{tabular}{|c|c|c|c|}
\hline
$\mathcal{A}$ & \(a_1\) & \(a_2\) & \(a_3\) \\
\hline
$S_i$ & $s_2, s_5$ & $s_4$ & $s_1, s_3$ \\
\hline
Satisfy & Y & N & N \\
\hline
$\mathcal{X}$ & 1 & -1 & -1 \\
\hline
\end{tabular}
\end{minipage}\hfill
\begin{minipage}{0.46\columnwidth}
\centering
\caption{Strategy II}
\label{tab:strategy2-x}
\scriptsize
\begin{tabular}{|c|c|c|c|}
\hline
$\mathcal{A}$ & \(a_1\) & \(a_2\) & \(a_3\) \\
\hline
$\mathcal{S}_i$ & $s_1, s_5$ & $s_2, s_3$ & $s_4$ \\
\hline
Satisfy & Y & Y & N \\
\hline
$\mathcal{X}$ & 0 & 1 & -2 \\
\hline
\end{tabular}
\end{minipage}
\end{table}
\end{example}

Note that the influence provider may suffer only unsatisfied regret ($a_3$ in Table \ref{tab:strategy2-x}), only excessive regret ($a_2$ in Table \ref{tab:strategy2-x}), no excessive or unsatisfied regret ($a_1$ in Table \ref{tab:strategy2-x}). To simplify our illustration here, we calculate influence $\mathcal{I}(\mathcal{S})$ simply by aggregating $\mathcal{I}(s_{i})$, however, in the real case influence $\mathcal{I}(\mathcal{S})$ is calculated by the influence model described in Definition \ref{Def:Inf_slot}. We formulate the problem as the \underline{T}ag-specific \underline{R}egret \underline{M}inimization in \underline{O}utdoor \underline{A}dvertising (TRMOA) problem.

% \paragraph{Hardness}

% \paragraph{Our Solutions}
% To address this problem, we first reduce the search space using Adaptive Influential Tag Selection (AITS). We introduce a fairness-aware greedy round-robin algorithm that prioritizes regret reduction while balancing allocations across advertisers. To improve scalability and solution quality, we further propose randomized greedy and randomized local search algorithms.

\paragraph{General Applicability}
Although motivated by outdoor advertising, our framework applies to a wide range of influence-driven allocation problems, including online advertising, recommendation systems, and resource provisioning platforms, where both under- and over-allocation lead to loss.

\paragraph{Our Contributions}
In summary, this work makes the following key contributions:
\begin{itemize}
    \item We formulate the Tag-Specific Regret Minimization in Outdoor Advertising (TRMOA) problem from the perspective of the influence provider, explicitly modeling both unsatisfied and excessive influence.
    \item We establish the computational hardness of the problem, highlighting the challenges of finding optimal solutions.
    \item We introduce an adaptive method for selecting influential tags to reduce computational overhead.
    \item We propose three efficient algorithms: two greedy-based methods and a randomized local search approach that balances regret minimization, fairness, and scalability.
    \item We validate the proposed solutions through extensive experiments on real-world billboard and trajectory datasets, showing their effectiveness and practical relevance.
\end{itemize}

\paragraph{Organization of the Paper.}
The rest of this paper is organized as follows. Section~\ref{Sec:RW} reviews the related literature. Section~\ref{Sec:problem-definition} presents the necessary background and formally defines the problem. The proposed methodologies are described in Section~\ref{Sec:methodologies}. Section~\ref{Sec:experimental-evaluation} reports the experimental setup and results. Finally, Section~\ref{Sec:Conclusion} concludes the paper and outlines directions for future research.

\section{Related Work}\label{Sec:RW}
\subsection{Influence Maximization}
With the increasing availability of handheld electronic devices, large-scale trajectory data has become easier to obtain. For this reason, the influence maximization problem in outdoor advertising has attracted the attention of researchers \cite{wang2021survey,zheng2015trajectory,zhang2020towards,tardos2003maximizing,sharma2024minimizing}. Most influence maximization work is studied from the perspective of social network advertising \cite{tardos2003maximizing,sharma2024minimizing,banerjee2019maximizing,banerjee2020survey}. In billboard advertising \cite{zhang2019optimizing, zhang2020towards}, the authors adopt slightly different influence measures to address diverse business needs from the advertiser's perspective. In these models, a billboard influences an audience only when it lies sufficiently close to the audience's travel trajectory. More specifically, studies \cite{zhang2020towards,zhang2018trajectory,ali2022influential} measure influence using traffic volume while accounting for overlap, whereas \cite{zhang2019optimizing} avoids double-counting users who encounter the same advertisement on multiple billboards by using impression counts to trigger an influence indicator. The key distinction between these works and our TRMOA lies in the objective. We focus on the host, which serves multiple advertisers simultaneously, aiming to minimize the host's overall loss while providing just enough influence to satisfy each advertiser's demand. In contrast, prior studies focus solely on maximizing an individual advertiser's influence. Some loosely related work also exists. Studies \cite{guo2018efficient,liu2016smartadp} develop visualization tools to assist advertisers in billboard selection, while \cite{lai2017improved} leverages social media data to enhance the effectiveness of targeted OOH advertising.

\subsection{Regret Minimization}
In the Social Viral Marketing (SVM) setting, influence-provider platforms such as Twitter and Facebook promote advertisements on social networks and receive payments from advertisers, typically under a cost-per-engagement (CPE) model. In this model, advertisers pay the platform for each click or engagement generated by their advertisements \cite{facebook_business_cpe,haven2007marketing,wang2020social}. Recent studies on SVM, with a prominent line of research focusing on minimizing the regret of the influence provider \cite{aslay2015viral,aslay2017revenue,banerjee2019maximizing}. Here, each advertiser specifies a required level of influence. If the achieved influence falls short of this demand, the advertiser pays only for the actual influence obtained. Conversely, if the achieved influence exceeds the demand, the advertiser does not pay for the surplus. Consequently, both under-delivery and over-delivery of influence result in regret for the influence provider \cite{aslay2014viral,ali2024regret,ali2024toward,ali2024minimizing}.   Other studies under the CPE model focus on revenue maximization, where total revenue is defined as the aggregate payment collected from all advertisers. In this setting, each advertiser’s payment equals the sum of the CPE values associated with the activated users \cite{aslay2015viral,banerjee2019maximizing}. In the advertising literature, regret minimization has primarily been studied in two contexts: SVM and the Minimizing Regret for the OOH Advertising Market (TRMOA) problem. These two settings differ fundamentally in both business models and influence mechanisms. First, the payment structures differ. In SVM, payments are engagement-based, and in some cases, the influence provider may receive no payment if the advertiser’s demand is not satisfied \cite{zhang2021minimizing,ali2024minimizing}. Second, the influence models are distinct. In TRMOA, influence is driven by geographical proximity; users are influenced only if they physically encounter a billboard, and influence does not propagate between users \cite{ali2022influential,zhang2021minimizing,ali2024minimizing,ali2024toward}. In contrast, SVM relies on probabilistic diffusion models such as Independent Cascade and Linear Threshold, where influence spreads through social connections in the network \cite{barbieri2013topic,bian2020efficient,guo2020influence}. As a result, SVM research primarily focuses on estimating influence diffusion in virtual social networks, whereas TRMOA centers on spatial interactions between users and billboards. These differences make the underlying optimization problems in the two settings fundamentally distinct.

\section{Preliminaries and Problem Formulation}\label{Sec:problem-definition}
% In this section, we discuss the background of the problem and formally define it. For any positive integer $x$, $[x]$ denotes the set $\{ 1,2, \ldots, k\}$ and for any two positive integers $x$, $y$ with $x<y$ denotes the set $\{x, x+1, \ldots, y\}$. Initially, we start by describing the notion of a set function and its properties.

\subsection{Trajectory, Billboard and Tag Database}
A trajectory database records the locations of moving individuals over time. In our setting, the trajectory database $\mathcal{D}$ consists of tuples of the form $(\mathcal{U}^{'}, \texttt{loc}, [t_1, t_2])$, indicating that the group of individuals $\mathcal{U}^{'}$ was present at location $\texttt{loc}$ during the time interval $[t_1, t_2]$. The database $\mathcal{D}$ contains $m$ such tuples. For any tuple $d \in \mathcal{D}$, let $t_u$ denote the set of individuals associated with that tuple. Let $\mathcal{U} = {u_1, u_2, \ldots, u_r}$ represent the set of all individuals covered by the trajectory database. Formally, $\mathcal{U}$ is defined as the union of all individuals appearing in the tuples of $\mathcal{D}$, that is, $\mathcal{U} = \bigcup_{t \in \mathcal{D}} t_u$, meaning that an individual belongs to $\mathcal{U}$ if they appear in at least one tuple of the database. Similarly, let $\mathcal{L}$ denote the set of locations covered by the trajectory database. A location belongs to $\mathcal{L}$ if it appears in at least one tuple of $\mathcal{D}$.
Finally, let $[T_1, T_2]$ denote the overall time span of the trajectory database. For every tuple $d \in \mathcal{D}$, its associated time interval $[t_1, t_2]$ lies within this global duration, i.e., $[t_1, t_2] \subseteq [T_1, T_2]$.

\par A billboard database $\mathcal{B}$ contains information about billboards deployed across a city. Each record in $\mathcal{B}$ is represented as a tuple $(b_{id}, \texttt{loc}, \texttt{slot\_duration}, \texttt{cost})$, where $b_{id}$ denotes unique id of the billboard, $\texttt{loc}$ indicates its location, and $\texttt{slot\_duration}$ specifies the duration and cost of each advertising slot. Consider there are $\ell$ billboards $\bar{\mathcal{B}}=\{b_1, b_2, \ldots, b_{\ell}\}$ and each one of them is running for the interval $[T_{1}, T_{2}]$ and assume that $T=T_{2}-T_{1}$. Also, assume that all billboards are allocated slot-wise for display advertising, and that the duration of each slot is fixed and denoted by $\Delta$. Let, $\mathcal{BS}$ denotes the set of all billboard slots; i.e.; $\mathcal{BS}=\{(b_i,[t_j,t_j+\Delta]): i \in [k] \text{ and } j \in \{T_{1}, 1 + \Delta, 1+ 2 \Delta, \ldots, T_{2}- \Delta +1\}, cost \}$. Now, if $|\mathcal{BS}|= h $ then $h =k \cdot \frac{T}{\Delta}$. Next, the tag database contains information about the contents of advertisements provided by commercial clients. Specifically, the tag database $\mathcal{T}$ consists of tags, each identified by a unique ID. In our setting, each billboard in $\mathcal{B}$ is associated with a trajectory database, and we restrict attention to billboards whose locations coincide with at least one location appearing in the trajectory data. We now proceed to describe the billboard advertisements.

\subsection{Billboard Advertisements}
Given a set of billboard slots set $\mathcal{S}$, the influence of $\mathcal{S}$ can be denoted as $\mathcal{I}(\mathcal{S})$. As reported in the previous work, there are several ways to calculate the influence on influence maximization \cite{zhang2018trajectory,zhang2020towards, ali2022influential,ali2023influential} and regret minimization \cite{zhang2021minimizing,ali2023efficient,ali2024regret}, in billboard advertisements. We adopt the same measurements as those used in previous studies. Next, the definition of billboard slot influence is stated in Definition \ref{Def:Inf_slot}.

\begin{definition}[Billboard Slot Influence]\label{Def:Inf_slot}
Given a subset of billboard slots $\mathcal{S} \in \mathcal{BS}$, the influence of $\mathcal{S}$ is denoted by $\mathcal{I}(\mathcal{S})$ and defined as the aggregate influence probability over all users in the trajectory database in Equation \ref{Eq:inf}.
\begin{equation} \label{Eq:inf}
\mathcal{I}(\mathcal{S})= \underset{u \in \mathcal{D}}{\sum} [1-  \underset{b \in \mathcal{S}}{\prod}(1-Pr(u,b))]
\end{equation}  
\end{definition}
Here, the influence function $\mathcal{I}(\cdot)$ maps each subset of billboard slots to its corresponding influence value, i.e., $\mathcal{I} : 2^{\mathcal{BS}} \rightarrow \mathbb{R}^{+}_{0}$, with $\mathcal{I}(\emptyset) = 0$. 
\begin{lemma}\label{lem:inf}
The influence function $\mathcal{I}(\cdot)$ is non-negative, monotone, and submodular.
\end{lemma}

As discussed earlier, when a person sees an advertisement on a billboard slot, they become influenced with some probability. Most existing studies \cite{zhang2019optimizing,zhang2020towards,zhang2021minimizing,ali2023influential} assume that this probability is independent of the advertisement content, and therefore ignore the effect of the specific advertisement (e.g., its tag). However, in reality, influence depends strongly on the ad's content. For instance, a person who is not interested in alcohol or soft drinks is unlikely to be influenced by such advertisements. On the other hand, if the same person is interested in movies, then movie-related advertisements are much more meaningful, as they may motivate him to watch a film after being influenced by the promotion.

\par Therefore, in this paper, we explicitly treat tags as a key factor influencing user behavior. Let $\mathcal{T}_{u_i}$ denote the set of tags associated with user $u_i$. Considering all users in $\mathcal{U}$, the overall set of relevant tags is given by $\mathcal{T} = \underset{u_i \in \mathcal{U}}{\bigcup} \mathcal{T}_{u_i}$ and each tag $x \in \mathcal{T}$, we denote the influence probability by $Pr(u \mid x)$, which is formally defined in Definition~\ref{Def:Tagprobability}.

\begin{definition}[Tag-Specific Influence]\label{Def:Tagprobability}
For a user $u \in \mathcal{U}$ and a subset of tags $\mathcal{T}^{'} \subseteq \mathcal{T}$, the tag-specific influence probability is
\begin{equation}\label{Eq:Tagspecificprobability}
Pr(u \mid \mathcal{T}^{'}) = 1 - \prod_{x \in \mathcal{T}^{'}} \big(1 - Pr(u \mid x)\big).
\end{equation}
\end{definition}

Next, we define the tag-specific influence of a set of billboard slots in Definition~\ref{Def:TSIBS}.

\begin{definition}[Tag-Specific Influence of Billboard Slots]\label{Def:TSIBS}
For a set of billboard slots $\mathcal{S} \subseteq \mathcal{BS}$ and tags $\mathcal{T}^{'} \subseteq \mathcal{T}$, the tag-specific influence is
\begin{equation}\label{Eq:TSIBS}
\mathcal{I}(\mathcal{S} \mid \mathcal{T}^{'})
= \sum_{u \in \mathcal{U}}( 1 - \prod_{b \in \mathcal{S}} \big(1 - Pr(u,b \mid \mathcal{T}^{'})\big)).
\end{equation}
\end{definition}

The function $\mathcal{I}(\cdot)$ jointly captures the impact of selected tags and billboard slots, mapping each pair $(\mathcal{T}^{'}, \mathcal{S})$ to a non-negative value, i.e., $\mathcal{I} : 2^{\mathcal{T}} \times 2^{\mathcal{BS}} \rightarrow \mathbb{R}^{+}_{0}$.

\subsection{Problem Definition}
Assume there are $n$ advertisers $\mathcal{A}={a_1,a_2,\ldots,a_n}$ and one influence provider $\mathcal{P}$. Each advertiser $a_i$ submits a campaign request specifying influence demand $\sigma_i$ over tags $x_1,x_2,\ldots,x_k$ with payment $u_i$. The influence provider maintains an advertiser database $\mathbb{A}=(a_i,\sigma_i,\langle x_{1},x_{2}, \ldots, x_{k}\rangle,u_i)$ for all $i\in[n]$.
Let $\mathcal{W}_{a_i}$ denote the set of slots allocated to advertiser $a_i$. As per the payment rule, if $\mathcal{I}(\mathcal{W}_{a_i}) \ge \sigma_i$, then $\mathcal{P}$ receives the full payment $u_i$; otherwise, only a partial (pro-rata) payment is received. We highlight two important cases:
\begin{itemize}
\item \textbf{Unsatisfied Regret:} If the provided influence is less than the required demand $\mathcal{I}(\mathcal{W}_{a_i}) < \sigma_i$, the provider receives a partial payment, resulting in a loss. We call this \textit{Unsatisfied Regret}.
\item \textbf{Excessive Regret:} If more influence is provided than required $\mathcal{I}(\mathcal{W}_{a_i}) > \sigma_i$, the extra influence does not generate additional revenue. Moreover, that excess influence could have been reassigned to another advertiser whose demand is unsatisfied. We call this \textit{Excessive Regret}.
\end{itemize}

Now, combining the two types of loss, we formally define the regret model in Definition \ref{Def:Regret-model}.

\begin{definition}[The Regret Model] \label{Def:Regret-model}
For advertiser $a_i \in \mathcal{A}$ with allocated billboard slots $\mathcal{W}_{a_i}$, the regret is denoted by $\mathcal{R}_{a_i}$ and defined as
 \[
    \mathcal{R}_{a_i} = 
\begin{cases}
    u_{i} \cdot (1- \delta \cdot \frac{\mathcal{I}(\mathcal{W}_{a_i})}{\sigma_{i}}),& \text{if }  \sigma_{i} > \mathcal{I}(\mathcal{W}_{a_i}) \\
    u_{i} \cdot \frac{\mathcal{I}(\mathcal{W}_{a_i}) - \sigma_{i}}{\sigma_{i}},              & \text{otherwise}
\end{cases}
\]
Here, $\frac{\mathcal{I}(\mathcal{W}_{a_i})}{\sigma_i}$ represents the satisfied fraction of demand, and $\delta$ is the penalty ratio for unsatisfied demand.
\end{definition}
% The proposed solution methods are independent of the specific value of $\delta$; this is further discussed in Section~\ref{subsec:otherparameter}.

\paragraph{\textbf{Allocation of Slots}}
The objective of this work is to allocate slots so that the tag-specific influence demand of advertisers is satisfied and the total regret of an influence provider is minimized. We consider the following two constraints:

\textit{Tag-specific Influence Demand:}
This constraint says that for any advertiser $a_{i} \in \mathcal{A}$, the demand for tag-specific influence should be satisfied. 

\textit{Disjointness Constraints:} For any two advertisers $a_{x}$ and $a_{y}$, let $\mathcal{W}_{a_x}$ and $\mathcal{W}_{a_y}$ be the allocations. The disjoint constraint says $\mathcal{W}_{a_x} \cap \mathcal{W}_{a_y} = \emptyset$. 

Now, let $\mathcal{G}$ be the set of all possible allocations for $n$ advertisers. Let a possible allocation be denoted as $\{\mathcal{W}_{a_1},\mathcal{W}_{a_2}, \ldots, \mathcal{W}_{a_n}\}$. Now, we define it in Definition \ref{Def:FA}.

\begin{definition}[Feasible Allocation of Slots]\label{Def:FA}
An allocation of slots to advertisers is feasible if it satisfies the Tag-specific influence demand and disjointness constraints.
\end{definition}

As discussed earlier, our main objective is to minimize the total regret of an influence provider. So, we define the notion of total regret in Definition \ref{Def:TR}. 

\begin{definition}[Total Regret]\label{Def:TR}
  The total regret associated with an allocation $\mathcal{Z}= \{\mathcal{W}_{a_1},\mathcal{W}_{a_2}, \ldots, \mathcal{W}_{a_n}\}$ can be denoted as $\mathcal{R}(\mathcal{Z})$ and defined as the aggregated regret associated with all the advertisers in Equation \ref{Eq:Total_Regret}.
\begin{equation}\label{Eq:Total_Regret}
\mathcal{R}(\mathcal{Z})= \underset{a_i \in \mathcal{A}}{\sum}  \mathcal{R}_{a_i}
\end{equation} 
\end{definition}

Now, we define our TRMOA problem in Definition \ref{Def:Problem}.

\begin{definition}[Tag-specific Regret Minimization Problem]\label{Def:Problem}
Given billboard slots $\mathcal{BS}$, Trajectory information $\mathcal{D}$, and advertisers $\mathcal{A}$, the goal of this problem is to create an allocation such that the total regret of an influence provider is minimized. Mathematically, it can be written as,

\begin{equation}
\mathcal{Z}^{OPT} = \underset{\mathcal{Z}_{i} \in \mathcal{G}(\mathcal{Z})}{argmin} \ \mathcal{R}(\mathcal{Z}_{i})
\end{equation}
\end{definition}
 Now, from a computational view, the problem can be written as follows:

\begin{center}
\begin{tcolorbox}[title=\textsc{TRMOA Problem}, width=8.5cm]
\textbf{Input:} Billboard Slots $\mathcal{BS}$, Influence Function $\mathcal{I}()$, Trajectory Database $\mathcal{D}$, Advertiser Database $\mathbb{A}$.

\textbf{Problem:} Find out an optimal allocation $\mathcal{Z}^{OPT} = \{\mathcal{W}_{a_1}, \mathcal{W}_{a_2}, \ldots, \mathcal{W}_{a_n}\}$ of slots that minimizes the overall regret.
\end{tcolorbox}
\end{center}

\subsection{Problem Hardness.}

\begin{theorem}
The Tag-Specific Regret Minimization in Outdoor Advertising (TRMOA) problem is NP-hard and NP-hard to approximate within any constant factor.
\end{theorem}

\begin{proof}
We prove the hardness of TRMOA via a reduction from  \emph{Numerical 3-Dimensional Matching (N3DM)} problem, which is known to be NP-complete. Let $X, Y, Z$ be three multisets of integers, each containing $n$ elements. Let $b$ be a bound such that $b = \frac{\sum X + \sum Y + \sum Z}{n}$. The N3DM decision problem asks whether there exists a perfect matching $M \subseteq X \times Y \times Z$ such that every element appears exactly once and for each triple $(x_i, y_i, z_i) \in M$, $x_i + y_i + z_i = b$. Given an instance of N3DM, we construct a TRMOA instance as follows. (1) Create $3n$ billboard slots and partition them into three disjoint sets: $D_1, D_2, D_3$, each of size $n$. (2) Map each element in $X \rightarrow D_1, Y \rightarrow D_2, Z \rightarrow D_3$. (3) For each slot $s_i$, construct a disjoint user set so that under Definition~\ref{Def:TSIBS}, $\mathcal{I}(\{s_i\} \mid \mathcal{T}) = w_i,$ where $w_i$ is the mapped integer. Since users are disjoint across slots, the full tag-specific influence function remains: $\mathcal{I}(\mathcal{S} \mid \mathcal{T}) = \sum_{s_i \in \mathcal{S}} w_i$, while preserving non-negativity, monotonicity, and submodularity. (4) Introduce a sufficiently large constant $c$. Modify slot influences as: $\forall~ s_i \in D_1: \mathcal{I}(s_i) \leftarrow c + w_i$, $\forall~ s_j \in D_2: \mathcal{I}(s_j) \leftarrow 3c + w_j$, $\forall~ s_k \in D_3: \mathcal{I}(s_k) \leftarrow 9c + w_k$. (5) Create $n$ advertisers with identical, $\sigma = b + 13c$.

\par We observe that $(c + w_i) + (3c + w_j) + (9c + w_k) = 13c + (w_i + w_j + w_k)$. When $c$ is sufficiently large, any deviation from selecting exactly one slot from each of $D_1, D_2, D_3$ yields influence strictly different from $\sigma$. Thus, regret of zero is achieved \emph{if and only if} each advertiser receives exactly one slot from each set.

\textbf{Correctness.} If TRMOA achieves total regret zero, then for every advertiser $a_i$: $\mathcal{I}(\mathcal{S}_i) = \sigma$. By the coefficient dominance of $c$, each $\mathcal{S}_i$ must contain exactly one element from each of $D_1, D_2, D_3$. Removing the $13c$ offset gives $w_i + w_j + w_k = b$. Thus, a valid N3DM matching exists. If a valid N3DM matching exists, assign each triple to a distinct advertiser. Each advertiser receives influence exactly $\sigma$, hence regret is zero.

\textbf{Approximation Hardness.}
Let $\text{OPT}^N$ denote the optimal value of N3DM (i.e., number of violated triples), and let $\text{OPT}^T$ denote the minimum regret in TRMOA. We have $\text{OPT}^N = 0, \text{OPT}^T = 0$. Suppose a polynomial-time algorithm approximates TRMOA within any constant factor $\tau$. Then, $\mathcal{R}(\mathcal{S}) \le \tau \cdot \text{OPT}^T$. Hence, If $\text{OPT}^N = 0$, then $\mathcal{R}(\mathcal{S})=0$. If $\text{OPT}^N > 0$, then $\mathcal{R}(\mathcal{S})>0$. Thus, we can solve N3DM in polynomial time by checking whether the regret equals zero, which implies $\mathrm{P}=\mathrm{NP}$. Therefore, TRMOA is NP-hard and NP-hard to approximate within any constant factor.
\end{proof}

\section{Proposed Solution Methodologies}\label{Sec:methodologies}
% We propose three approaches to solve this problem. 
\subsection{Adaptive Influential Tag Selection (AITS)}
 In our problem context, selecting tags is the initial requirement for all subsequent proposed solution approaches. In the AITS approach, use the trajectory database $\mathcal{D}$, the Tag set $\mathcal{T}$, and the adaptive stopping threshold $(\omega = 0.01)$ as input. This incrementally constructs an influential tag set using a greedy strategy. At each iteration, it selects the tag with the maximum marginal influence gain $\Delta(t\mid \mathcal{T}^{'}) = \mathcal{I}(\mathcal{T}^{'} \cup \{t\}) - \mathcal{I}(\mathcal{T}^{'})$ with respect to the current selection. The process ends when the best marginal gain falls below $\omega$-fraction of the current influence, preventing negligible additions. The algorithm outputs a set of tags $\mathcal{T}^{'}$ that collectively maximize influence over trajectories.
 
\begin{algorithm}[h!]
\scriptsize
\caption{Adaptive Influential Tag Selection}
\label{Algo:AITS}
\KwIn{Tag set $\mathcal{T}$, Trajectory data $\mathcal{D}$, $\omega$}
\KwOut{Selected tag set $\mathcal{T}^{'}$}

$\mathcal{T}^{'} \leftarrow \emptyset$\;

\While{True}{
$t^* = \arg\max_{t\in\mathcal{T}\setminus\mathcal{T}^{'}} \Delta(t\mid \mathcal{T}^{'})$\;

\If{$\Delta(t^*\mid \mathcal{T}^{'}) < \omega \cdot \mathcal{I}(\mathcal{T}^{'})$}{
\textbf{break}
}

$\mathcal{T}^{'} \leftarrow \mathcal{T}^{'} \cup \{t^*\}$\;
}

\Return $\mathcal{T}^{'}$\;
\end{algorithm}

\paragraph{\textbf{Complexity Analysis.}}
We analyze the time and space complexity of Algorithm~\ref{Algo:AITS}. The initialization step (Line~1) takes $\mathcal{O}(1)$ time. In each iteration, Line~$3$ scans all remaining tags to identify tags with maximum marginal influence gain, requiring $\mathcal{O}(|\mathcal{T}|^{2} \cdot c)$ time, where $c$ is the number of tuples in the trajectory database. The stopping condition in Line~4 and the update in Line~$6$ each take $\mathcal{O}(1)$ time. If the algorithm selects $|\mathcal{T}^{'}|$ tags before termination, the total time complexity is $\mathcal{O}(|\mathcal{T}^{'}| \cdot |\mathcal{T}|^{2} \cdot c)$. The additional space complexity is $\mathcal{O}(|\mathcal{T}|)$.

\subsection{Fairness-Aware Round-Robin Greedy Algorithm for Tag-Specific Regret Minimization (BG)}
Algorithm~\ref{Algo:RR-Tag-Greedy} proposes a fairness-aware greedy framework to minimize the total regret of an influence provider. First, advertisers are sorted by budget–demand ratio $\frac{u_i}{\sigma_i}$. For each advertiser, the algorithm first refines the tag set using Algorithm \ref{Algo:AITS}, retaining only the most influential tags to eliminate low-impact tags. Allocation then proceeds in a round-robin manner across the refined tag set, ensuring inter-tag fairness and preventing domination by any single tag. At each step, the algorithm greedily selects the billboard slot that maximizes regret reduction per unit cost, explicitly optimizing a normalized regret objective. The process continues until advertiser-specific influence demands are met or resources are exhausted. The resulting allocation achieves regret minimization, budget efficiency, and a fair distribution of slots across advertisers and tags, yielding a globally balanced solution.

\begin{algorithm}[h!]
\scriptsize
\SetAlgoLined
\KwData{
Trajectory Database $\mathcal{D}$, Billboard Slot Information $\mathcal{BS}$, Advertiser information $\mathcal{A}$, Influence Function $\mathcal{I}(\cdot)$, Tag information $\mathcal{T}$
}
\KwResult{
An allocation $\mathcal{Y}$ of billboard slots to advertiser--tag pairs minimizing total regret
}
\textbf{Initialization:} \\
$\mathcal{Y} \leftarrow \{ \mathcal{S}_{i,t} = \emptyset \mid a_i \in \mathcal{A},\, t \in \mathcal{T}_i \}$\;
$\mathcal{Z} \leftarrow \{ Z_{1}, Z_{2}, \ldots, Z_{|\mathcal{A}|}   \}$

Sort advertisers $a_i \in \mathcal{A}$ in descending order of $\frac{u_i}{\sigma_i}$\;

\For{each advertiser $a_i \in \mathcal{A}$}{
    
    Compute refined tags $\mathcal{T}_i^{'} \leftarrow AITS(\mathcal{T}_i, \mathcal{D})$\;
    
    Let $\mathcal{T}_i^{'} = \{t_{i,1}, t_{i,2}, \ldots, t_{i,|\mathcal{T}_i^{'}|}\}$ be the tag set of $a_i$\;
    
    Initialize tag pointer $k \leftarrow 0 $\;
    
    \While{$\exists\, t \in \mathcal{T}_i^{'}$ such that $\mathcal{I}(\mathcal{Z}_{i}) < \sigma_{i}$ \textbf{and} $\mathcal{BS} \neq \emptyset$}{
        
        $t \leftarrow t_{i,k}$\;  
            $s^{*} \leftarrow 
            \arg\max\limits_{s \in \mathcal{BS}}
            \frac{
            \mathcal{R}(\mathcal{Z}_{i}) -
            \mathcal{R}(\mathcal{Z}_{i} \cup \{s\})
            }{
            \sigma(\{s\})
            }$\;
            
            $\mathcal{S}_{i,t} \leftarrow \mathcal{S}_{i,t} \cup \{s^{*}\}$\;
            $\mathcal{Z}_{i} \leftarrow \mathcal{Z}_{i} \cup \{s^{*}\} $
            
            $\mathcal{BS} \leftarrow \mathcal{BS} \setminus \{s^{*}\}$\;
        
        $k \leftarrow (k+1) \bmod |\mathcal{T}_i|$\; 
    }
}
\Return $\mathcal{Y}$\;
\caption{Fairness-Aware Round-Robin Greedy Algorithm for Tag-Specific Regret Minimization}
\label{Algo:RR-Tag-Greedy}
\end{algorithm}

\paragraph{\textbf{Complexity Analysis.}}
Now, we analyze the time and space requirements of Algorithm~\ref{Algo:RR-Tag-Greedy}. Sorting advertisers by $\frac{u_i}{\sigma_i}$ requires $\mathcal{O}(|\mathcal{A}|\log|\mathcal{A}|)$ time. For each advertiser $a_i$, the Adaptive Influential Tag Selection procedure runs in $\mathcal{O}(|\mathcal{T}_i^{'}| \cdot |\mathcal{T}_i|^{2} \cdot c)$ time, where $c$ denotes the number of trajectory tuples. In the allocation phase, each billboard slot is selected at most once, and each selection scans the remaining slots, yielding a total complexity of $\mathcal{O}(|\mathcal{BS}|^{2})$ in the worst case. Therefore, the overall time complexity is $\mathcal{O}(|\mathcal{A}|\log|\mathcal{A}| + \sum_{a_i \in \mathcal{A}} |\mathcal{T}_i^{'}| \cdot |\mathcal{T}_i|^{2} \cdot c + |\mathcal{BS}|^{2})$. The additional space complexity is $\mathcal{O}(|\mathcal{BS}| + |\mathcal{A}|\cdot|\mathcal{T}|)$.

\subsection{Fairness-Aware Randomized Round-Robin Greedy Algorithm for Tag-Specific Regret Minimization (RG)}
Algorithm~\ref{Algo:RR-Randomized-Greedy} allocates billboard slots to advertiser-tag pairs to minimize total regret under influence and budget constraints. Advertisers are prioritized by their normalized budget–demand ratio, and for each advertiser, a refined influential tag set is obtained using Algorithm \ref{Algo:AITS}. Allocation proceeds in a round-robin manner over tags, where at each step a random subset of available slots is sampled, and the slot maximizing regret reduction per unit cost is selected. The sampling is done using $\frac{|\mathcal{BS}|}{k} \log \frac{1}{\epsilon}$, where $k$ is the $10\%$ of available $|\mathcal{BS}|$ in each iteration. From this subset, the slot that maximizes regret reduction per unit cost is selected. This sampling strategy follows the stochastic greedy framework for submodular optimization \cite{mirzasoleiman2014lazierlazygreedy}, which significantly reduces computational complexity while preserving solution quality with high probability.  In general, the algorithm achieves a budget-efficient, fair, and scalable allocation suitable for large-scale billboard datasets.
\begin{algorithm}[h!]
\scriptsize
\SetAlgoLined
\KwData{
Trajectory Database $\mathcal{D}$, Billboard Slot Information $\mathcal{BS}$, Advertiser information $\mathcal{A}$, Influence Function $\mathcal{I}(\cdot)$, Tag information $\mathcal{T}$
}
\KwResult{
An allocation $\mathcal{Y}$ of billboard slots to advertiser--tag pairs minimizing total regret
}
\textbf{Initialization:} \\
$\mathcal{Y} \leftarrow \{ \mathcal{S}_{i,t} = \emptyset \mid a_i \in \mathcal{A},\, t \in \mathcal{T}_i \}$\;
$\mathcal{Z} \leftarrow \{ Z_{1}, Z_{2}, \ldots, Z_{|\mathcal{A}|}   \}$

Sort advertisers $a_i \in \mathcal{A}$ in descending order of $\frac{u_i}{\sigma_i}$\;

\For{each advertiser $a_i \in \mathcal{A}$}{
    
    Compute refined tags $\mathcal{T}_i^{'} \leftarrow AITS(\mathcal{T}_i, \mathcal{D})$\;
    
    Let $\mathcal{T}_i^{'} = \{t_{i,1}, t_{i,2}, \ldots, t_{i,|\mathcal{T}_i^{'}|}\}$ be the tag set of $a_i$\;
    
    Initialize tag pointer $k \leftarrow 0 $\;
    
    \While{$\exists\, t \in \mathcal{T}_i^{'}$ such that $\mathcal{I}(\mathcal{Z}_{i}) < \sigma_{i}$ \textbf{and} $\mathcal{BS} \neq \emptyset$}{
        
        $t \leftarrow t_{i,k}$\; 
        $\mathcal{BS}^{'} \leftarrow$ a random subset obtained by sampling $s$ random
elements from $\mathcal{BS} \setminus \mathcal{Z}_{i}$\;
            $s^{*} \leftarrow 
            \arg\max\limits_{s \in \mathcal{BS}^{'}}
            \frac{
            \mathcal{R}(\mathcal{Z}_{i}) -
            \mathcal{R}(\mathcal{Z}_{i} \cup \{s\})
            }{
            \sigma(\{s\})
            }$\;
            
            $\mathcal{S}_{i,t} \leftarrow \mathcal{S}_{i,t} \cup \{s^{*}\}$\;
            $\mathcal{Z}_{i} \leftarrow \mathcal{Z}_{i} \cup \{s^{*}\} $
            
            $\mathcal{BS} \leftarrow \mathcal{BS} \setminus \{s^{*}\}$\;
        
        $k \leftarrow (k+1) \bmod |\mathcal{T}_i|$\; 
    }
}
\Return $\mathcal{Y}, \mathcal{BS}$\;
\caption{Fairness-Aware Randomized Round-Robin Greedy Algorithm for Tag-Specific Regret Minimization}
\label{Algo:RR-Randomized-Greedy}
\end{algorithm}

\paragraph{\textbf{Complexity Analysis.}}
We analyze the time and space complexity of Algorithm~\ref{Algo:RR-Randomized-Greedy}. Sorting advertisers by $\frac{u_i}{\sigma_i}$ requires $\mathcal{O}(|\mathcal{A}|\log|\mathcal{A}|)$ time. For each advertiser $a_i$, the Adaptive Influential Tag Selection procedure runs in $\mathcal{O}(|\mathcal{T}_i^{'}| \cdot |\mathcal{T}_i|^{2} \cdot c)$ time, where $c$ denotes the number of trajectory tuples. During allocation, each iteration evaluates only a sampled subset of size $s$ from the remaining billboard slots, resulting in $\mathcal{O}(s)$ time per selection. Since each billboard slot is assigned at most once, the total allocation cost is $\mathcal{O}(|\mathcal{BS}| \cdot s)$. Hence, the overall time complexity is $\mathcal{O}(|\mathcal{A}|\log|\mathcal{A}| + \sum_{a_i \in \mathcal{A}} |\mathcal{T}_i^{'}| \cdot |\mathcal{T}_i|^{2} \cdot c + |\mathcal{BS}| \cdot s)$. The additional space complexity is $\mathcal{O}(|\mathcal{BS}| + |\mathcal{A}|\cdot|\mathcal{T}|)$.

\subsection{Randomized Local Search Algorithm for Tag-Specific Regret Minimization (RLS)}
Algorithm \ref{Algo:RLS} proposes a randomized local search framework to further improve tag-specific billboard allocations with respect to total regret. The algorithm starts by initializing empty allocations for all advertiser-tag pairs and sorting advertisers according to $\frac{u_i}{\sigma_i}$. An initial solution is first constructed using the Algorithm \ref{Algo:RR-Randomized-Greedy}, providing a high-quality starting point. The algorithm then performs $N$ randomized local search iterations. In each iteration, billboard slots are reallocated advertiser-by-advertiser using refined influential tags obtained via Algorithm \ref{Algo:AITS}. Allocation proceeds in a round-robin fashion over tags, where at each step a billboard slot is selected uniformly at random from the remaining pool, encouraging exploration of diverse allocation patterns. After each iteration, the total regret of the constructed solution is evaluated and compared against the best solution found so far. If an improvement is observed, the current allocation replaces the incumbent best solution. Finally, any remaining unassigned billboard slots are greedily allocated to the best solution using the Algorithm \ref{Algo:RR-Randomized-Greedy} procedure to further reduce regret. The algorithm returns the best allocation across all iterations, effectively balancing exploration and exploitation to minimize regret.

\begin{algorithm}[h!]
\scriptsize
\caption{Randomized Local Search Algorithm for Tag-Specific Regret Minimization}
\label{Algo:RLS}
\KwData{
Trajectory Database $\mathcal{D}$, Billboard Slot Information $\mathbb{BS}$, 
Advertiser information $\mathcal{A}$, Influence Function $\mathcal{I}(\cdot)$, Sample size $s$, Tag information $\mathcal{T}$
}
\KwResult{
An allocation $\mathcal{Z}$ of billboard slots to advertiser--tag pairs minimizing total regret
}
\textbf{Initialization:} \\
$\mathcal{M} \leftarrow \{ \mathcal{S}_{i,t} = \emptyset \mid a_i \in \mathcal{A},\, t \in \mathcal{T}_i \}$\;
$\mathcal{Z} \leftarrow \{ Z_{1}, Z_{2}, \ldots, Z_{|\mathcal{A}|}   \}$

Sort advertisers $a_i \in \mathcal{A}$ in descending order of $\frac{u_i}{\sigma_i}$\;
$\mathcal{S}^{best}, \mathcal{BS}^{''} \gets$ RG$(\mathcal{BS}, \mathcal{T}, \mathcal{A}, \mathcal{Z})$\;
\For{$iter \gets 1$ \KwTo $N$}{
    $\mathcal{BS}^{'} \gets \mathcal{BS}$\;
\For{each advertiser $a_i \in \mathcal{A}$}{
    
    Compute refined tags $\mathcal{T}_i^{'} \leftarrow AITS(\mathcal{T}_i, \mathcal{D})$\;
    
    Let $\mathcal{T}_i^{'} = \{t_{i,1}, t_{i,2}, \ldots, t_{i,|\mathcal{T}_i^{'}|}\}$ be the tag set of $a_i$\;
    
    Initialize tag pointer $k \leftarrow 0 $\;
    
    \While{$\exists\, t \in \mathcal{T}_i^{'}$ such that $\mathcal{I}(\mathcal{Z}_{i}) < \sigma_{i}$ \textbf{and} $\mathcal{BS} \neq \emptyset$}{
        
        $t \leftarrow t_{i,k}$\;  
            $s^{*} \leftarrow$ Select  a random billboard $o \in \mathcal{BS}^{'}$\;
            
            $\mathcal{S}_{i,t} \leftarrow \mathcal{S}_{i,t} \cup \{s^{*}\}$\;
            $\mathcal{Z}_{i} \leftarrow \mathcal{Z}_{i} \cup \{s^{*}\} $
            
            $\mathcal{BS}^{'} \leftarrow \mathcal{BS}^{'} \setminus \{s^{*}\}$\;
        
        $k \leftarrow (k+1) \bmod |\mathcal{T}_i|$\; 
    }
}
    \If{$\mathcal{R}(\mathcal{Z}) < \mathcal{R}(\mathcal{S}^{best})$}{
        $\mathcal{S}^{best} \gets \mathcal{Z}$\;
        $\mathcal{BS}^{''} \leftarrow \mathcal{BS}^{'}$\;
    }
}
\If{$\mathcal{BS}^{''} \neq \emptyset$}{
$\mathcal{S}^{best} \leftarrow RG(\mathcal{BS}^{''}, \mathcal{T}, \mathcal{A}, \mathcal{S}^{best})$ by skipping line no. $6$ in Algorithm \ref{Algo:RR-Randomized-Greedy}
}
\Return $\mathcal{S}^{best}$\;
\end{algorithm}

\paragraph{\textbf{Complexity Analysis.}}
We analyze the time and space complexity of Algorithm~\ref{Algo:RLS}. Sorting advertisers by $\frac{u_i}{\sigma_i}$ requires $\mathcal{O}(|\mathcal{A}|\log|\mathcal{A}|)$ time. The initial randomized greedy construction incurs $\mathcal{O}(|\mathcal{BS}| \cdot s)$ time. The algorithm then performs $N$ local search iterations; in each iteration, for every advertiser $a_i$, the Adaptive Influential Tag Selection procedure runs in $\mathcal{O}(|\mathcal{T}_i^{'}| \cdot |\mathcal{T}_i|^{2} \cdot c)$ time, where $c$ denotes the number of trajectory tuples, followed by randomized slot assignments that together require $\mathcal{O}(|\mathcal{BS}|)$ time. Hence, the total time complexity is $\mathcal{O}(|\mathcal{A}|\log|\mathcal{A}| + |\mathcal{BS}|\cdot s + N \cdot \sum_{a_i \in \mathcal{A}} \big(|\mathcal{T}_i^{'}| \cdot |\mathcal{T}_i|^{2} \cdot c + |\mathcal{BS}|\big))$. The additional space complexity is $\mathcal{O}(|\mathcal{BS}| + |\mathcal{A}|\cdot|\mathcal{T}|)$.

\section{Experimental Evaluation}\label{Sec:experimental-evaluation}
We experimented with different parameter settings to determine how well the algorithms minimize regret and the computational time required. Since this is the first work on tag-specific regret minimization on billboard advertisement, there is no existing method for direct comparison. Therefore, we conduct an extensive analysis of the two research questions (RQ) below to provide practical insights into minimizing regret across different real-world scenarios.

\begin{tcolorbox}[left=2pt,right=2pt,top=4pt,bottom=4pt]
\begin{itemize}
\item RQ1: What happens when the total demand of all advertisers is far below, close to, or exceeds the influence provider’s maximum supply?
\item RQ2: Which type of advertisers helps the influence provider minimize regret: a small number of large advertisers with high demand, or a large number of small advertisers with low demand?
\end{itemize}
\end{tcolorbox}

\subsection{\textbf{Challenges in TRMOA Problem}} The TRMOA problem has few challenges, as follows.
\begin{itemize}
\item The problem is NP-hard and inapproximable, making it impossible to efficiently compute optimal solutions for large real-world instances. This problem simultaneously penalizes under-allocation (unsatisfied regret) and over-allocation (excessive regret). The objective is to reduce total regret, in which reducing one often increases the other, making balanced allocation non-trivial.
\item Influence depends on both selected tags and allocated slots, creating a large combinatorial search space and complex inter-dependencies. Additionally, multiple advertisers compete for disjoint slots, and satisfying one advertiser can negatively impact others, especially when demand is high. Furthermore, large trajectory and billboard datasets significantly increase computational costs, requiring efficient, scalable heuristic solutions.
\item As this is the first work on tag-specific regret minimization, no existing methods are available for comparison. We compare our methods with `Random' as a baseline.
\end{itemize}
\subsection{\textbf{Experimental Setup}} 
\subsubsection{Datasets}
We experimented with two real-world trajectory datasets: check-in data from New York City (NYC) and Los Angeles (LA). A total of 227,428 check-ins were recorded in NYC\footnote{\url{https://www.nyc.gov/site/tlc/about/tlc-trip-record-data.page}} between April 12, 2012, and February 16, 2013. The LA\footnote{\url{https://github.com/Ibtihal-Alablani}} dataset contains 74,170 user records collected from 15 streets, including street names, GPS locations, timestamps, and related information. We crawled the billboard data from Lamar Advertising Company\footnote{\url{http://www..lamar.com/InventoryBrowser}}, one of the largest outdoor advertising providers worldwide. The NYC billboard dataset includes 716 billboards (1,031,040 slots), while the LA dataset contains 1,483 billboards (2,135,520 slots).

\subsubsection{Key Parameters}
Our experiments use key parameters, including the demand-supply ratio $\alpha$, the average individual demand ratio $\beta$, the distance threshold $\gamma$, and the unsatisfied penalty ratio $\delta$. We summarize all key parameters in Table \ref{Table:Key-Parameter}.
\vspace{-0.15 in}
\begin{table}[h!]
\caption{\label{Table:Key-Parameter} Key Parameters}
\vspace{-0.15 in}
\begin{center}
    \begin{tabular}{ | p{2cm}| p{5.5cm}|}
    \hline
    Parameter & Values  \\ \hline
    $\alpha$ & $40\%, 60\%, 80\%, \textbf{100\%}, 120\%$   \\ \hline
    $\beta$ & $1\%, 2\%, \textbf{5\%}, 10\%, 20\%$  \\ \hline
    $\epsilon$ & $\textbf{0.01}, 0.05, 0.1, 0.15, 0.2$ \\ \hline
    $\omega$ & $\textbf{0.01}, 0.05, 0.1, 0.15, 0.2$ \\ \hline
    $\delta$ & $0, 0.25, \textbf{0.5}, 0.75, 1$  \\ \hline
    $\gamma$ & $25m,50m,\textbf{100m},125m,150m$  \\ \hline
    \end{tabular}
\end{center}
\end{table}
\vspace{-0.15 in}

\paragraph{Demand-supply Ratio $(\alpha)$}
This is the ratio of total influence demand over influence supply i.e., $\alpha = \sigma^{d} / \sigma^{s}$, where $\sigma^{d}$ = $\sum_{i=1}^{\mathcal{|A|}} \sigma_{i} $ refers the total demand, and $\sigma^{s}$ refers total influence supply, i.e., $\sigma^{s} = \sum_{b \in \mathcal{BS}} \sigma(b)$. 

\paragraph{Average Individual Demand Ratio $(\beta)$} It is the percentage of average individual demand of advertisers over influence supply, i.e., $\beta = \frac{\sigma^{d} / |\mathcal{A}|}{\sigma^{s}}$. By controlling the value of $\beta$, we can control the influence demands of the advertisers.
\paragraph{Advertiser Influence Demand $(\sigma)$}
Each advertiser’s demand is generated using parameters $\alpha$ and $\beta$, and is computed as $\sigma = \lfloor \psi \cdot \sigma^{s} \cdot \beta \rfloor$, where $\psi \in [0.8, 1.2]$ is randomly selected to introduce variability in advertiser demand.
\paragraph{Advertiser Tag $(\mathcal{T})$} We consider each advertiser have at least $100$ and at most $500$ tags for all the experiments.
\paragraph{Advertiser Payment $(\lambda)$}
Following the widely adopted setting \cite{aslay2015viral,aslay2017revenue}, one can formulate the advertiser's payment as $\lambda = \lfloor \eta \cdot \sigma \rfloor$, where $\eta \in [0.9,1.1]$.
\paragraph{Billboard Slot Cost}
All OOH advertising companies, like LAMAR, do not disclose the exact rental price for a slot; however, the literature \cite{zhang2020towards,ali2024regret,zhang2021minimizing} treats the slot cost as proportional to its influence. So, we follow the same setting as $\lfloor \tau \times \mathcal{I}(s)/10\rfloor$, where $\mathcal{I}(s)$ is the influence of slot $s$ and $\tau$ is randomly chosen between $0.9$ to $1.1$ to simulate fluctuation.

\paragraph{Penalty Ratio $(\delta)$}
The penalty ratio $\delta \in [0,1]$ controls the fraction of payment when the advertisers are not satisfied.
% \paragraph{Sampling Parameter $(\epsilon)$}

\paragraph{Environment Setup}
All algorithms are implemented in C++ using Visual Studio Code. Experiments are conducted on an Ubuntu-based desktop system with 64 GB RAM and a Xeon(R) 3.50 GHz processor.
% %%%%%%%%%%%%%%%%%%%%%%%%%%%%%   Algorithm Vs. Regret NYC(1-3) %%%%%%%%%%%%%%%%%%%%%%%%%%%%%%%
\begin{figure*}[!ht]
\centering
    \begin{tabular}{lclc}
       Excessive Regret & \includegraphics[width=0.11\linewidth]{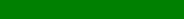} & Unsatisfied Regret & \includegraphics[width=0.11\linewidth]{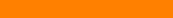} \\
    \end{tabular}
\setlength{\tabcolsep}{0.1pt} % tighter spacing between columns
\renewcommand{\arraystretch}{0.9} % tighter spacing between rows
\begin{tabular}{ccccc}
\includegraphics[scale=0.156]{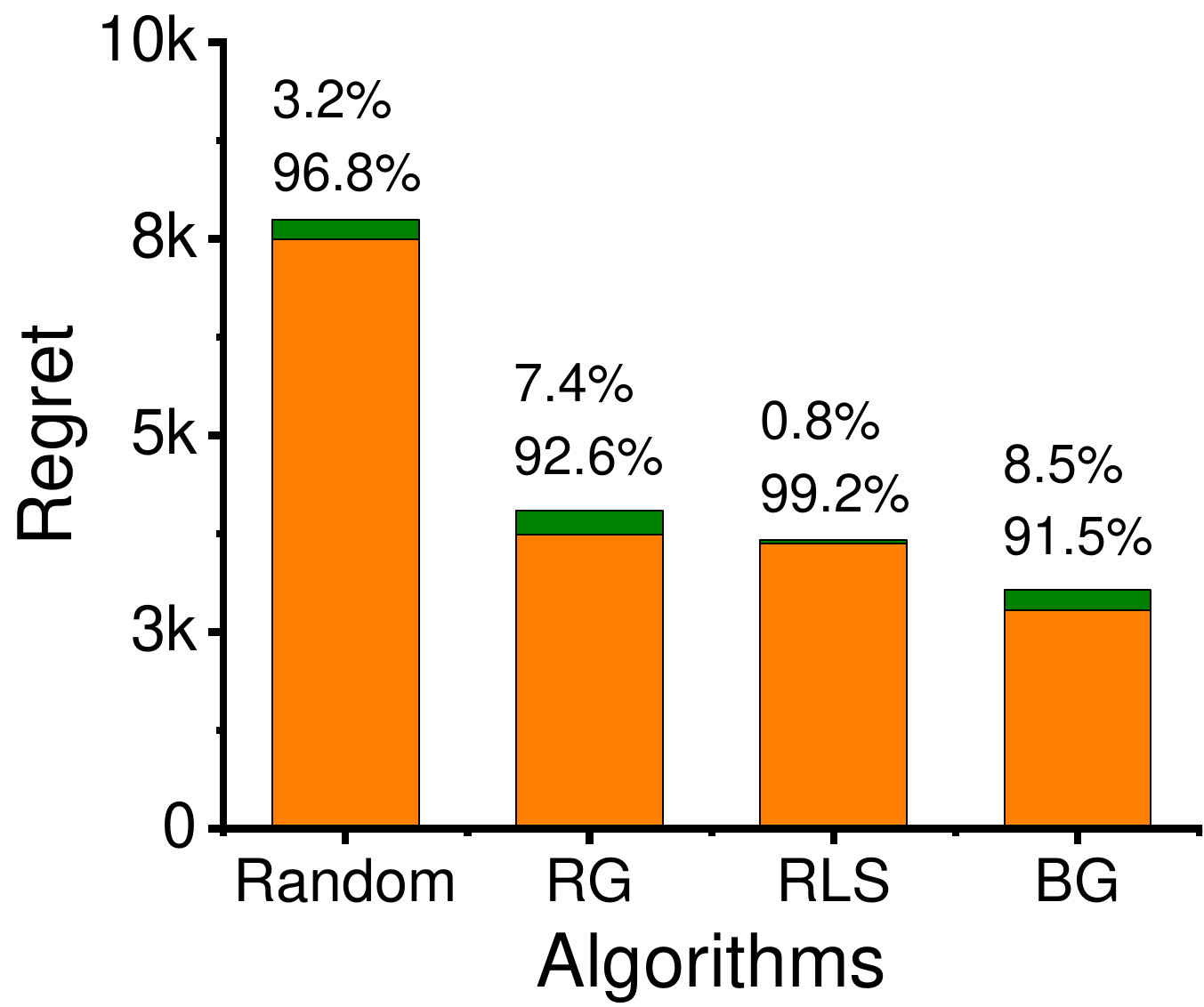} & \includegraphics[scale=0.156]{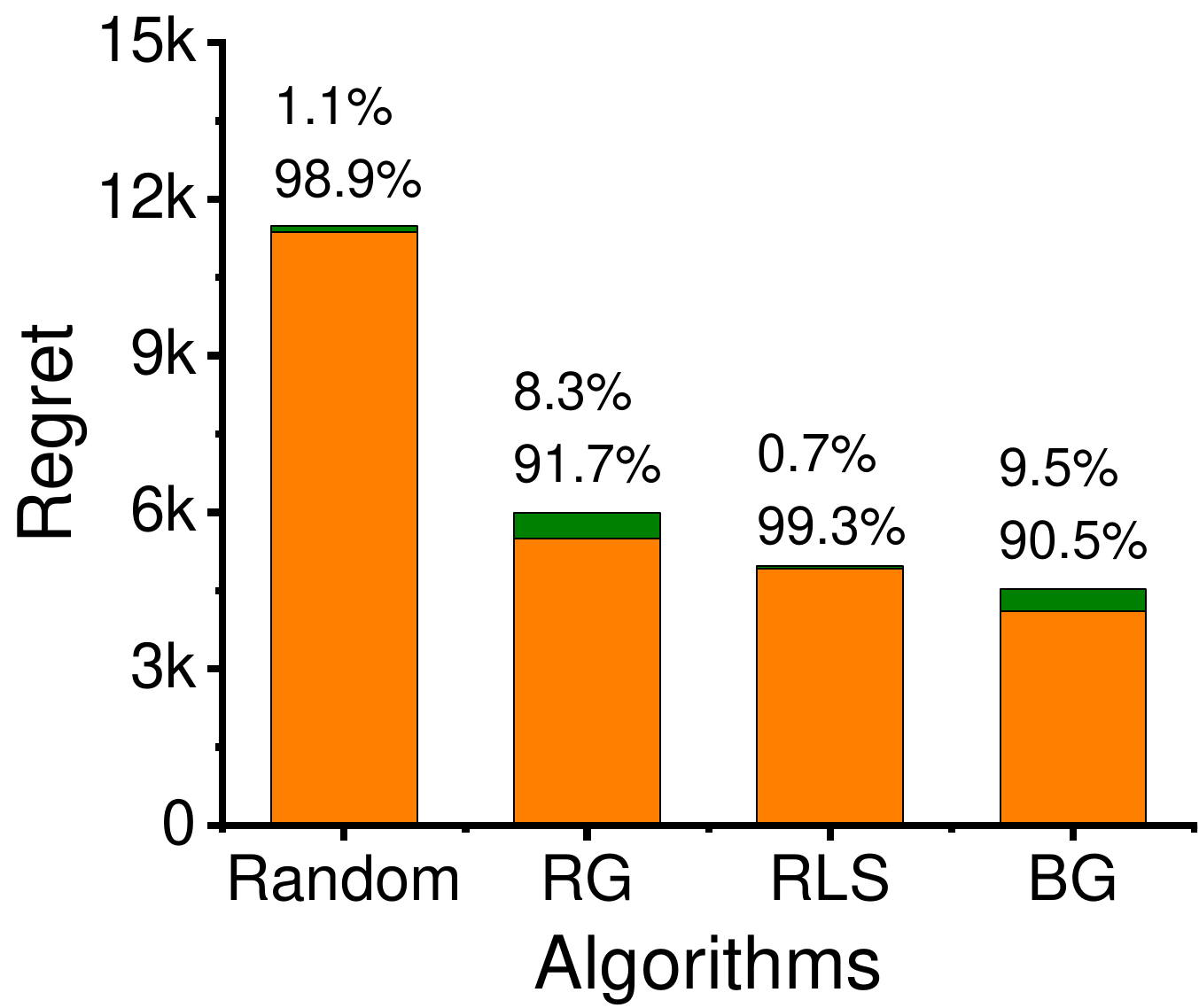} & \includegraphics[scale=0.156]{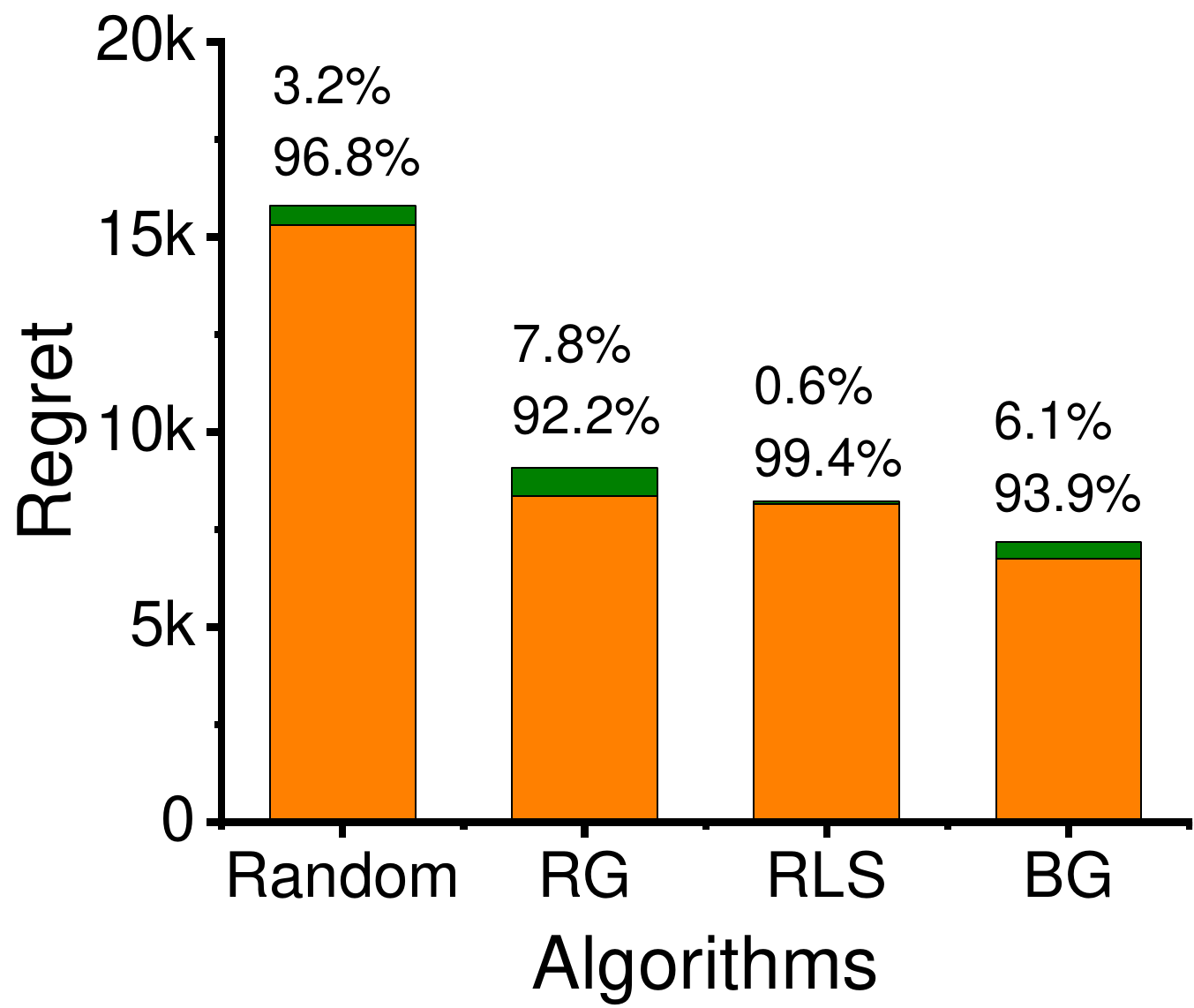} &
\includegraphics[scale=0.156]{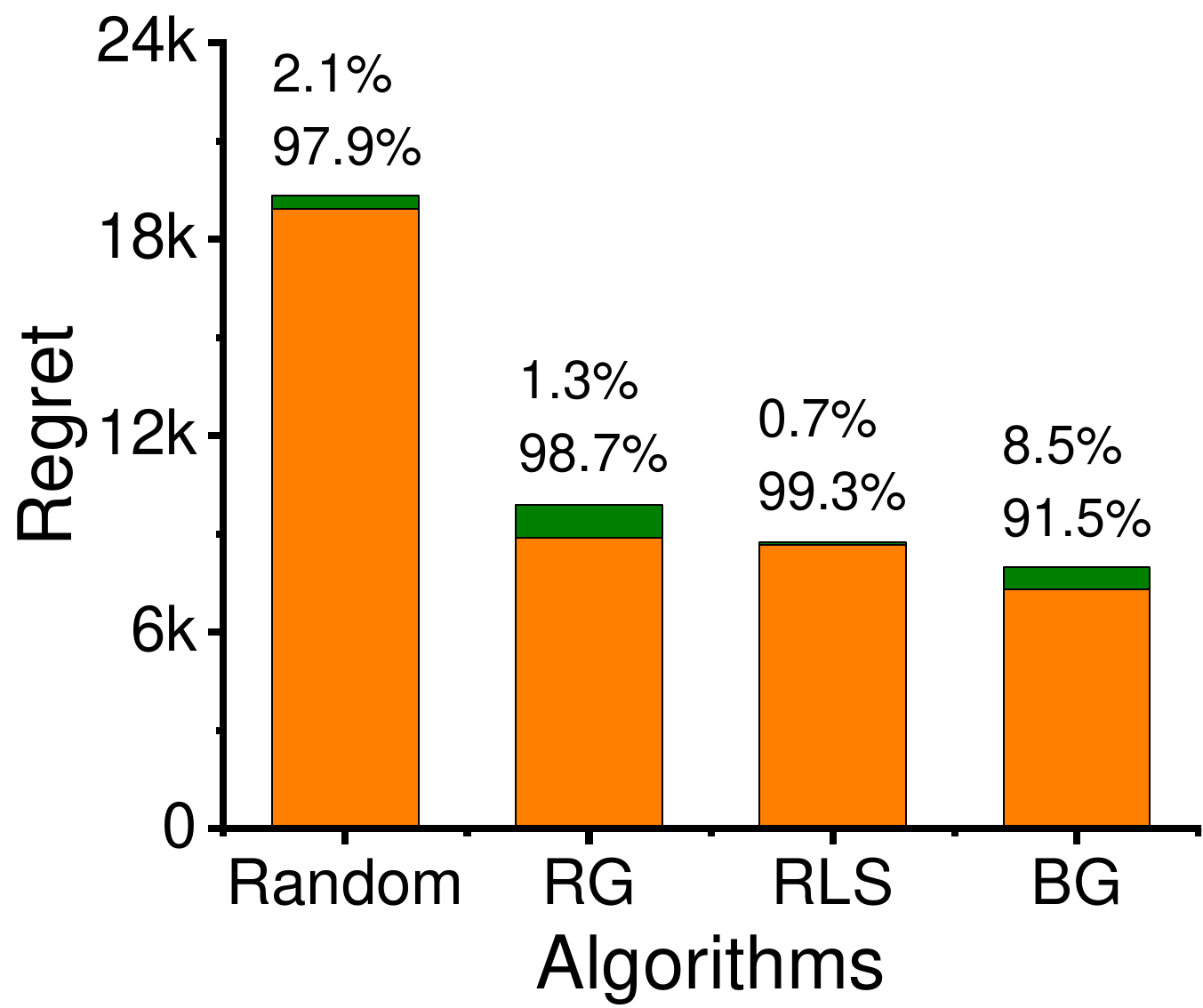} &
\includegraphics[scale=0.156]{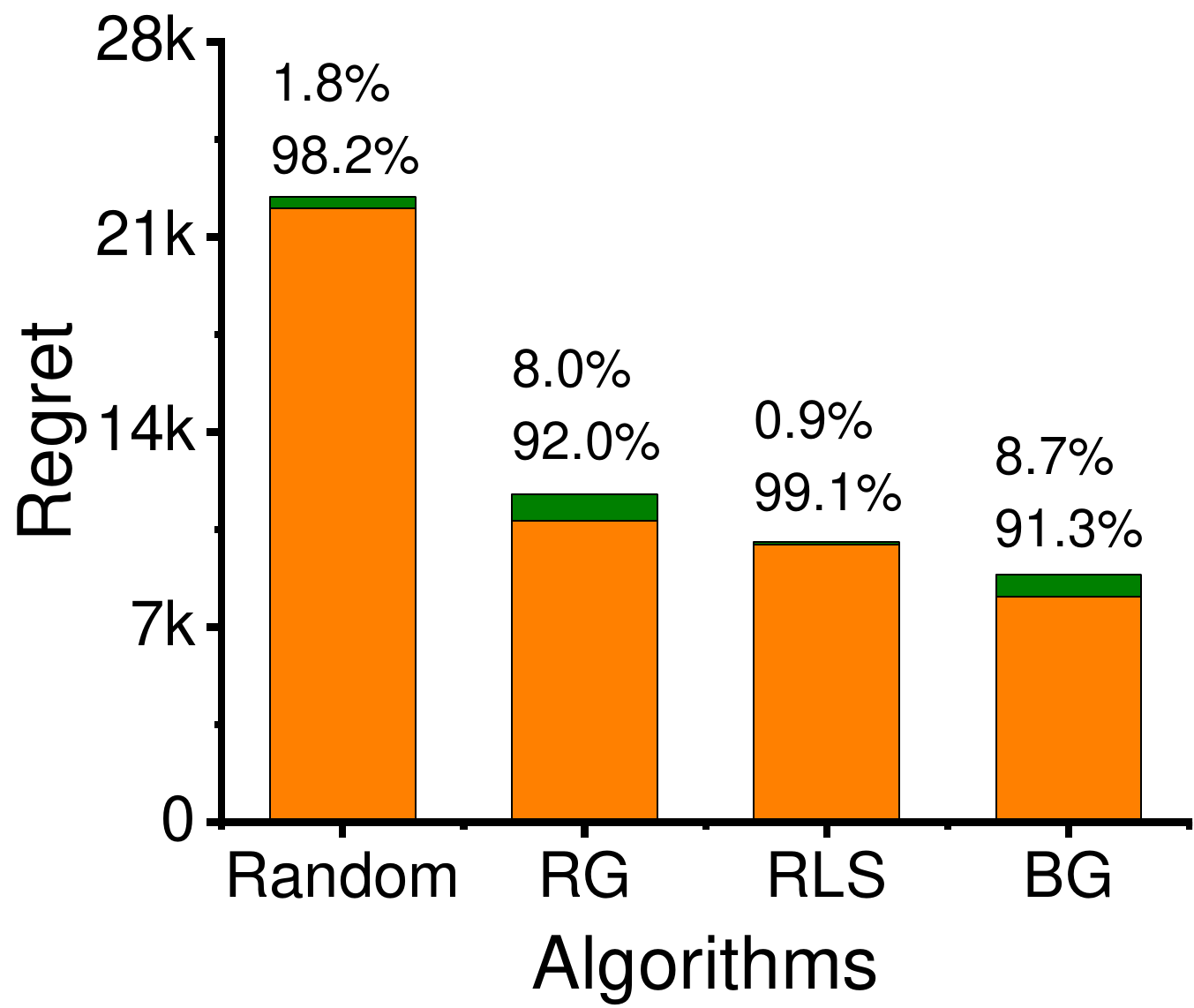}\\
\tiny{(a) $\alpha = 40\%$} &  \tiny{(b) $\alpha = 60\%$} & \tiny{(c) $\alpha = 80\%$} & \tiny{(d) $\alpha = 100\%$} & \tiny{(e) $\alpha = 120\%$}\\
\includegraphics[scale=0.156]{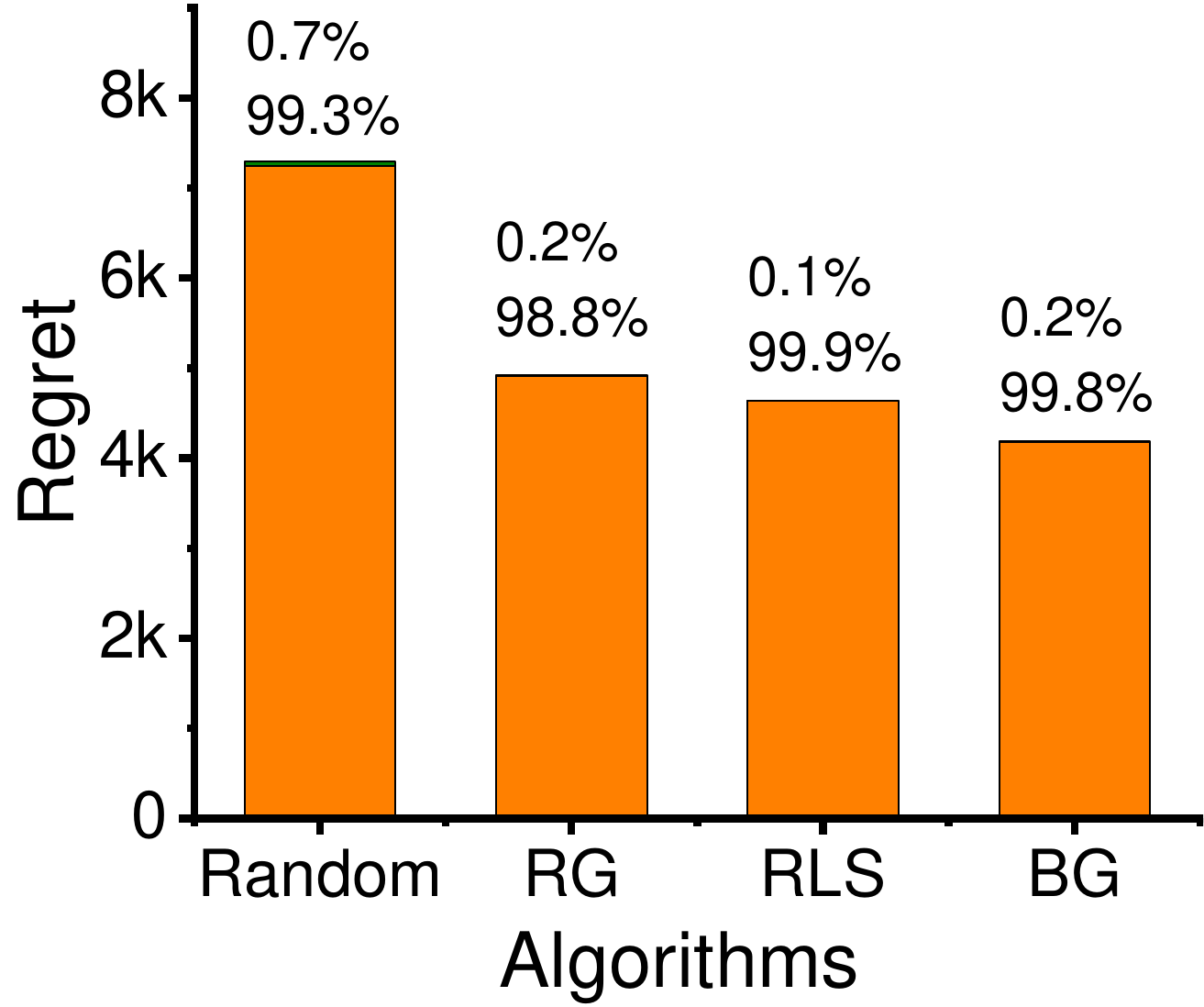} & \includegraphics[scale=0.156]{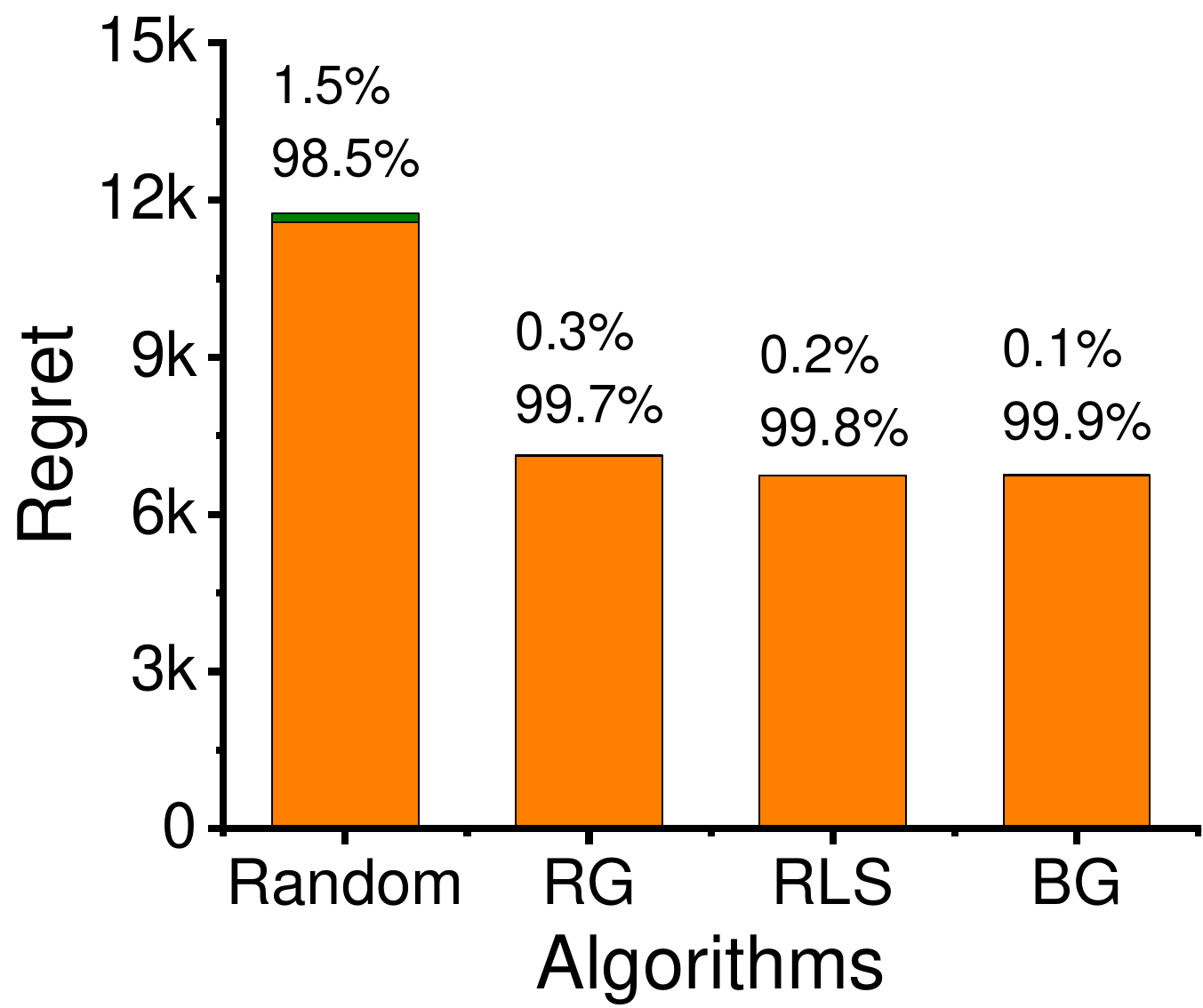} & \includegraphics[scale=0.156]{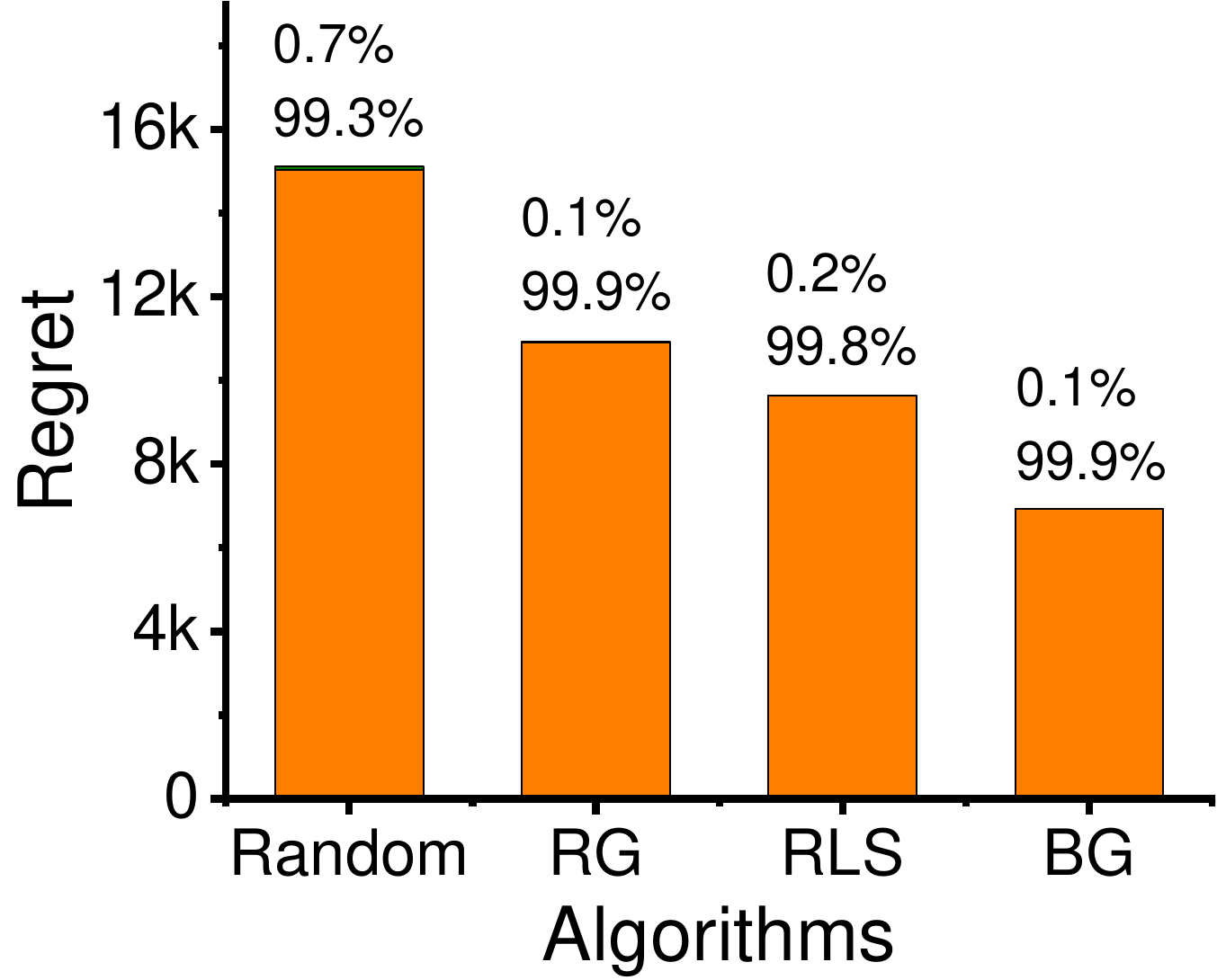} &
\includegraphics[scale=0.156]{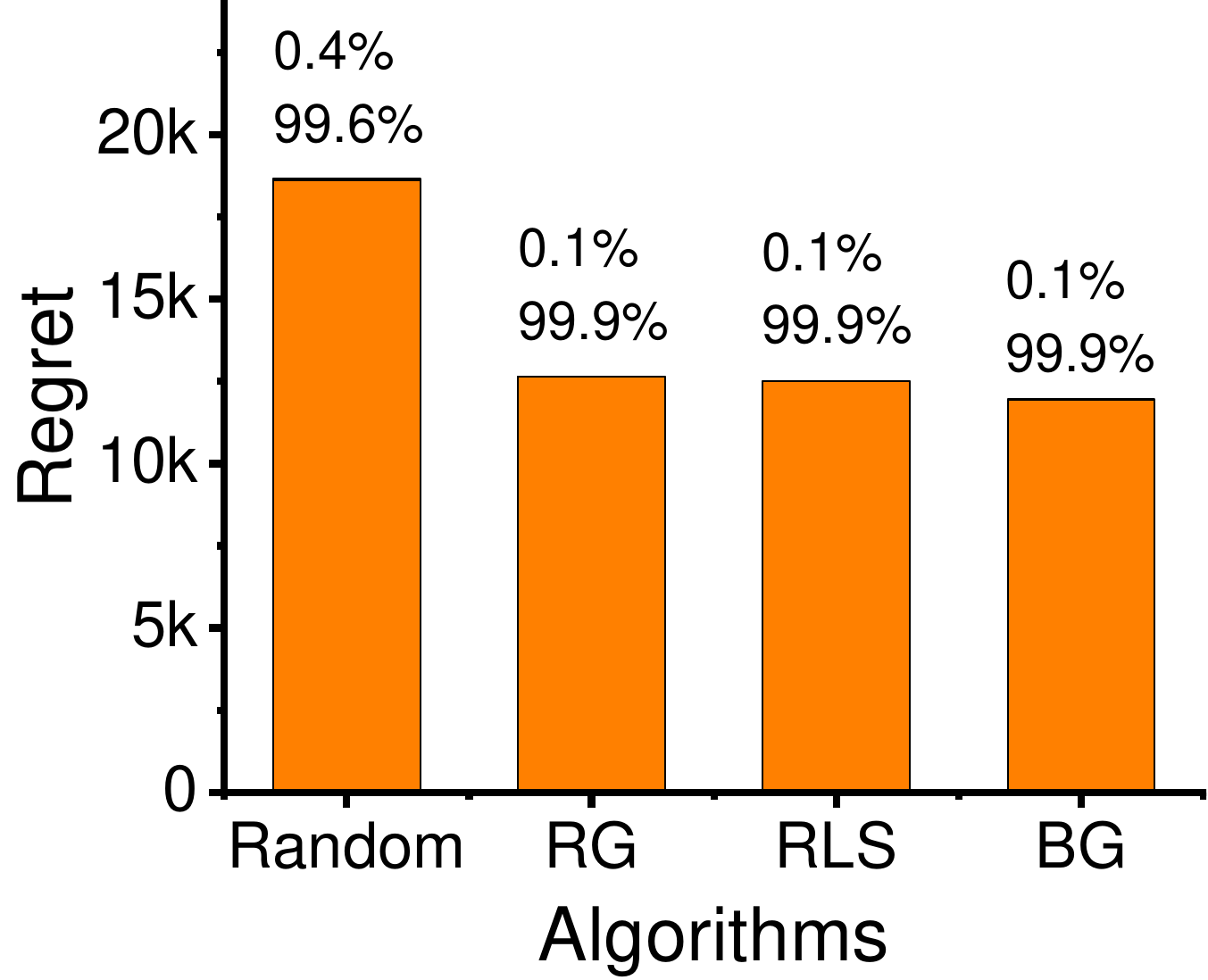} &
\includegraphics[scale=0.156]{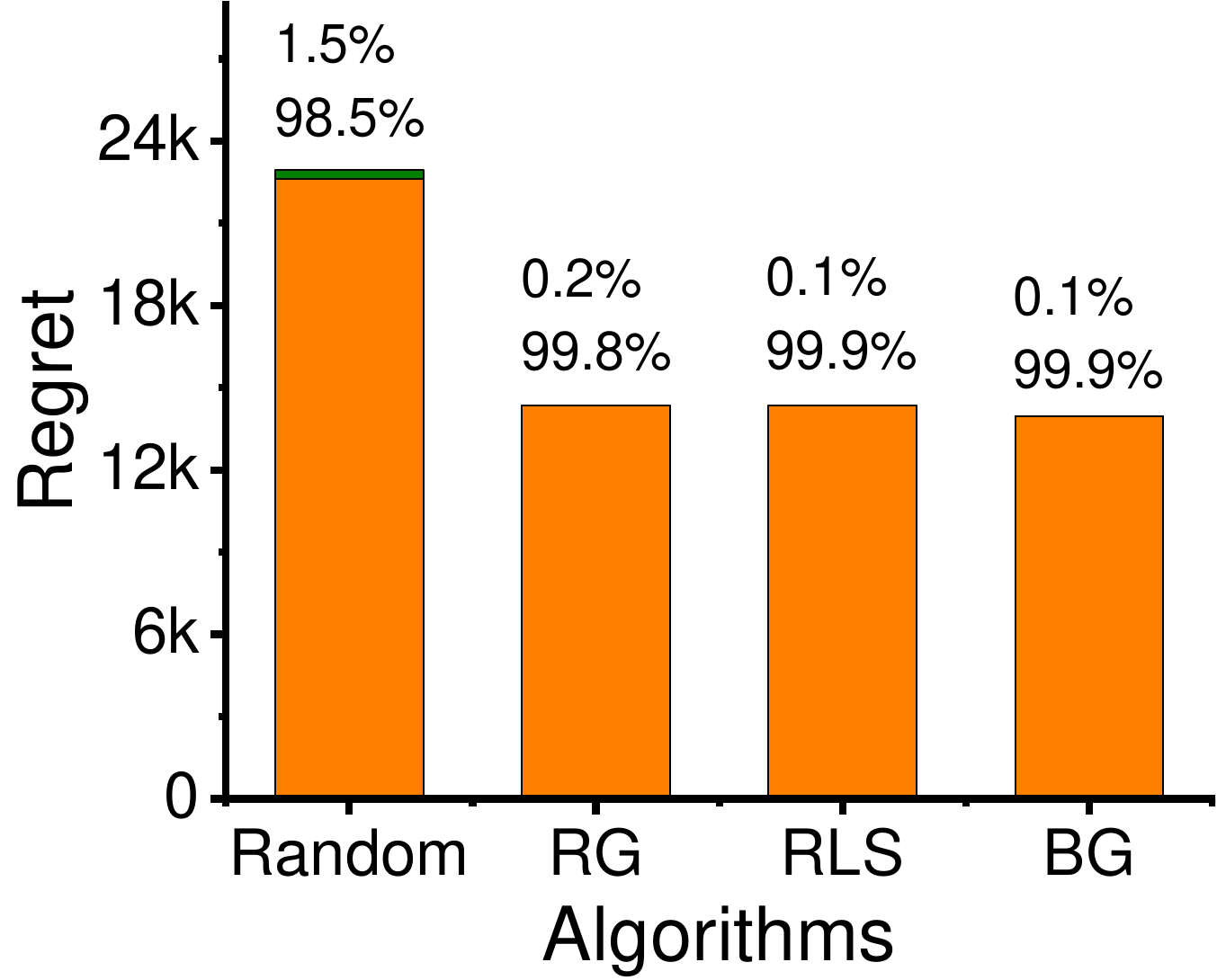}\\
\tiny{(f) $\alpha = 40\%$} &  \tiny{(g) $\alpha = 60\%$} & \tiny{(h) $\alpha = 80\%$} & \tiny{(i) $\alpha = 100\%$} & \tiny{(j) $\alpha = 120\%$}\\
\includegraphics[scale=0.156]{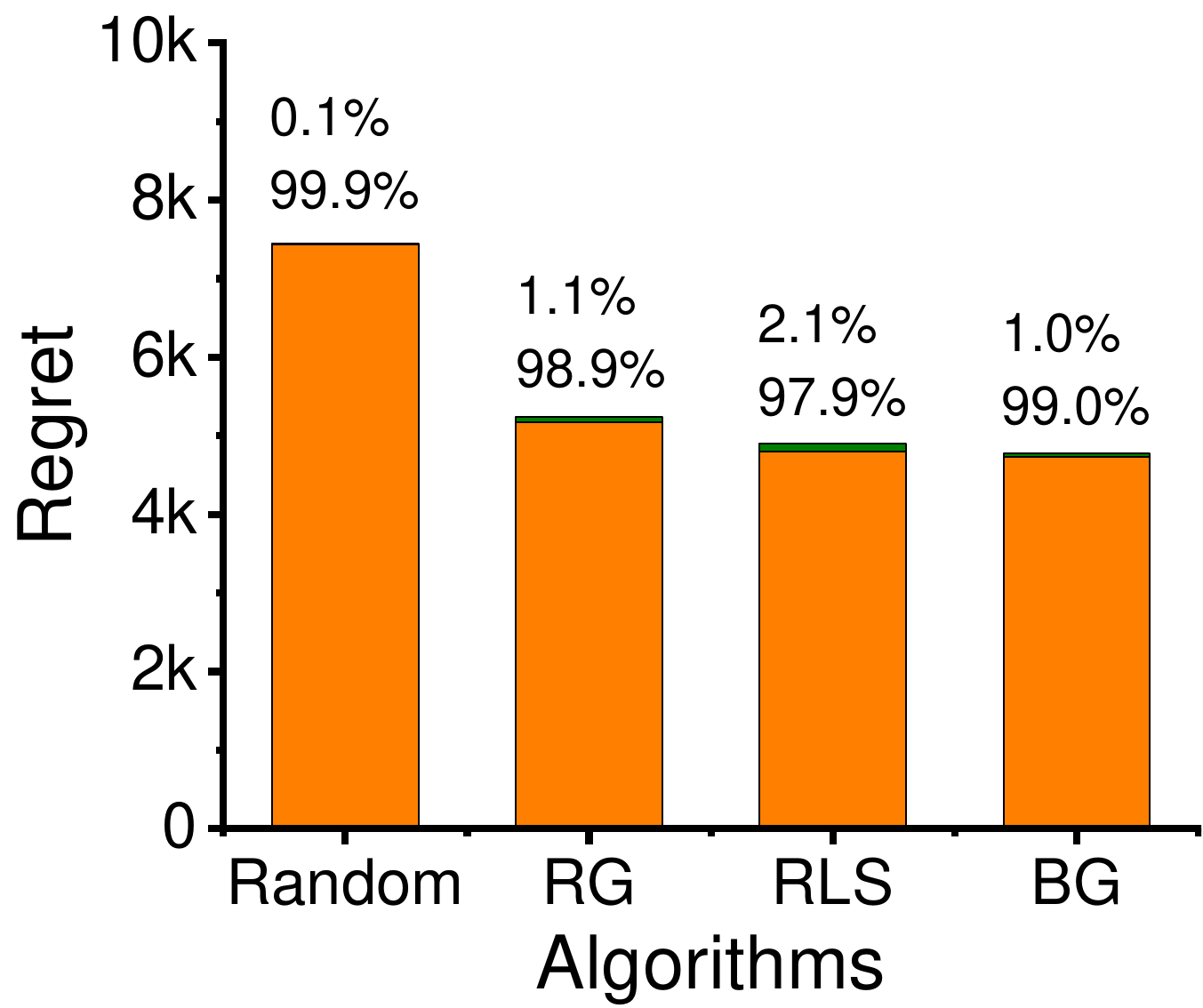} & \includegraphics[scale=0.156]{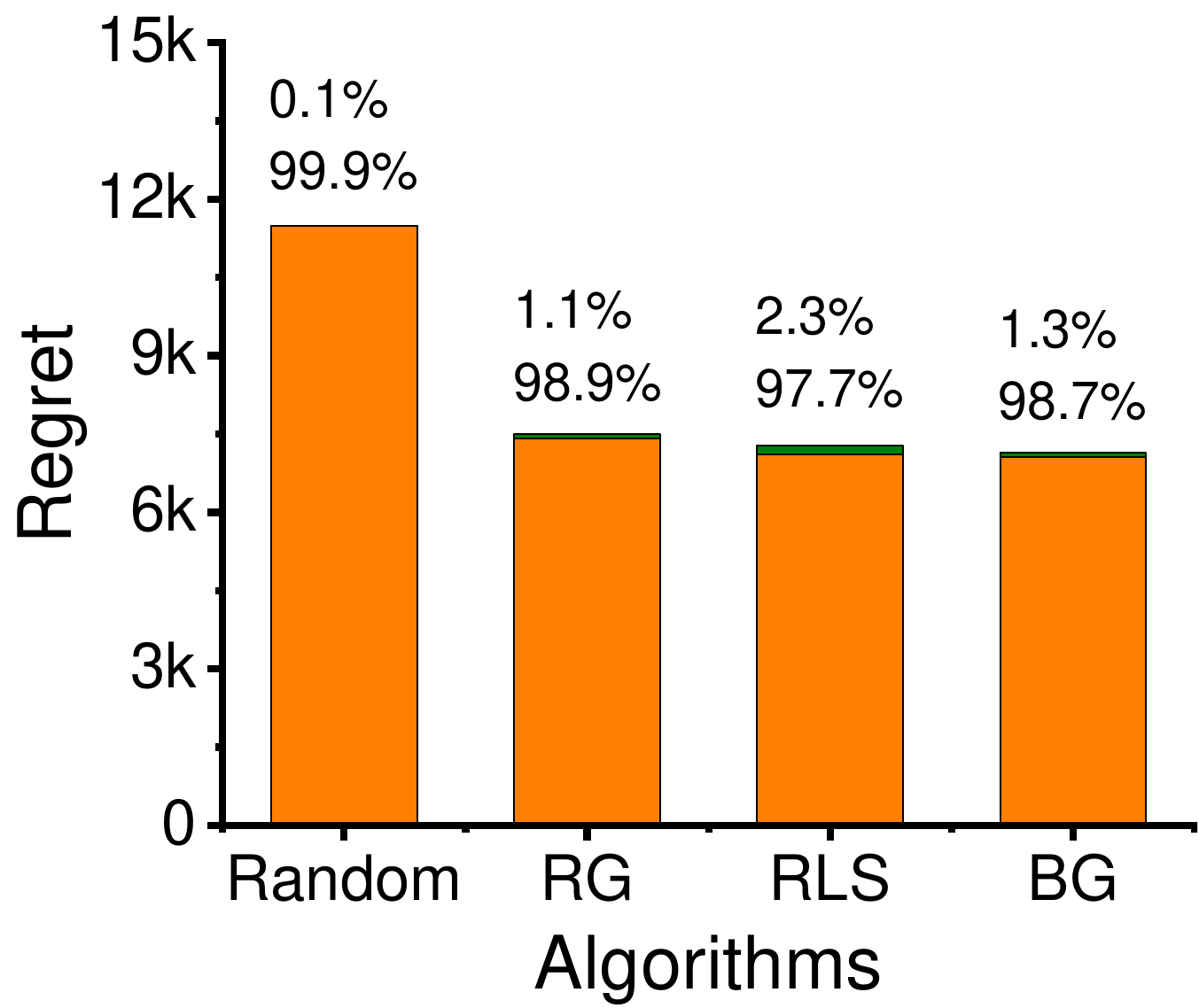} & \includegraphics[scale=0.156]{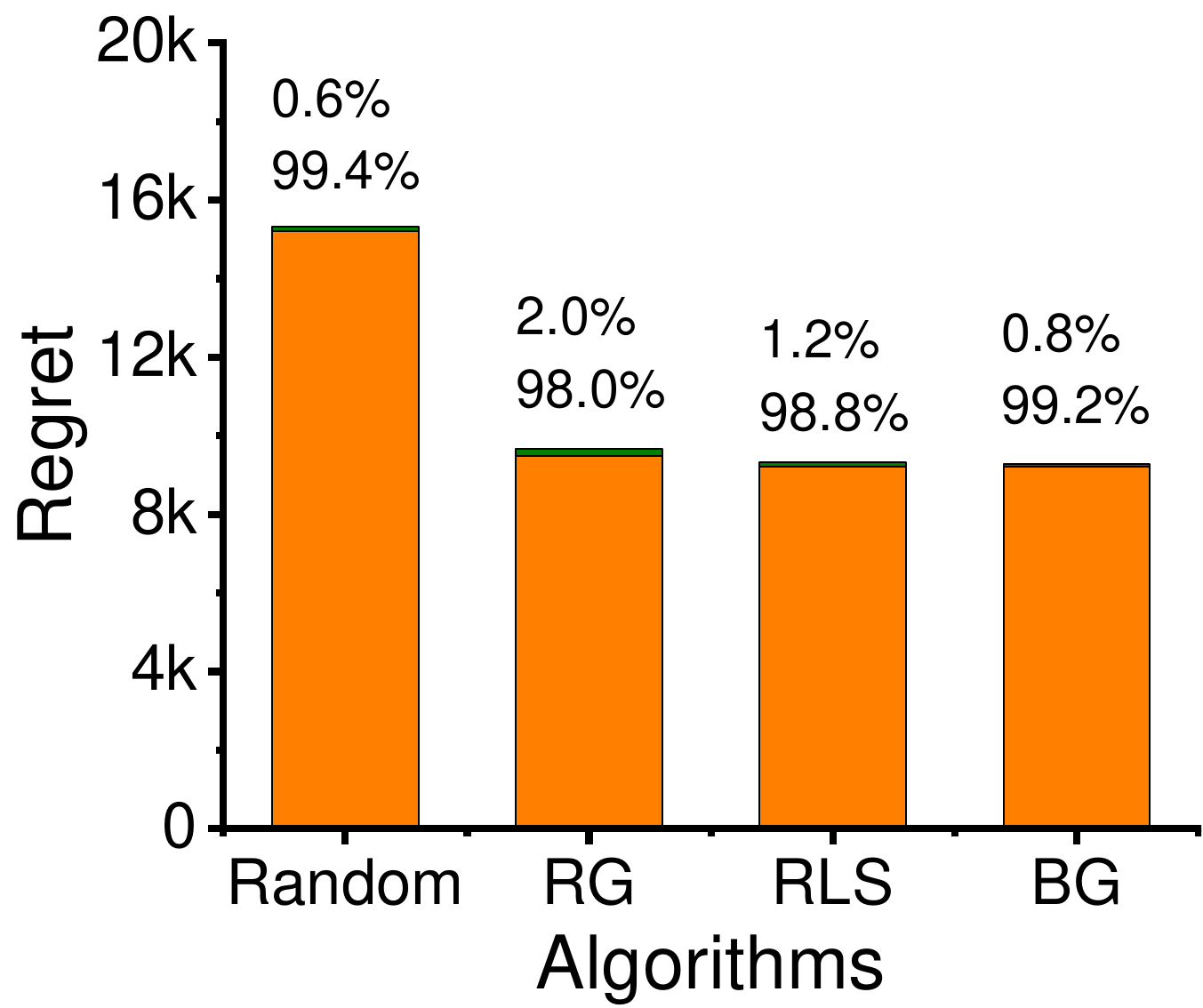} &
\includegraphics[scale=0.156]{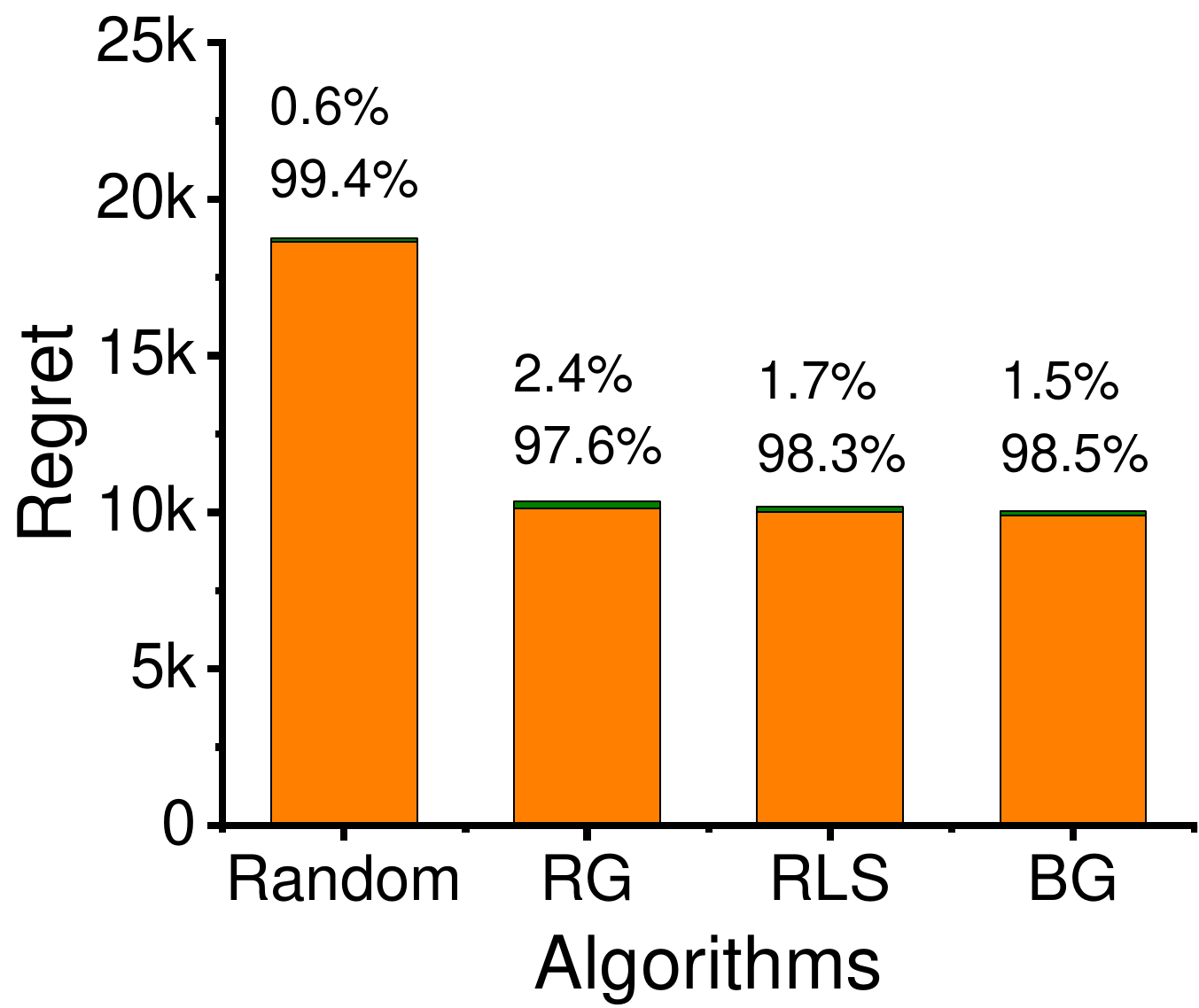} &
\includegraphics[scale=0.156]{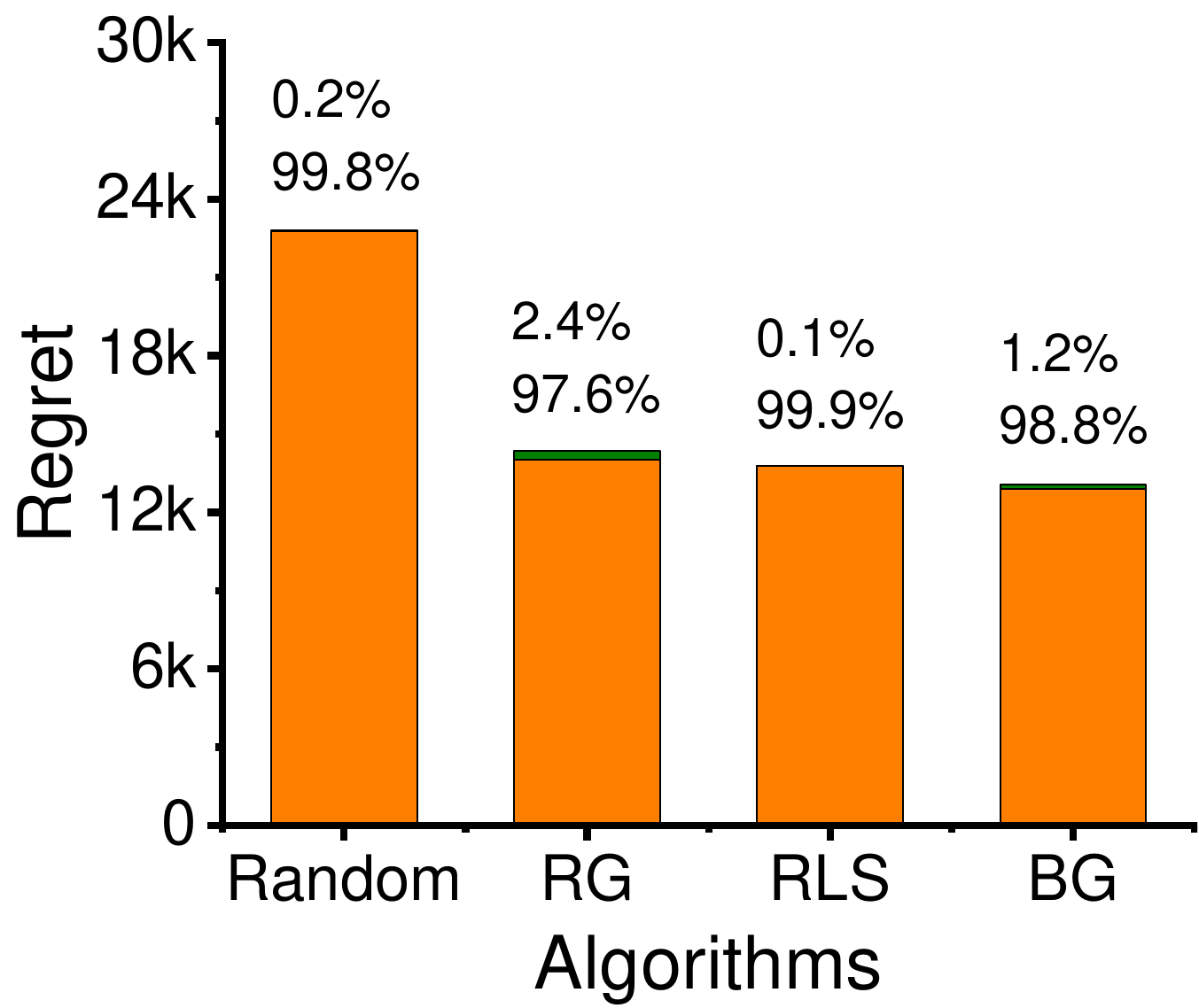}\\
\tiny{(k) $\alpha = 40\%$} &  \tiny{$(\ell)$ $\alpha = 60\%$} & \tiny{(m) $\alpha = 80\%$} & \tiny{(n) $\alpha = 100\%$} & \tiny{(o) $\alpha = 120\%$}\\
\end{tabular}
\caption{Varying Demand-supply ratio $\alpha$ when $\beta = 1\%, |\mathcal{A}| = 100$ $(a,b,c,d,e)$, when $\beta = 2\%, |\mathcal{A}| = 50$ $(f,g,h,i,j)$, when $\beta = 5\%, |\mathcal{A}| = 20$ $(k,\ell,m,n,o)$ (NYC Dataset)}
\label{Fig:Plot1-3NYC}
\end{figure*}
\subsection{\textbf{Experimental Results and Discussions.}} 
\subsubsection{\textbf{Experiments over NYC Dataset}}
The experimental discussion is categorized into four scenarios as follows. 
\paragraph{\textbf{Case 1. Low $\alpha$, Low $\beta$ of Figure \ref{Fig:Plot1-3NYC} (a,b,c,f,g,h,k,$\ell$,m)}} Corresponding to Case 1, we have a situation where $\alpha \leq 80\%$ and $\beta \leq 5\%$. This refers to a situation in which both global demand and individual advertiser demand are low. The influence provider has a small number of advertisers with low influence demands. We have three main observations. \textbf{First}, when $\alpha$ increases, the excessive influence of all the proposed approaches decreases, and unsatisfied regret increases. This happens because, as $\alpha$ increases, advertisers' demand rises, and since most advertisers are satisfied early, this leads to excessive regret. \textbf{Second}, among the proposed approaches, `BG' outperforms `RG' and `RLS' in terms of both excessive and unsatisfied regret. This happens because the `RG' and `RLS' use randomization when allocating slots to advertisers. \textbf{Third}, the `Random' approach incurs the highest total regret compared to the proposed `BG',`RG', and `RLS'. However, the demand-supply ratio $(\alpha)$ is very low, and the influence provider cannot satisfy all advertisers' demand because advertisers have tag-specific influence demand.
%%%%%%%%%%%%%%%%%%%%%%%%%%%%%   Algorithm Vs. Regret NYC %%%%%%%%%%%%%%%%%%%%%%%%%%%%%%%
\begin{figure*}[h!]
\centering
    \begin{tabular}{lclc}
       Excessive Regret & \includegraphics[width=0.11\linewidth]{Unsatisfied.png} & Unsatisfied Regret & \includegraphics[width=0.11\linewidth]{Excessive.png} \\
    \end{tabular}
\setlength{\tabcolsep}{0.1pt} % tighter spacing between columns
\renewcommand{\arraystretch}{0.9} % tighter spacing between rows
\begin{tabular}{ccccc}
\includegraphics[scale=0.156]{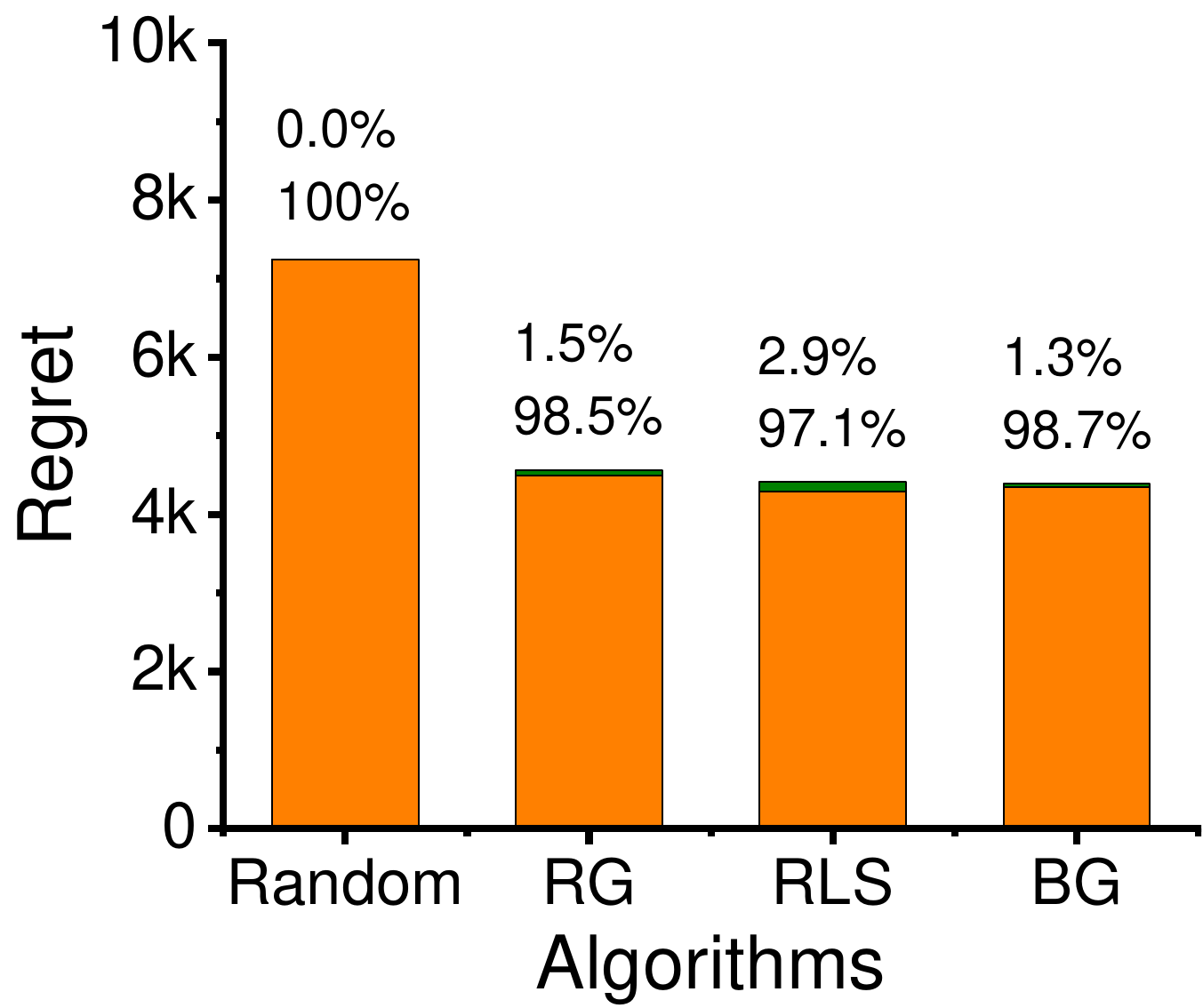} & \includegraphics[scale=0.156]{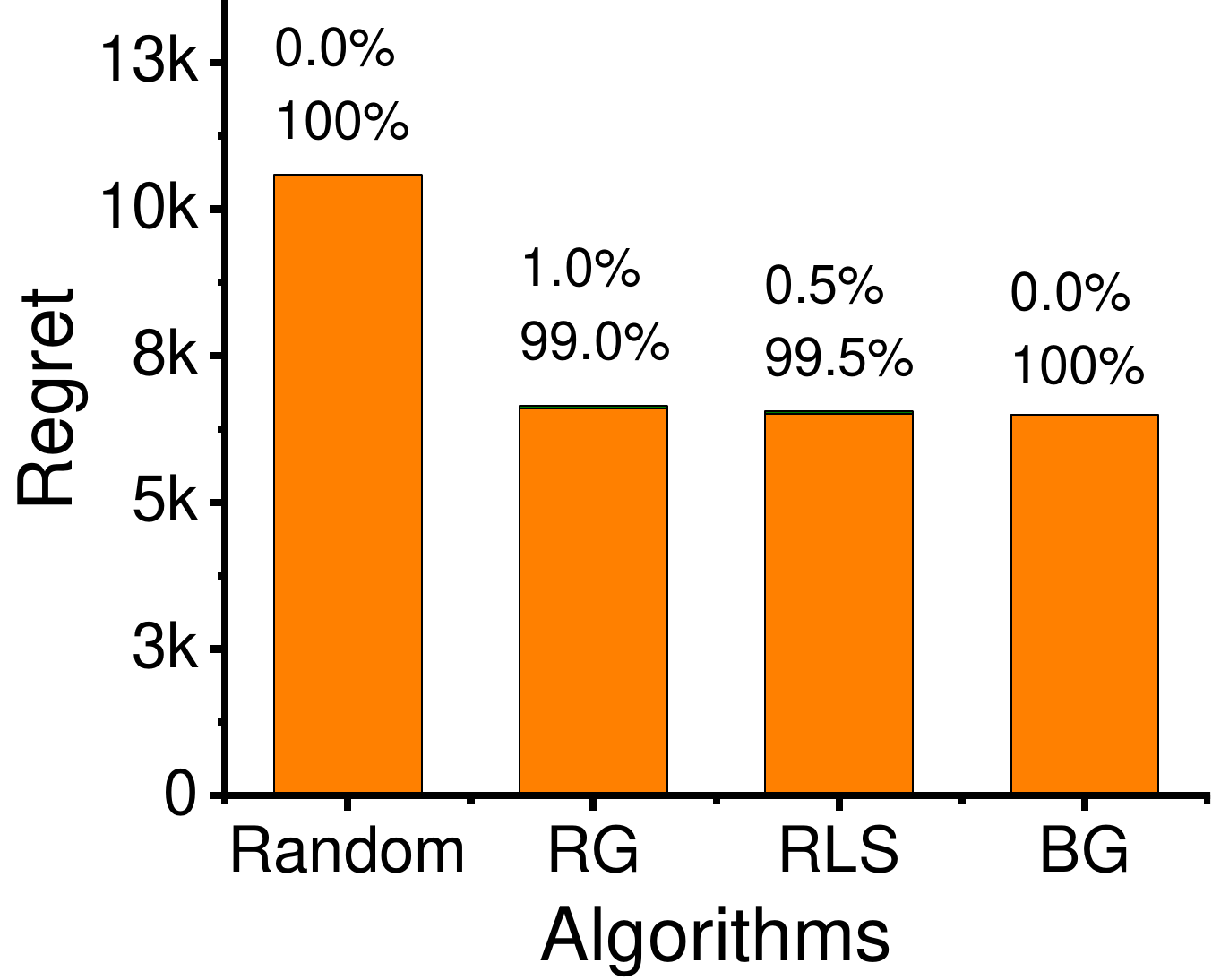} & \includegraphics[scale=0.156]{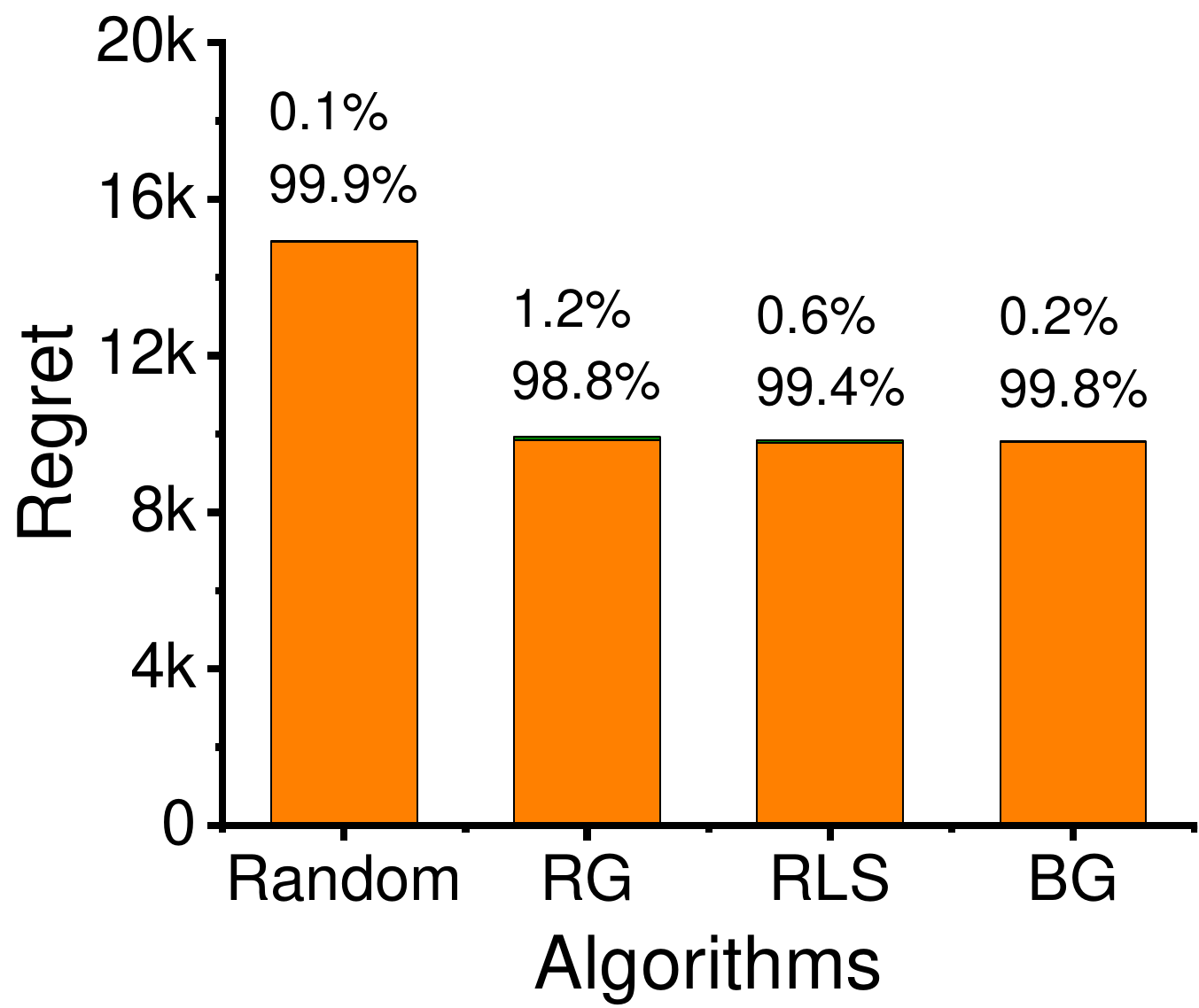} &
\includegraphics[scale=0.156]{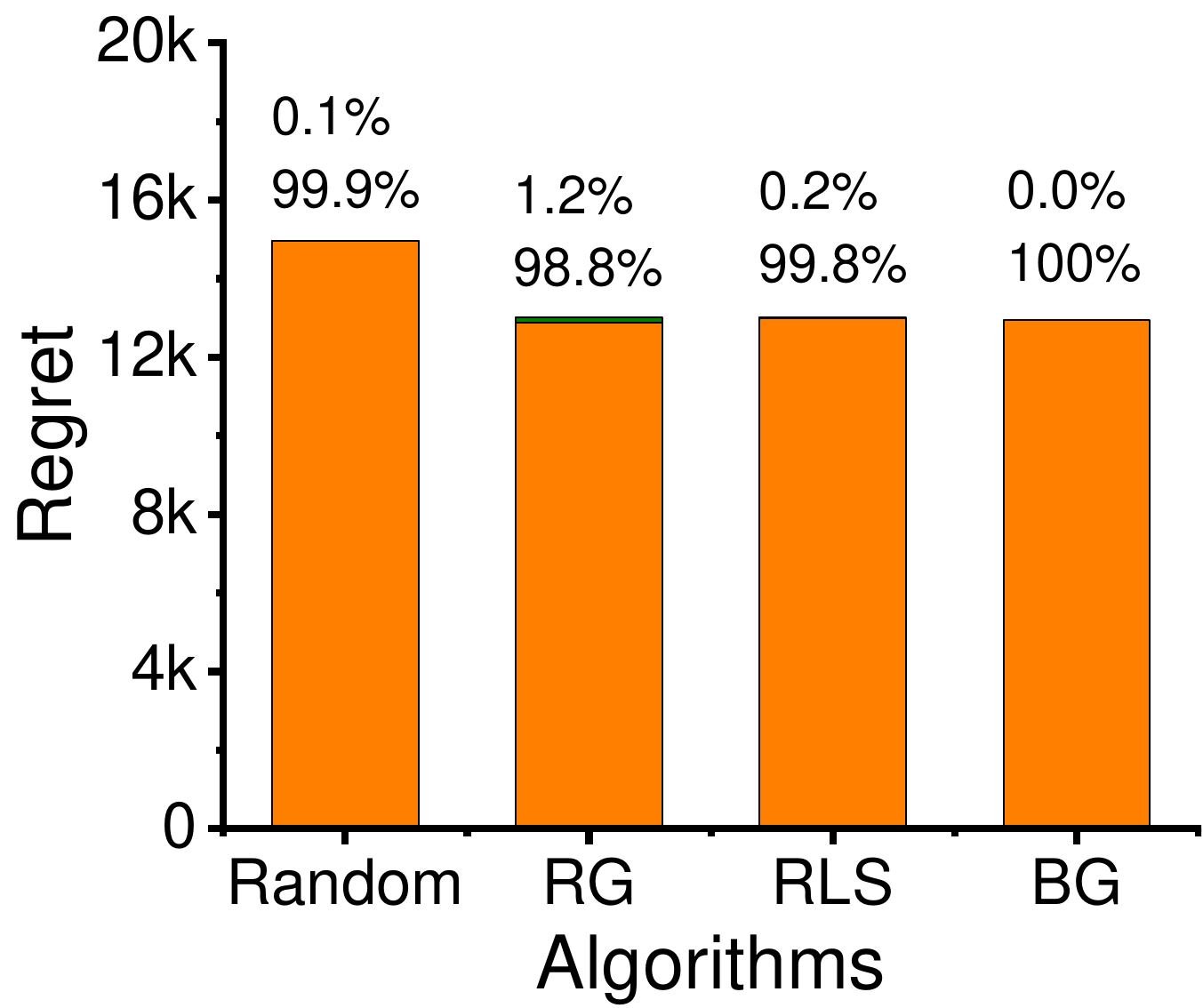} &
\includegraphics[scale=0.156]{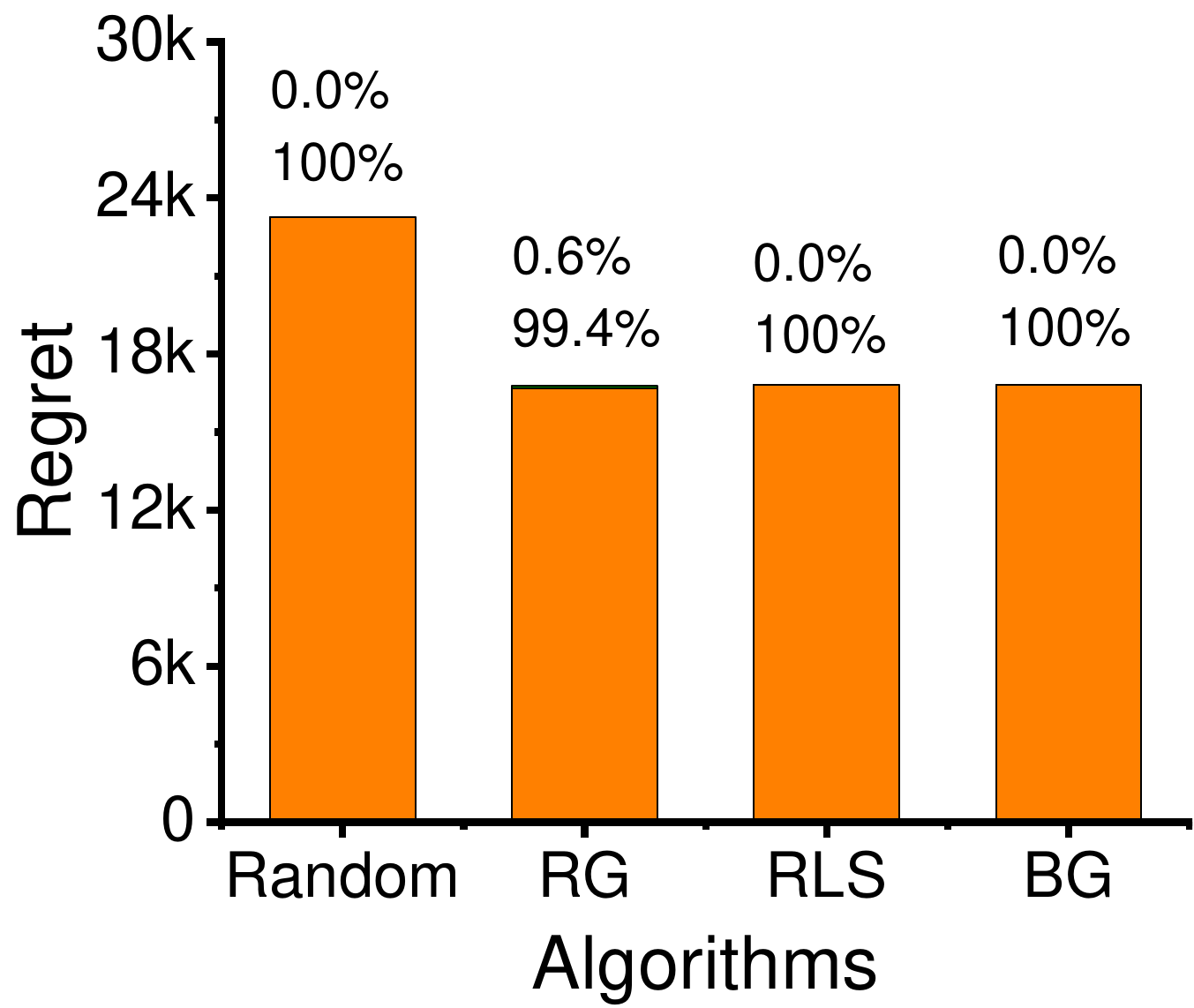}\\
\tiny{(a) $\alpha = 40\%$} &  \tiny{(b) $\alpha = 60\%$} & \tiny{(c) $\alpha = 80\%$} & \tiny{(d) $\alpha = 100\%$} & \tiny{(e) $\alpha = 120\%$}\\
\includegraphics[scale=0.156]{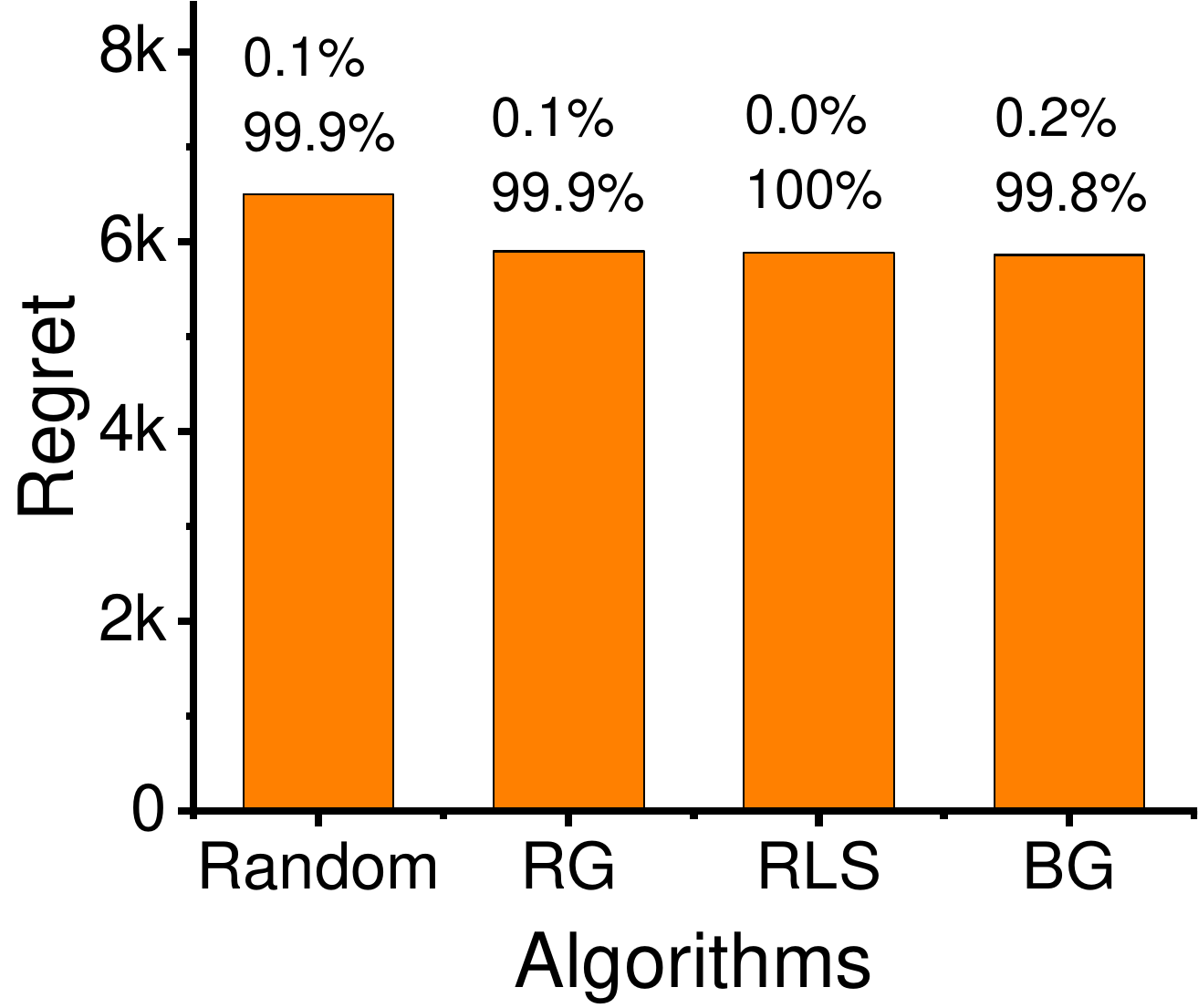} & \includegraphics[scale=0.156]{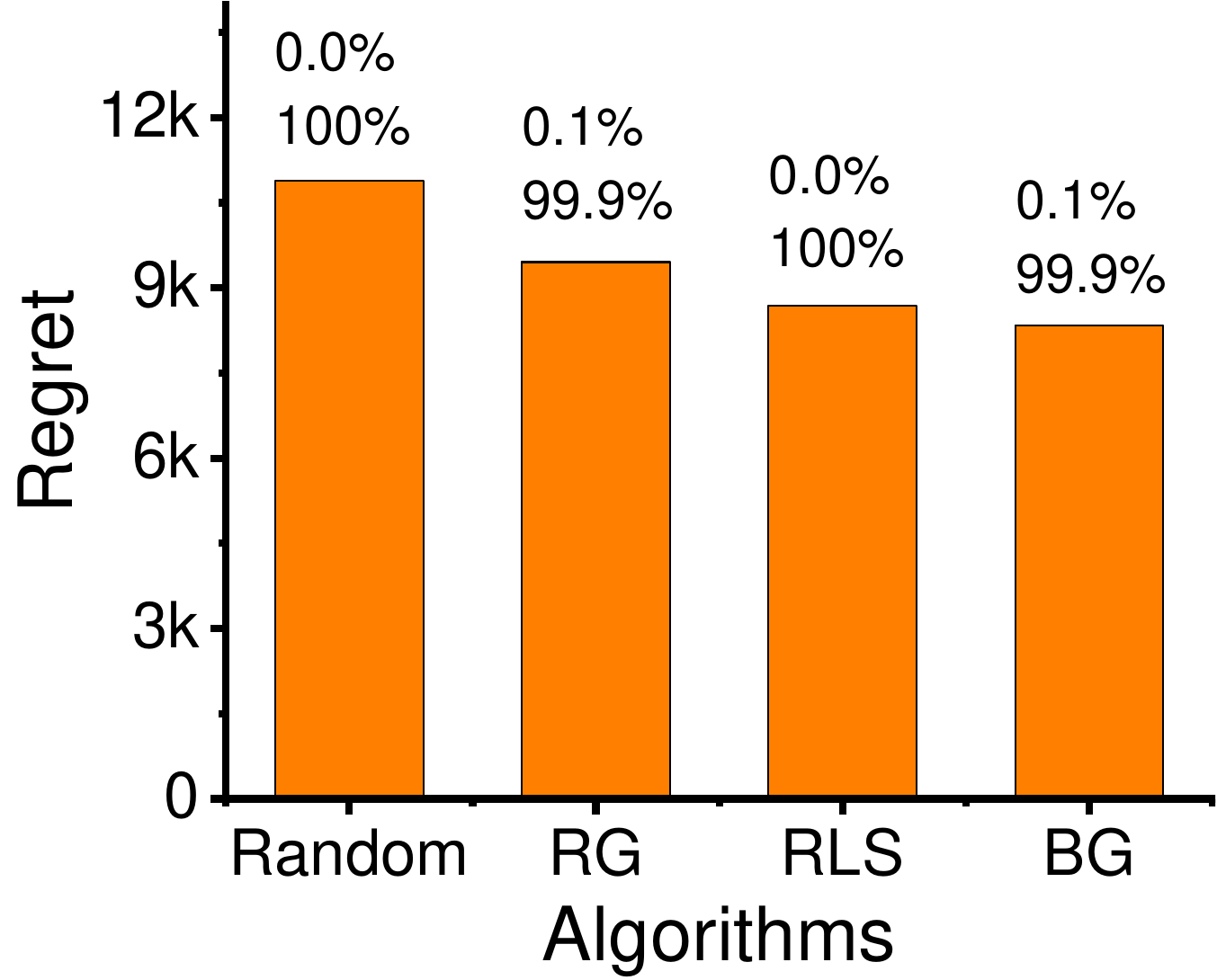} & \includegraphics[scale=0.156]{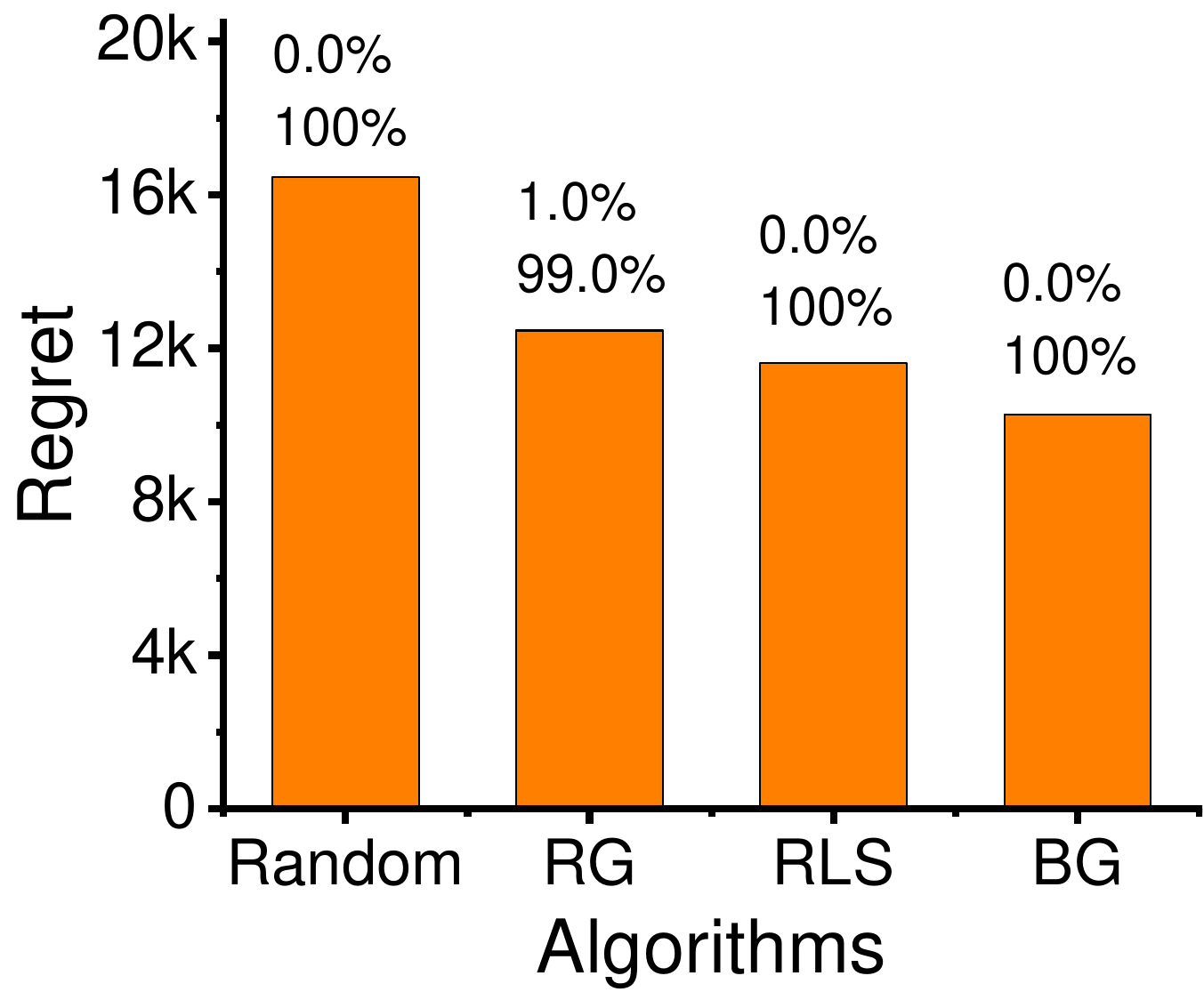} &
\includegraphics[scale=0.156]{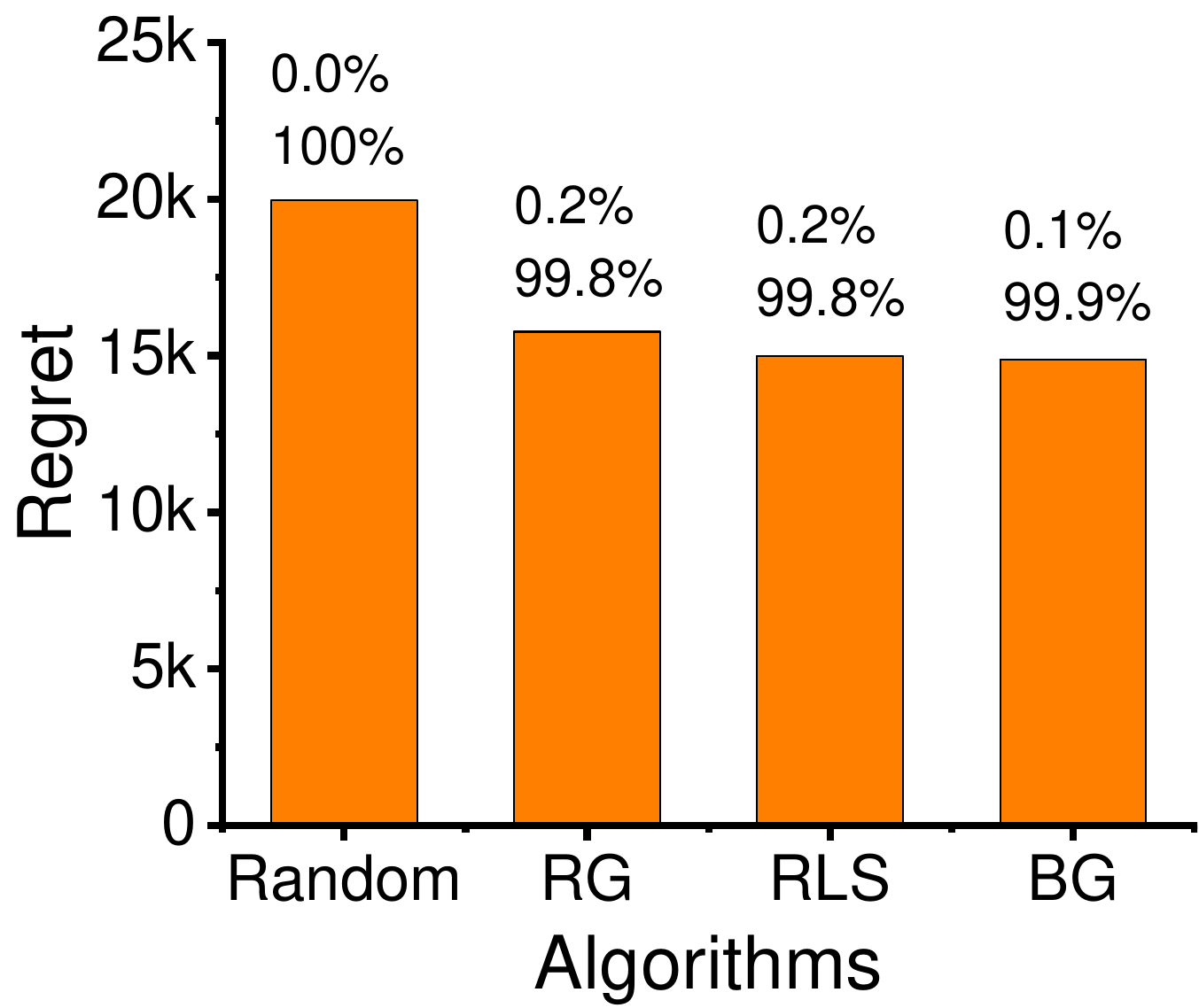} &
\includegraphics[scale=0.156]{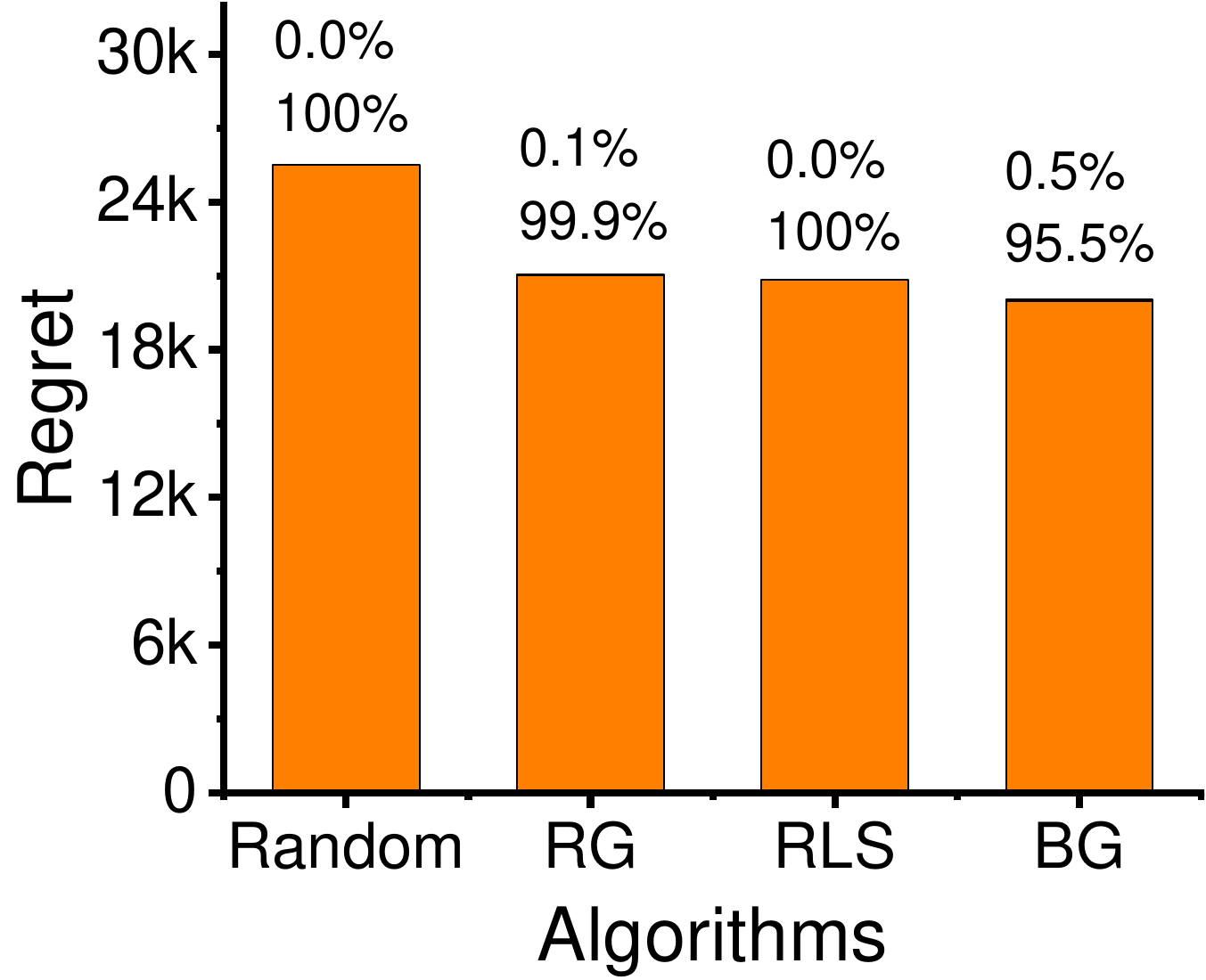}\\
\tiny{(f) $\alpha = 40\%$} &  \tiny{(g) $\alpha = 60\%$} & \tiny{(h) $\alpha = 80\%$} & \tiny{(i) $\alpha = 100\%$} & \tiny{(j) $\alpha = 120\%$}\\
\end{tabular}
\caption{Varying Demand-supply ratio, $\alpha$ when $\beta = 10\%, |\mathcal{A}| = 10$, when $\beta = 20\%, |\mathcal{A}| = 5$ (NYC Dataset)}
\label{Fig:Plot4-5NYC}
\end{figure*}
%%%%%%%%%%%%%%%%%%%%%%%%%%%%%   Algorithm Vs. Regret LA %%%%%%%%%%%%%%%%%%%%%%%%%%%%%%%
\begin{figure*}[h!]
\centering
    \begin{tabular}{lclc}
       Excessive Regret & \includegraphics[width=0.11\linewidth]{Unsatisfied.png} & Unsatisfied Regret & \includegraphics[width=0.11\linewidth]{Excessive.png} \\
    \end{tabular}
\setlength{\tabcolsep}{0.1pt} % tighter spacing between columns
\renewcommand{\arraystretch}{0.9} % tighter spacing between rows
\begin{tabular}{ccccc}
\includegraphics[scale=0.156]{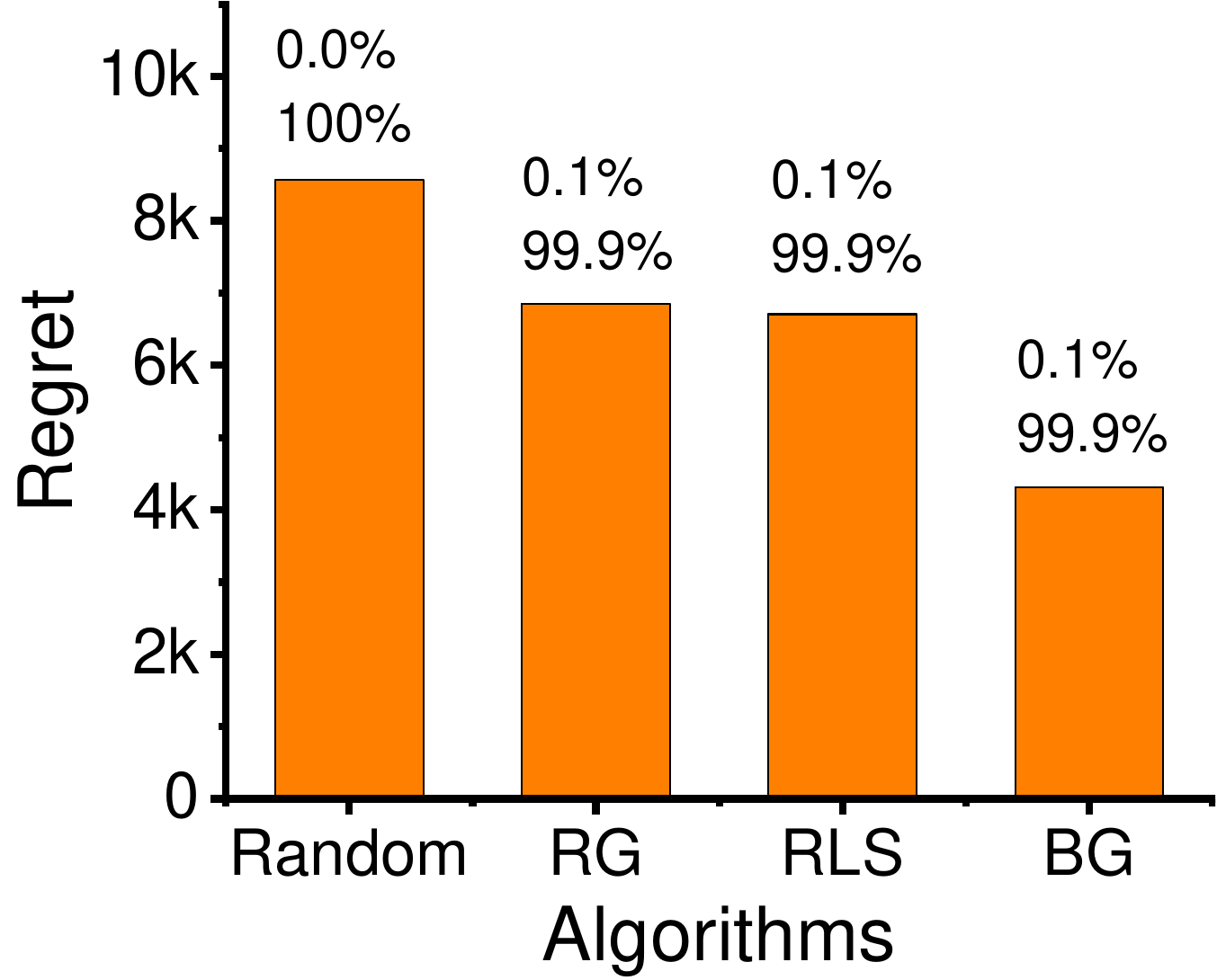} & \includegraphics[scale=0.156]{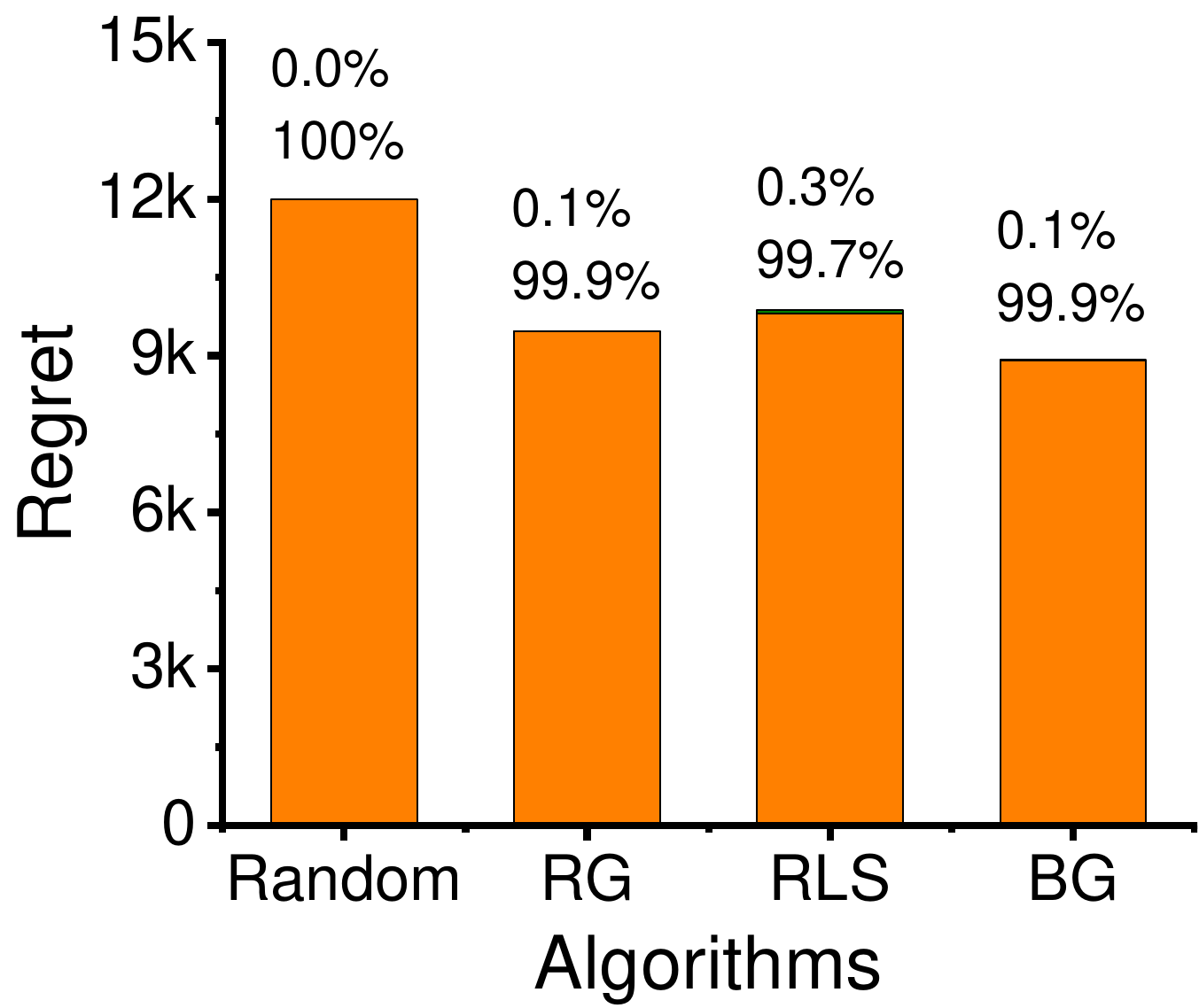} & \includegraphics[scale=0.156]{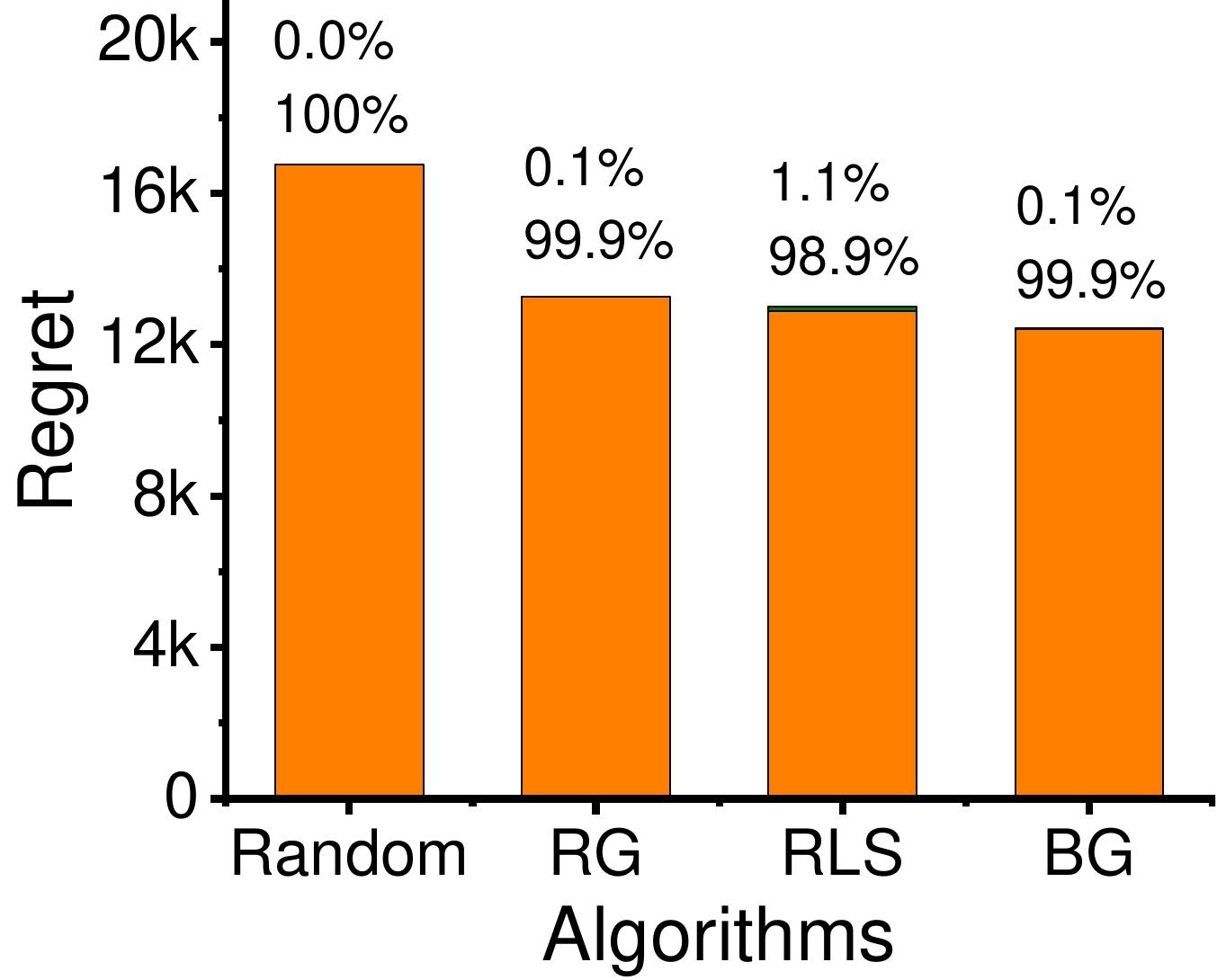} &
\includegraphics[scale=0.156]{Algo-Vs-Regret/LA/EXP14.pdf} &
\includegraphics[scale=0.156]{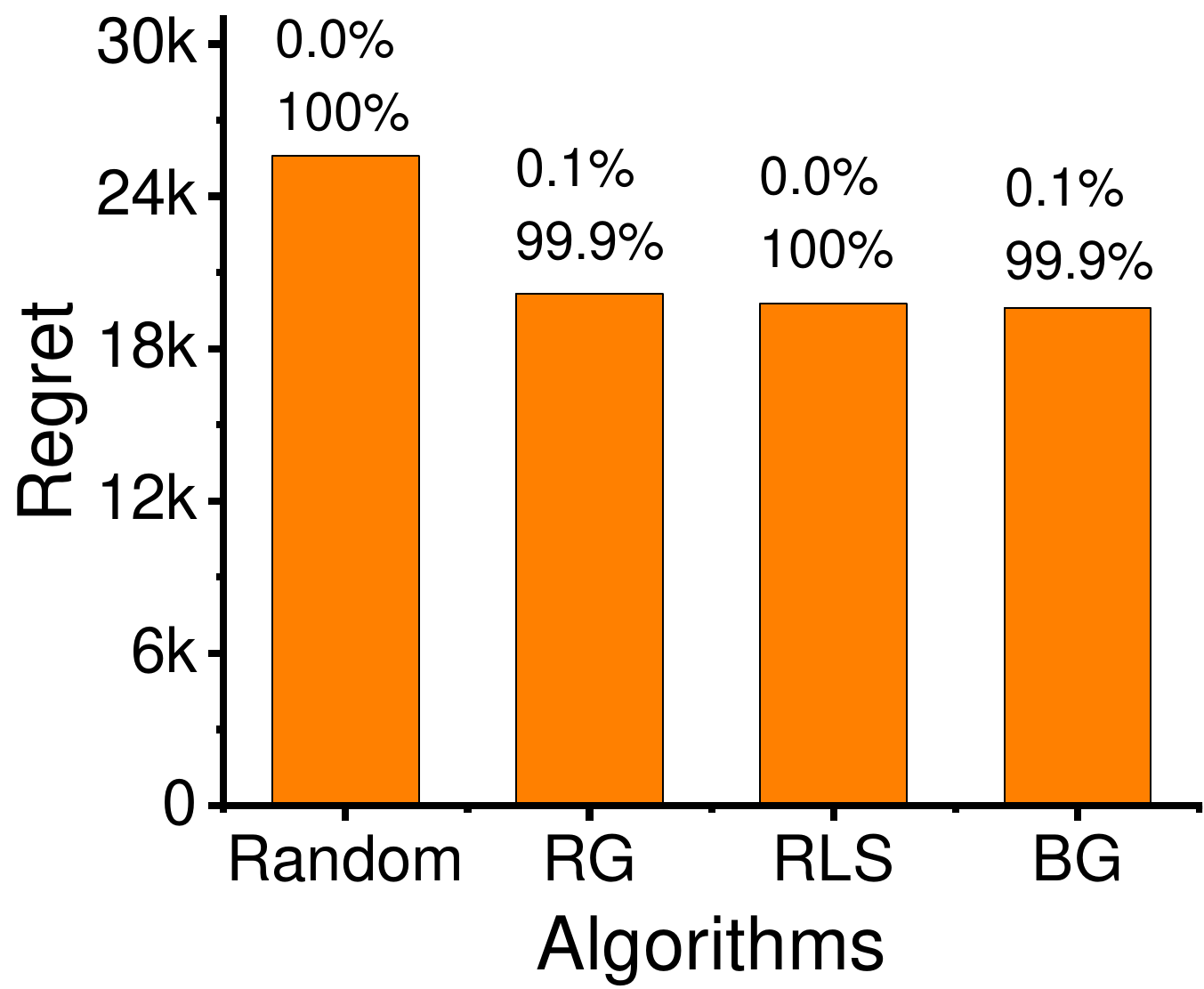}\\
\tiny{(a) $\alpha = 40\%$} &  \tiny{(b) $\alpha = 60\%$} & \tiny{(c) $\alpha = 80\%$} & \tiny{(d) $\alpha = 100\%$} & \tiny{(e) $\alpha = 120\%$}\\
\end{tabular}
\caption{Varying Demand-supply ratio, $\alpha$ when $\beta = 5\%, |\mathcal{A}| = 20$ (LA Dataset)}
\label{Fig:Plot1LA}
\end{figure*}

\paragraph{\textbf{Case 2: Low $\alpha$, High $\beta$ of Figure \ref{Fig:Plot4-5NYC} (a,b,c,f,g,h)}} Corresponding to Case 2, we have $\alpha \leq 80\%$ and $\beta \geq 10\%$. This denotes a situation in which global influence demand remains low, while individual demand is higher. The influence provider has a small number of advertisers, each with higher demand. We have two main observations. \textbf{First}, as global demand remains low and individual demand increases, the overall regret across all algorithms decreases. With the increase of $\beta$, the number of advertisers decreases, but the individual demand of the advertisers increases. That is, the supply of influence is close to the demand of the advertisers, and consequently, the excessive regret drops. \textbf{Second}, with higher individual demand, the advertisers could deploy more slots, and the `RLS' can explore a larger neighborhood search space by allocating slots to the advertisers.

\paragraph{\textbf{Case 3: High $\alpha$, Low $\beta$ of Figure \ref{Fig:Plot1-3NYC} (d,e,i,j,n,o)}} Corresponding to Case 3, we have $\alpha \geq 100\%$ and $\beta \leq 5\%$. This denotes a situation where the global demand is high, but individual influence demand is low. The influence provider has a large number of advertisers with small demand. We have two main observations. \textbf{First}, with very high global demand ($ \alpha \geq 100\%$), none of the algorithms satisfy all the advertisers. Consider the situation where $\alpha \geq 100\%$ and $\beta = 1\%$. The overall regret consists of excessive regret, even though not all advertisers are fully satisfied. Now, with increasing $\beta$, the excessive regret drops and only unsatisfied regret remains.\textbf{ Second}, when $\alpha = 100\%$, the excessive regret is less, and the unsatisfied regret is more. The influence provider cannot satisfy all advertisers due to the tag-specific influence demands of those advertisers. Although the influence provider has many slots remaining, he cannot accommodate them for the advertisers. One point needs to be noted that the `RLS' methods satisfy almost an equal number of advertisers as the `BG'. Whereas the `RG' and `Random' satisfy a smaller number of advertisers than the `RLS'.

%%%%%%%%%%%%%%%%%%%%%%%%%%%%%   alpha Vs. Time NYC %%%%%%%%%%%%%%%%%%%%%%%%%%%%
\begin{figure*}[h!]
\centering
\setlength{\tabcolsep}{0.1pt} % tighter spacing between columns
\renewcommand{\arraystretch}{0.9} % tighter spacing between rows
\begin{tabular}{ccccc}
\includegraphics[scale=0.146]{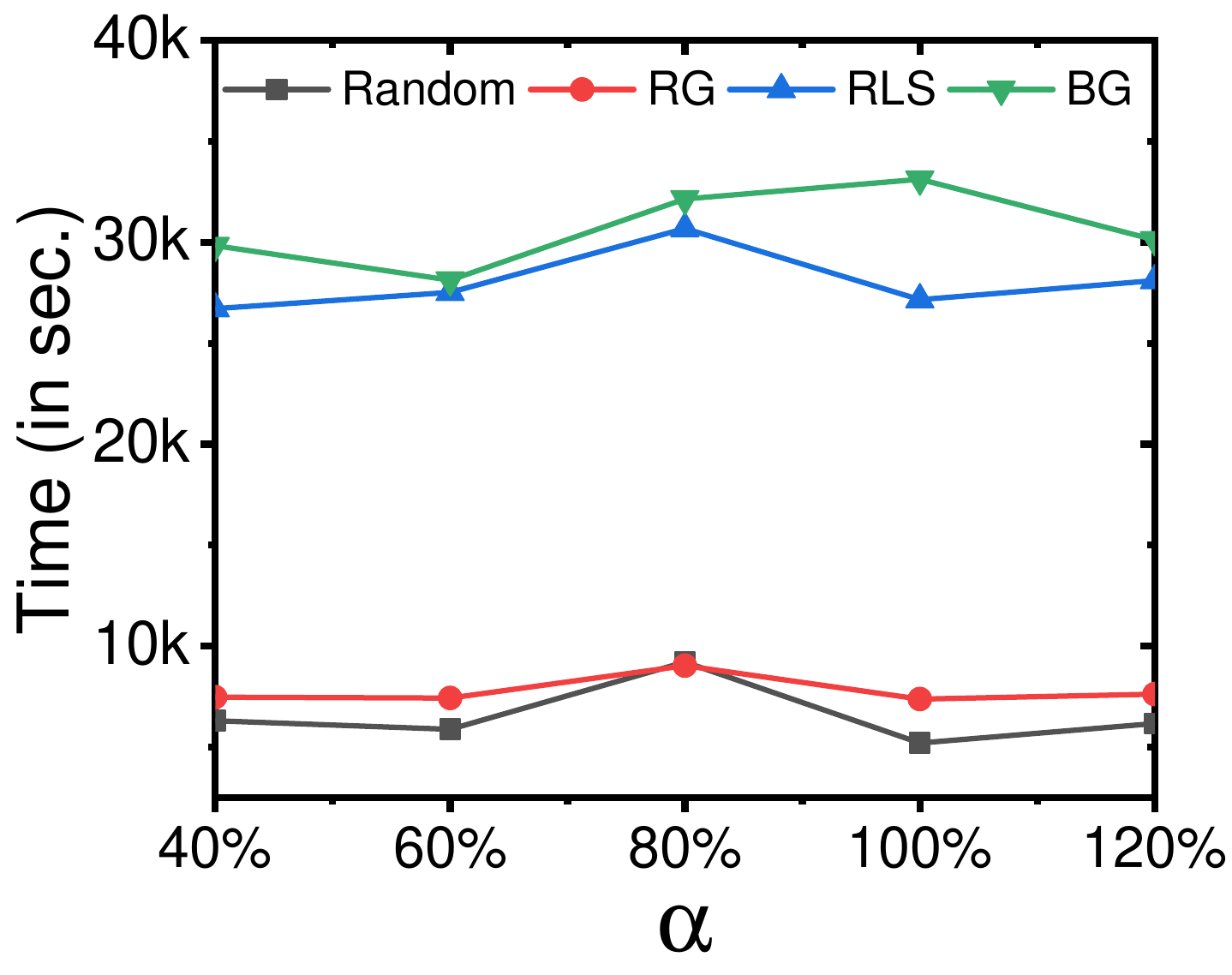} & 
\includegraphics[scale=0.146]{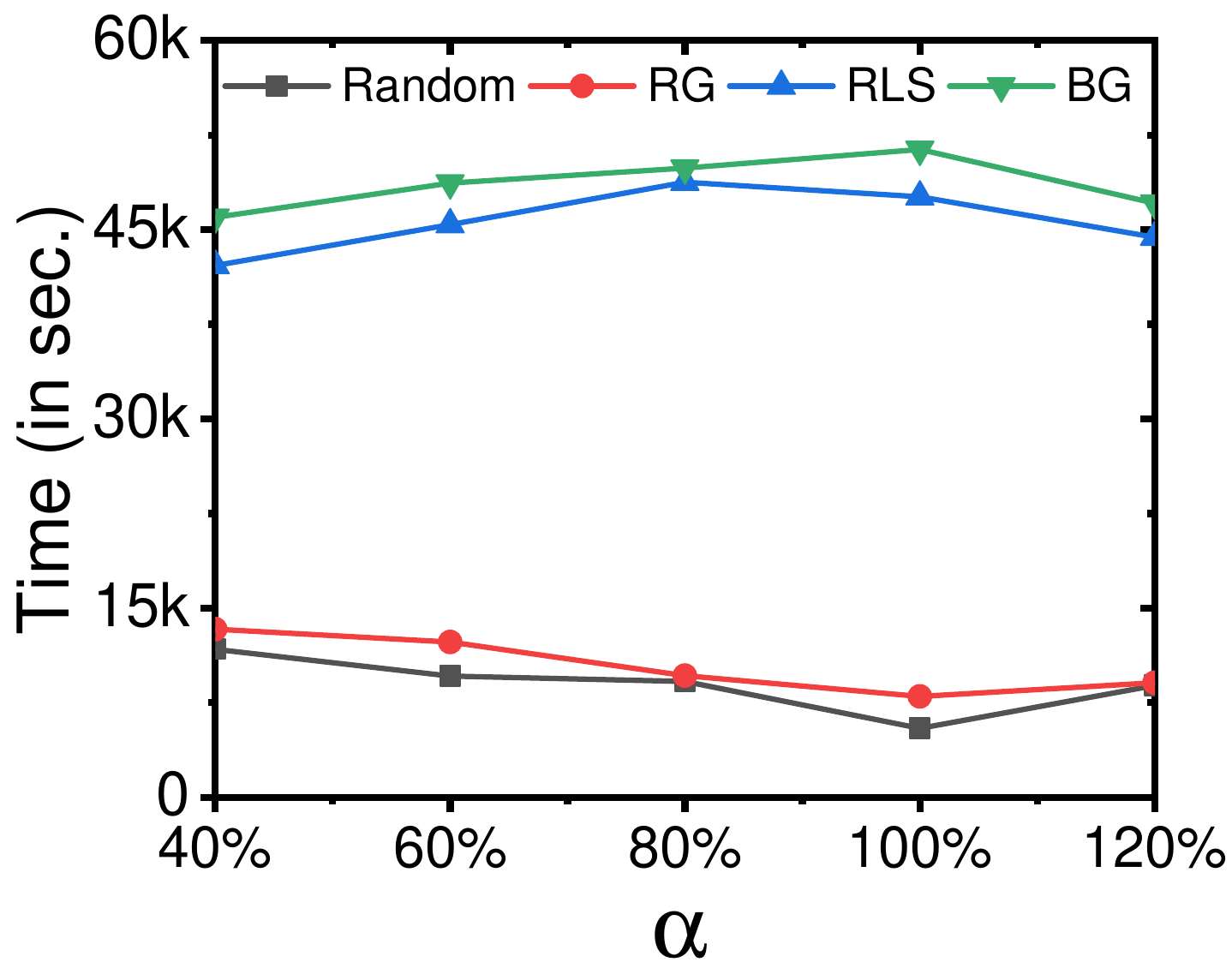} & 
\includegraphics[scale=0.146]{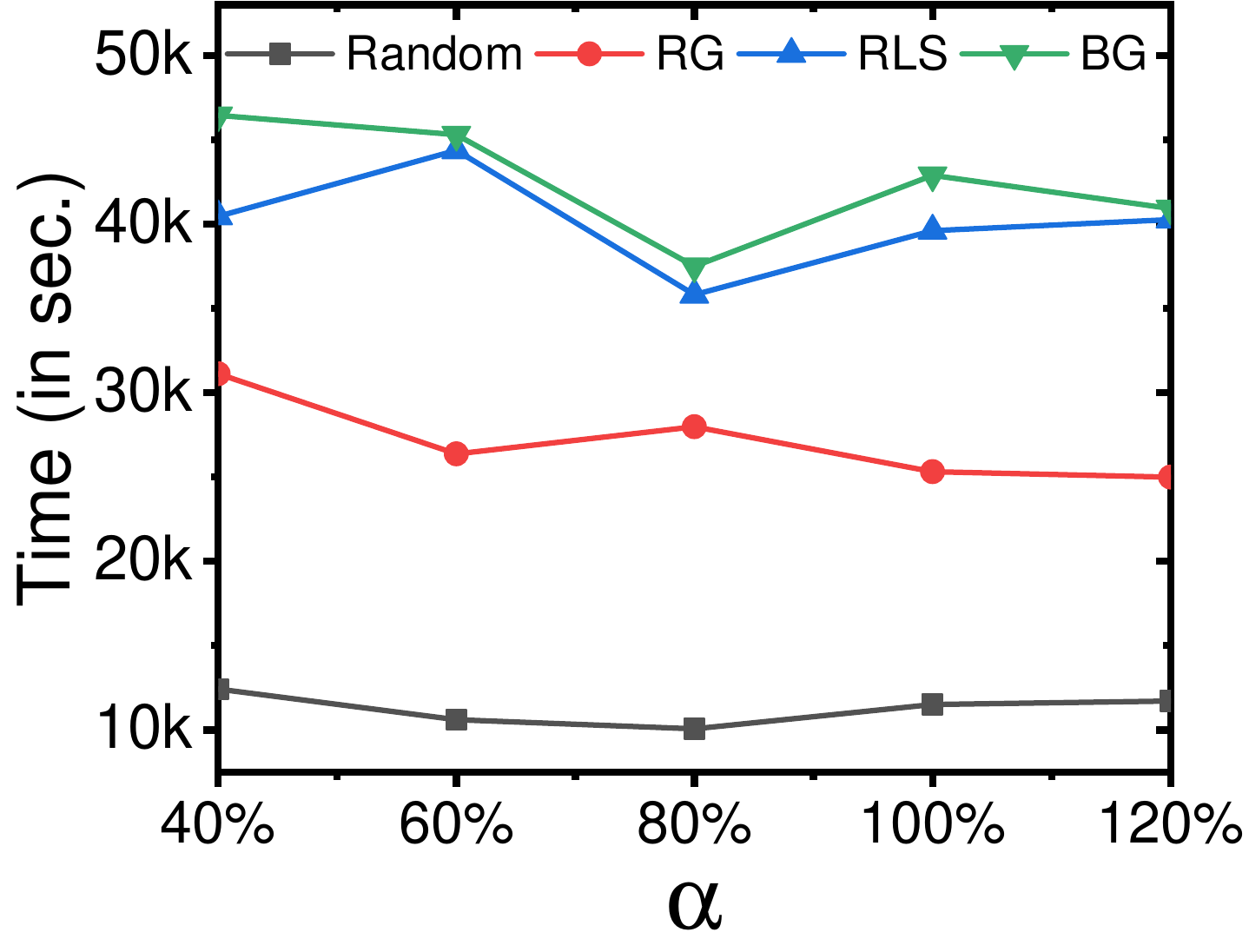} &
\includegraphics[scale=0.146]{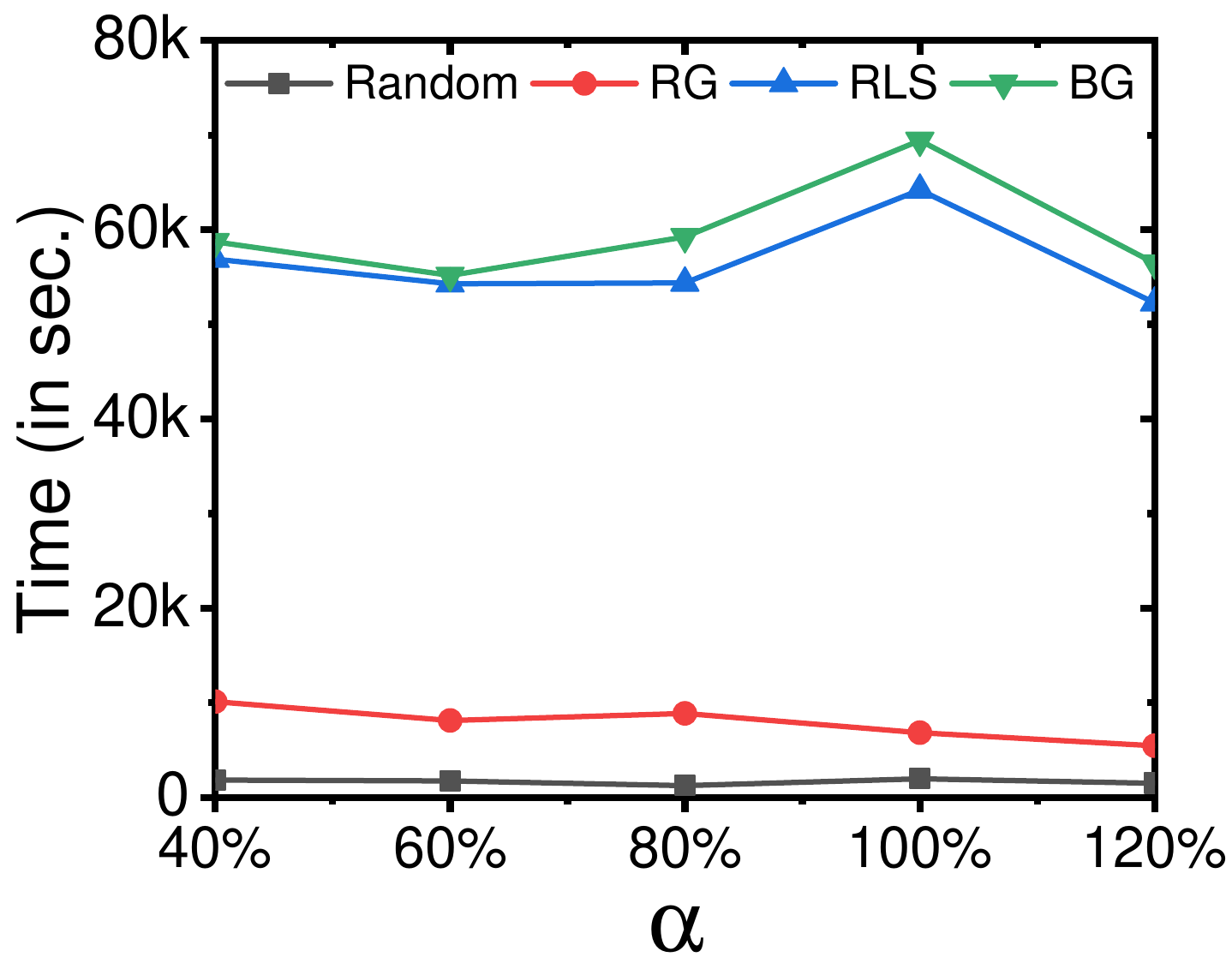} &
\includegraphics[scale=0.146]{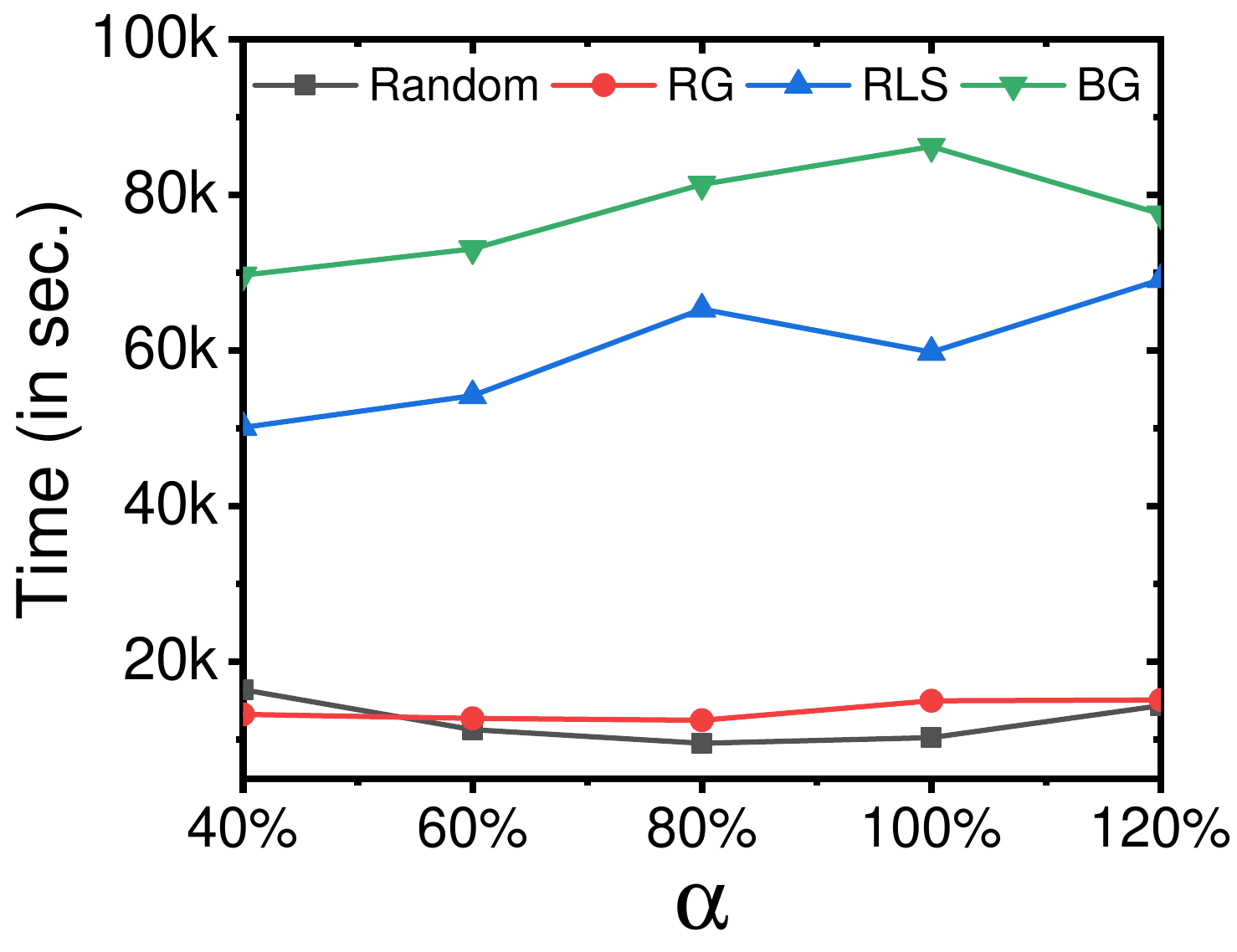}\\
\tiny{(a) $\beta = 1\%, |\mathcal{A}|= 100$} &  \tiny{(b) $\beta = 2\%, |\mathcal{A}|= 50$} & \tiny{(c) $\beta = 5\%, |\mathcal{A}|= 20$} & \tiny{(d) $\beta = 10\%, |\mathcal{A}|= 10$} & \tiny{(e) $\beta = 20\%, |\mathcal{A}|= 5$}\\
\includegraphics[scale=0.146]{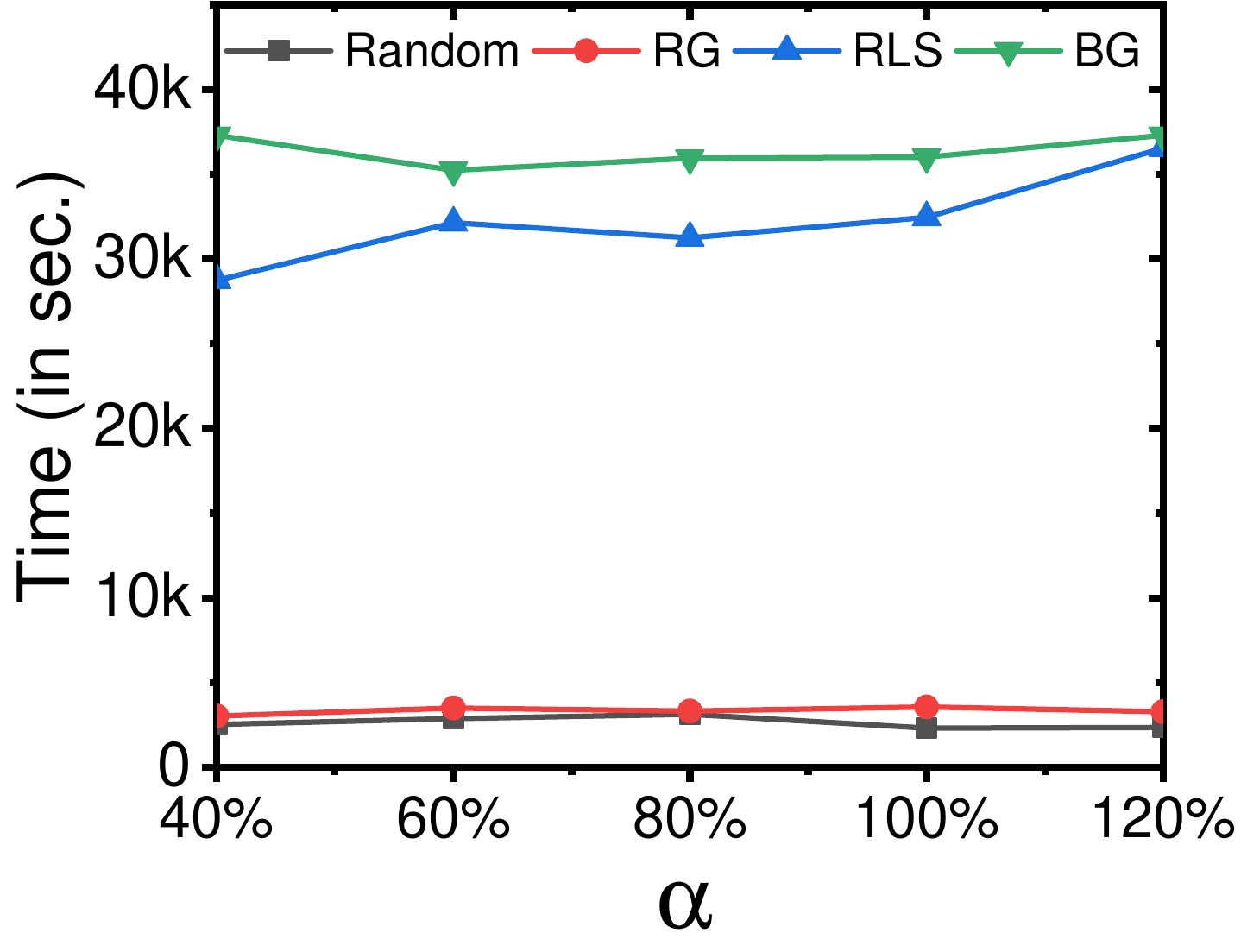} & 
\includegraphics[scale=0.146]{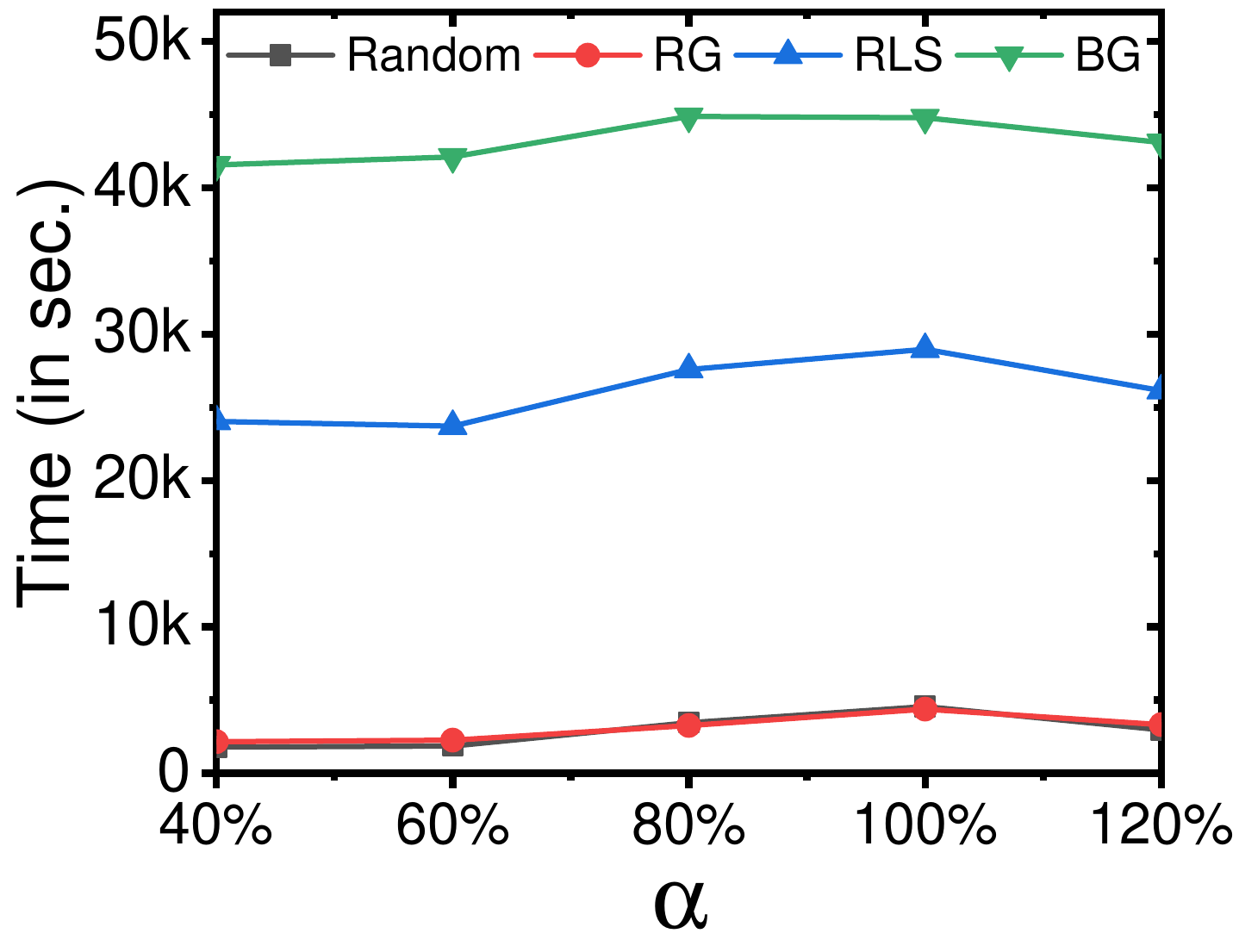} & 
\includegraphics[scale=0.146]{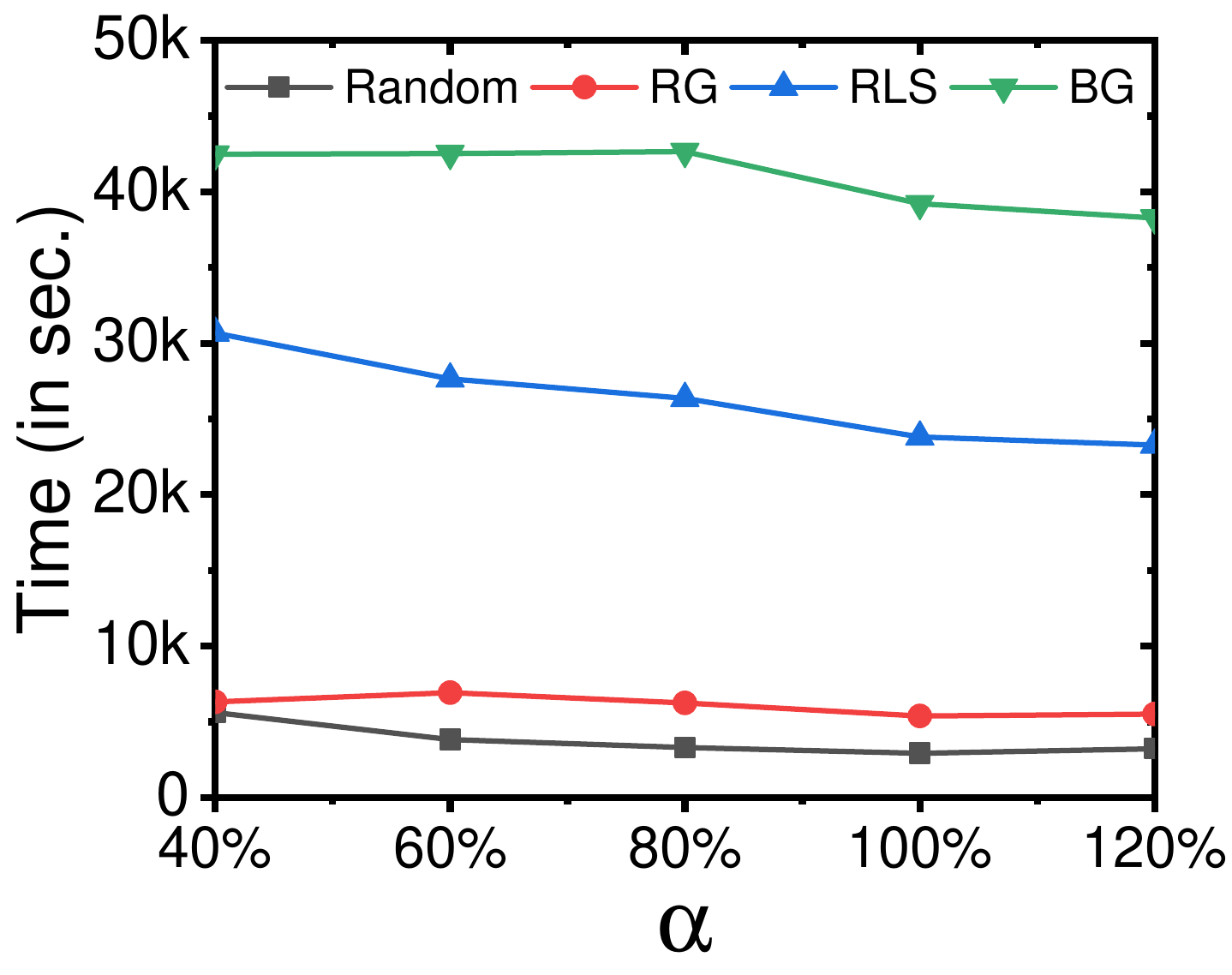} &
\includegraphics[scale=0.146]{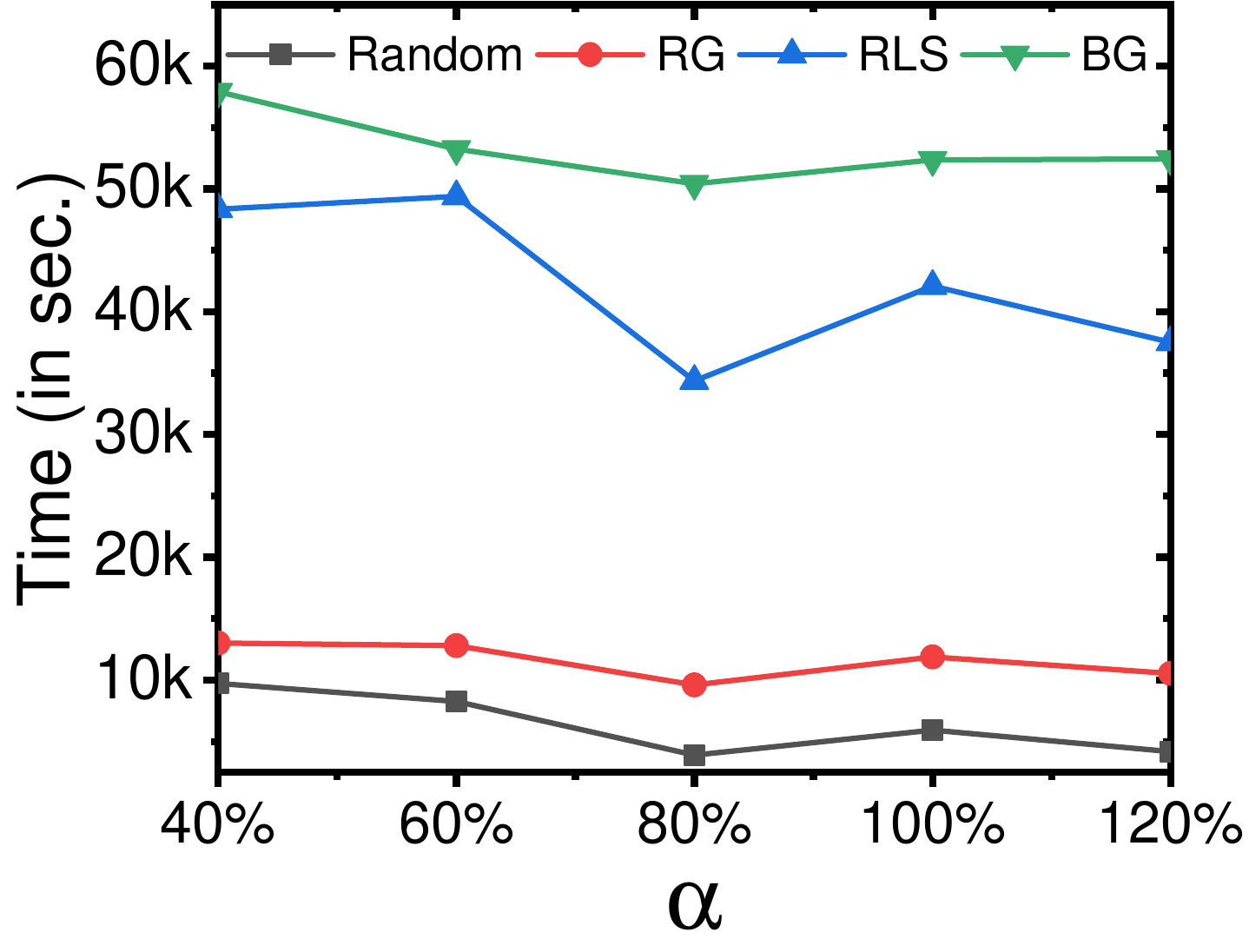} &
\includegraphics[scale=0.146]{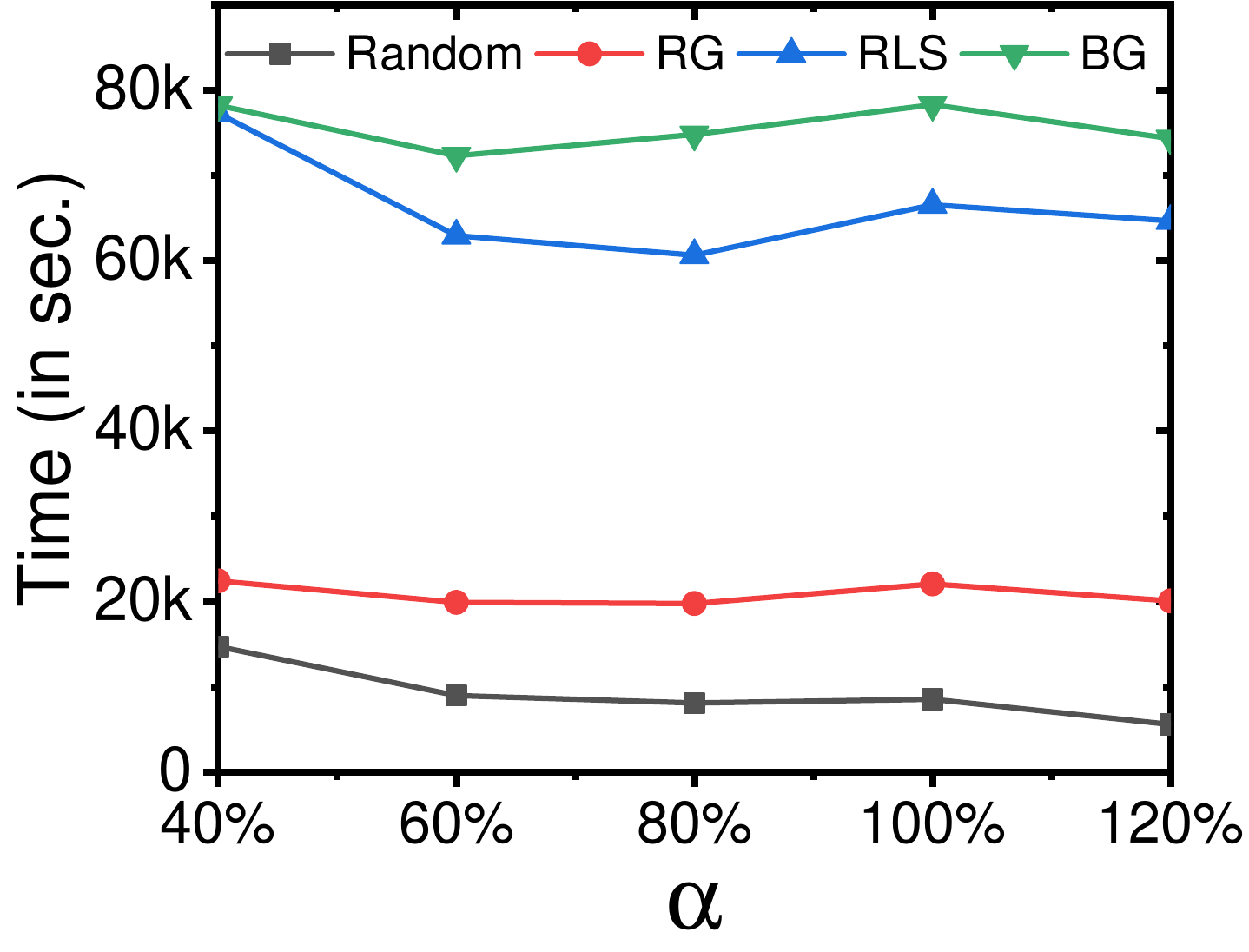}\\
\tiny{(f) $\beta = 1\%, |\mathcal{A}|= 100$} &  \tiny{(g) $\beta = 2\%, |\mathcal{A}|= 50$} & \tiny{(h) $\beta = 5\%, |\mathcal{A}|= 20$} & \tiny{(i) $\beta = 10\%, |\mathcal{A}|= 10$} & \tiny{(j) $\beta = 20\%, |\mathcal{A}|= 5$}\\
\end{tabular}
\caption{Efficiency Test $a,b,c,d,e$ (NYC Dataset), $f,g,h,i,j$ (LA Dataset)}
\label{Fig:PlotTime}
\end{figure*}

%%%%%%%%%%%%%%%%%%%%%%%%%  Varying Delta LA %%%%%%%%%%%%%%%%%%%%%%
\begin{figure*}[h!]
\centering
    \begin{tabular}{lclc}
       Excessive Regret & \includegraphics[width=0.11\linewidth]{Unsatisfied.png} & Unsatisfied Regret & \includegraphics[width=0.11\linewidth]{Excessive.png} \\
    \end{tabular}
\setlength{\tabcolsep}{0.1pt} % tighter spacing between columns
\renewcommand{\arraystretch}{0.9} % tighter spacing between rows
\begin{tabular}{ccccc}
\includegraphics[scale=0.156]{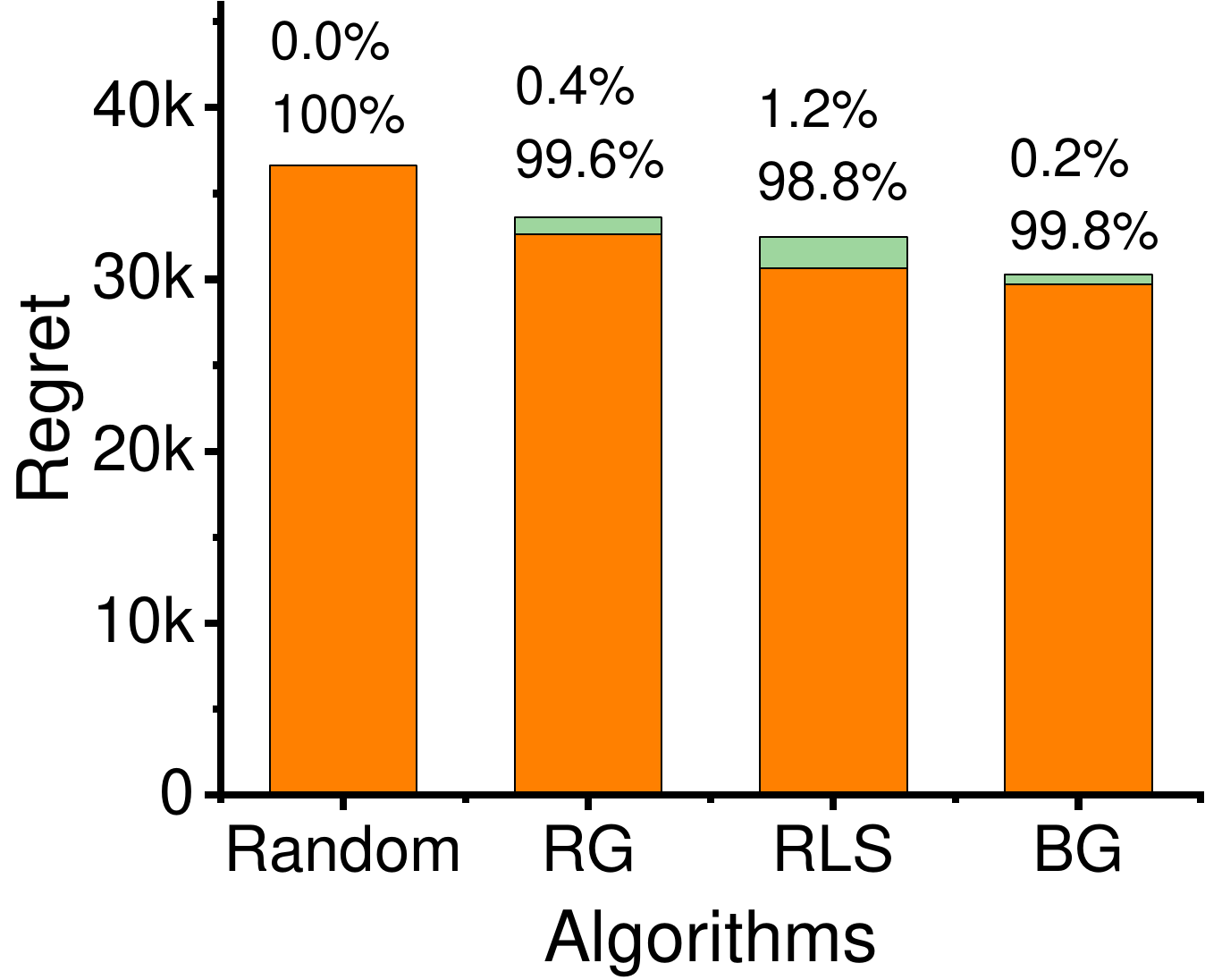} & \includegraphics[scale=0.156]{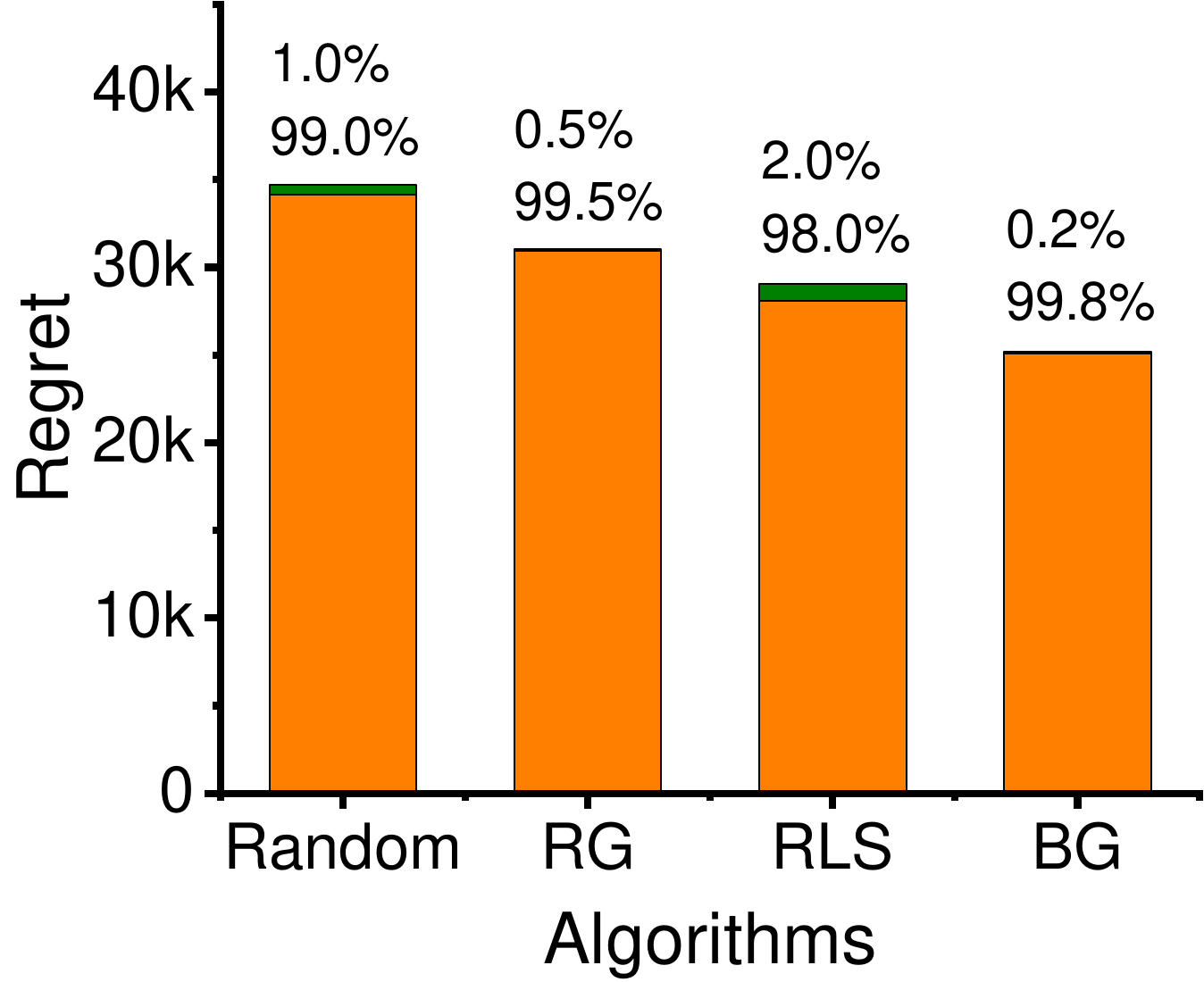} & \includegraphics[scale=0.156]{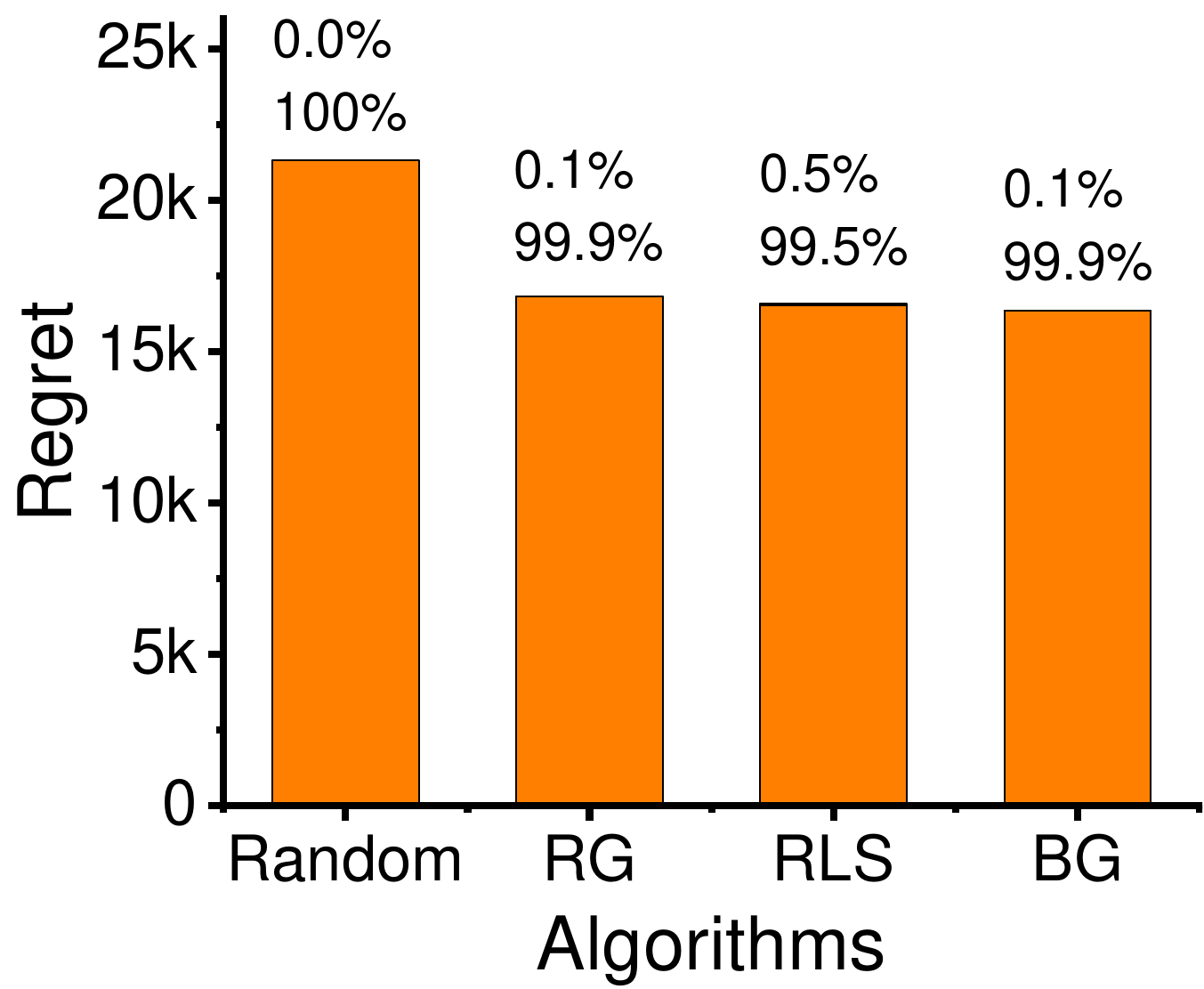} &
\includegraphics[scale=0.156]{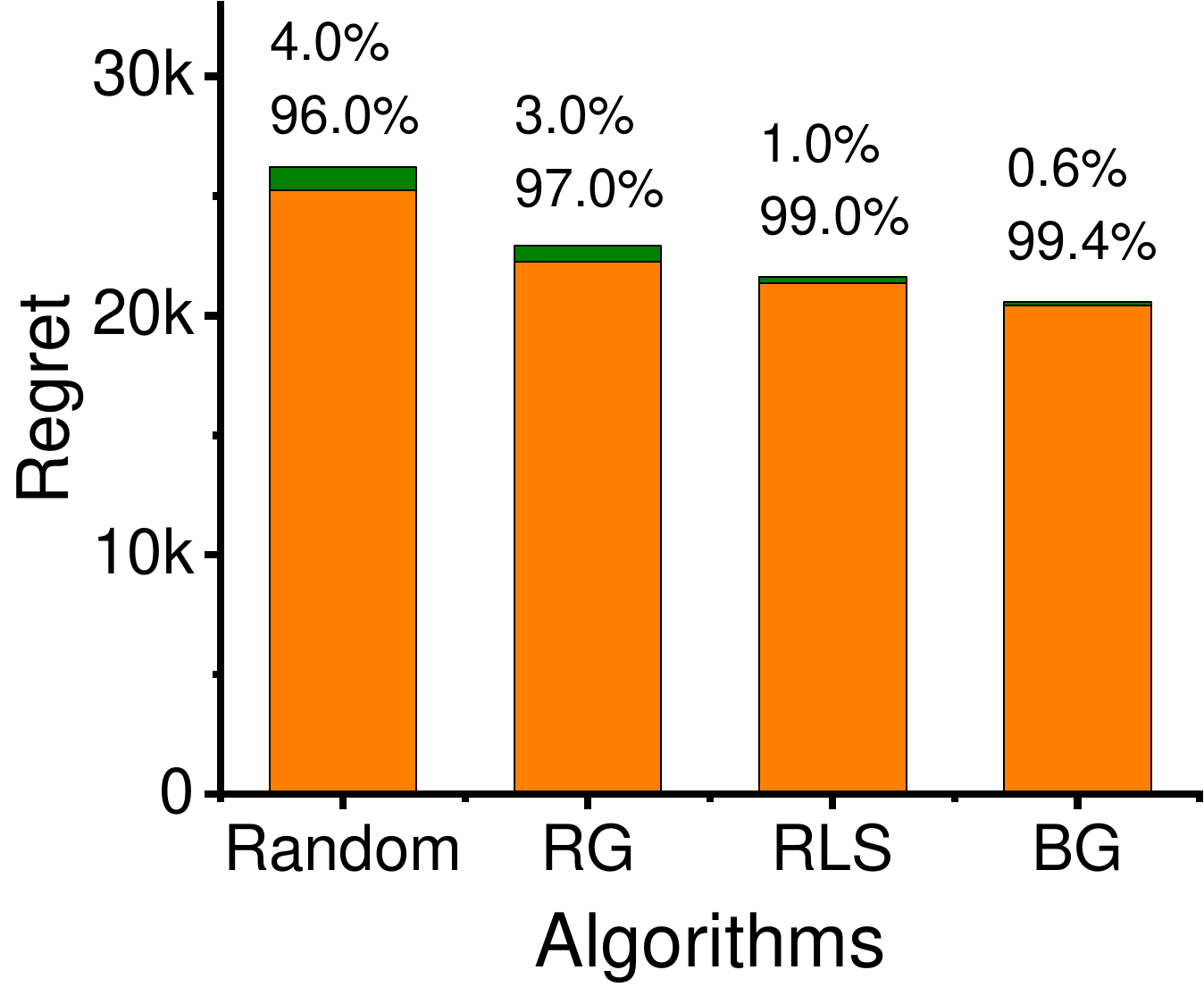} &
\includegraphics[scale=0.156]{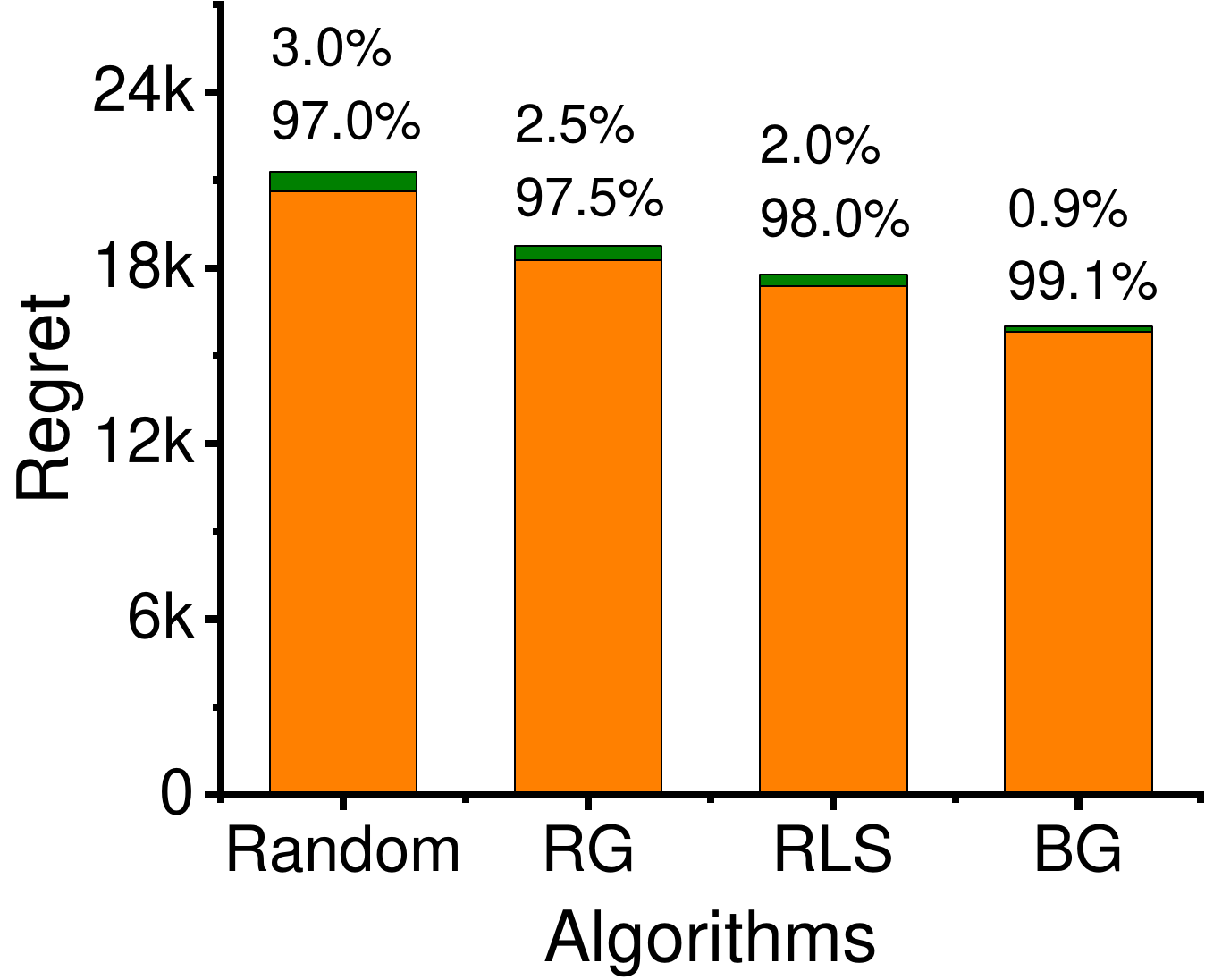}\\
\tiny{(a) $\delta = 0$} &  \tiny{(b) $\delta = 0.25$} & \tiny{(c) $\delta = 0.5$} & \tiny{(d) $\delta = 0.75\%$} & \tiny{(e) $\delta = 1$}\\
\end{tabular}
\caption{ Regret on varying  penalty ratio $(\delta)$ when $\alpha = 100\%, \beta = 5\%, |\mathcal{A}| = 20$ (LA Dataset)}
\label{Fig:Plotdelta}
\end{figure*}

\paragraph{\textbf{Case 4: High $\alpha$, High $\beta$ of Figure \ref{Fig:Plot4-5NYC} (d,e,i,j)}} Corresponding to Case 4, we have $\alpha \geq 100\%$ and $\beta \geq 10\%$. This refers to a situation where both global and individual demands are very high. The influence provider has a small number of advertisers, each with high individual demand. We have two main observations. \textbf{First}, with large $\alpha$ and $\beta$ value all the algorithms suffers from higher unsatisfied regret. The `RG', `RLS', and `BG' suffer from almost identical overall regret, whereas the `Random' suffers from higher overall regret, in which most of the portion comes from unsatisfied regret. \textbf{Second}, the total demand is very high compared to supply when $\alpha = 120\%$, i.e., the influence provider does not have sufficient slots to satisfy the demand of the advertisers.

\subsubsection{\textbf{Experiments over LA Dataset}}
Figure \ref{Fig:Plot1LA} shows the experimental results over the LA dataset. We have observed that the LA dataset has more billboard slots than the NYC dataset. The slots in the LA are low-influence, and the average influence is lower. Compared with NYC, slot influences in LA are more uniformly distributed. Moreover, the overlap of influence among different slots is smaller. We have two main observations.\textbf{ First}, the results for the LA dataset look quite similar to those for NYC, but the proportion of excessive regret is smaller across all algorithms. This is mainly because slot influence in LA is lower, and there is less overlap among slots, which makes it easier for the influence provider to place them more efficiently and accurately. \textbf{Second}, since LA has more slots and a lower average influence per billboard, `RLS' can explore more potential allocation with them. As a result, similar to the trend observed in NYC, `RLS' effectively reduces excessive influence in LA.

\par \textbf{Revisit RQ1 and RQ2.} Based on our observations from these cases, we answer the two research questions that arose earlier.
(1) When global demand is low, total regret comes from excessive and unsatisfied influence. The excessive influence is also an essential part of the overall regret. In this case, advertiser selection becomes crucial. When the total demand is close to the supply, the excessive influence decreases, and the unsatisfied influence increases. When total demand exceeds the influence provider's supply, total regret is high, and a major portion of it stems from unsatisfied influence. (2) When global demand is very high or exceeds supply, regret is dominated by unsatisfied penalties. In this setting, having many medium-demand advertisers works best, as it provides greater flexibility in allocating billboards, while the penalty for missing a single large advertiser remains relatively small. Hence, the large number of advertisers with low influence demand will be beneficial for the influence provider.

\subsubsection{\textbf{Efficiency Study}}
In outdoor advertising, efficiency is important for an influence provider with more than a thousand billboard slots in a city and several advertisers approaching daily for influence. Hence, we conduct an efficiency study under different values of $\alpha$ and $\beta$. The results are reported in Figure \ref{Fig:PlotTime} for the NYC and LA datasets. We have three main observations. \textbf{First}, among the proposed approaches, `BG' has a larger runtime than `RG' and `RLS'. This happens because, at every iteration to allocate a slot, `BG' computes excessive regret computation. On the other hand, the `RG' and `RLS' use randomization during regret calculation and take less time. The `Random' approach has the minimum runtime due to uniformly random allocation. \textbf{Second}, with increasing $\alpha$, the computational time for all the algorithms increases due to the demand of the advertisers increasing, and the influence provider needs to allocate more slots.\textbf{ Third}, as $\beta$ increases, the individual demand of advertisers increases. Hence, runtime also increases for all the algorithms, and `RLS' has to explore more search during allocation.

\subsubsection{\textbf{Other Parameter Study}}\label{subsec:otherparameter}
We discuss the additional parameters used in our experiments.\textbf{ First}, the distance parameter $\gamma$ determines the range in which a billboard slot can influence the trajectories. With increasing $\gamma$, the influence of a slot increases as slots can influence a larger number of trajectories. In our experiment, we set the default value of $\gamma$ to $100$ meters. \textbf{Second}, the sampling parameter $\epsilon$ decides the size of the random sample set. When $\epsilon$ increases, the sample set size decreases, and the run time of `RG' and `RLS' also decreases. At the same time, the solution's quality also degrades. If the $\epsilon$ value decreases, the quality of the solution becomes better because the sample set size becomes larger. We set $\epsilon = 0.01$ as the default for all experiments. \textbf{Third}, the penalty ratio $\delta$ controls the fraction of the penalty on the payment when the advertiser is not satisfied. When $\delta$ is very small, the influence provider suffers greater regret. With increasing $\delta$ the regret for all the algorithms decreases as shown in Figure \ref{Fig:Plotdelta}. We set $\delta = 0.5$ as the default setting for all experiments. 

\section{Concluding Remarks} \label{Sec:Conclusion}
In this work, we studied the Tag-Specific Regret Minimization in Outdoor Advertising (TRMOA) problem from the perspective of the influence provider. Unlike traditional influence maximization approaches, our model captures both unsatisfied and excessive regret under tag-specific influence demands. We proved that the problem is NP-hard and proposed three scalable heuristic algorithms, each supported by adaptive tag selection. Experiments on real-world NYC and LA datasets show that our methods significantly reduce total regret while remaining computationally efficient. The results also provide practical information on how demand–supply dynamics and advertiser demand distribution affect regret behavior. This work lays the foundation for regret-aware allocation strategies in real-world outdoor advertising systems.

\bibliographystyle{IEEEtran}
\bibliography{Paper.bib}
\end{document}